\documentclass[11pt]{article} 
\usepackage[
  margin=3cm,
  includefoot,
  footskip=30pt,
]{geometry}
\RequirePackage[OT1]{fontenc}
\RequirePackage{amsthm,amsmath}
\RequirePackage[numbers]{natbib}
\RequirePackage[colorlinks,citecolor=blue,urlcolor=blue]{hyperref}

\usepackage{url,hyperref}  
\usepackage{graphicx}  

\usepackage{times}
\usepackage{booktabs}       
\usepackage{amsfonts}       
\usepackage{nicefrac}       
\usepackage{microtype}      
\usepackage{algorithm}
\usepackage{bbold}
\usepackage{tikz}
\usetikzlibrary{decorations.pathreplacing}
\usetikzlibrary{patterns}
\usepackage{amsmath}
\usepackage{algorithmic}
\usepackage{verbatim}
\usepackage{subcaption}
\usepackage{graphicx}
\usepackage{multirow}
\usepackage{amssymb}
\usepackage{color}
\usepackage{caption}
\usepackage{enumitem}
\usepackage{hyphenat}
\usepackage{array}
\usepackage{booktabs}
\usepackage{wrapfig}
\usepackage{relsize}
\usepackage{tikz}

\usepackage{afterpage}
\usetikzlibrary{arrows}

\DeclareGraphicsExtensions{.eps}

\newcommand\nc\newcommand
\nc\bfa{{\boldsymbol a}}\nc\bfA{{\boldsymbol A}}\nc\cA{{\mathcal A}}
\nc\bfb{{\boldsymbol b}}\nc\bfB{{\boldsymbol B}}\nc\cB{{\mathcal B}}
\nc\bfc{{\boldsymbol c}}\nc\bfC{{\boldsymbol C}}\nc\cC{{\mathcal C}}
\nc\sC{{\mathscr C}}
\nc\bfd{{\boldsymbol d}}\nc\bfD{{\boldsymbol D}}\nc\cD{{\mathcal D}}
\nc\bfe{{\boldsymbol e}}\nc\bfE{{\boldsymbol E}}\nc\cE{{\mathcal E}}
\nc\bff{{\boldsymbol f}}\nc\bfF{{\boldsymbol F}}\nc\cF{{\mathcal F}}
\nc\bfg{{\boldsymbol g}}\nc\bfG{{\boldsymbol G}}\nc\cG{{\mathcal G}}
\nc\bfh{{\boldsymbol h}}\nc\bfH{{\boldsymbol H}}\nc\cH{{\mathcal H}}
\nc\bfi{{\boldsymbol i}}\nc\bfI{{\boldsymbol I}}\nc\cI{{\mathcal I}}
\nc\bfj{{\boldsymbol j}}\nc\bfJ{{\boldsymbol J}}\nc\cJ{{\mathcal J}}
\nc\bfk{{\boldsymbol k}}\nc\bfK{{\boldsymbol K}}\nc\cK{{\mathcal K}}
\nc\bfl{{\boldsymbol l}}\nc\bfL{{\boldsymbol L}}\nc\cL{{\mathcal L}}
\nc\bfm{{\boldsymbol m}}\nc\bfM{{\boldsymbol M}}\nc\sM{{\mathscr M}}\nc\cM{{\mathcal M}}
\nc\bfn{{\boldsymbol n}}\nc\bfN{{\boldsymbol N}}\nc\cN{{\mathcal N}}
\nc\bfo{{\boldsymbol o}}\nc\bfO{{\boldsymbol O}}\nc\cO{{\mathcal O}}
\nc\bfp{{\boldsymbol p}}\nc\bfP{{\boldsymbol P}}\nc\cP{{\mathcal P}}
\nc\bfq{{\boldsymbol q}}\nc\bfQ{{\boldsymbol Q}}\nc\cQ{{\mathcal Q}}
\nc\bfr{{\boldsymbol r}}\nc\bfR{{\boldsymbol R}}\nc\cR{{\mathcal R}}
\nc\bfs{{\boldsymbol s}}\nc\bfS{{\boldsymbol S}}\nc\cS{{\mathcal S}}
\nc\bft{{\boldsymbol t}}\nc\bfT{{\boldsymbol T}}\nc\cT{{\mathcal T}}
\nc\bfu{{\boldsymbol u}}\nc\bfU{{\boldsymbol U}}\nc\cU{{\mathcal U}}
\nc\bfv{{\boldsymbol v}}\nc\bfV{{\boldsymbol V}}\nc\cV{{\mathcal V}}
\nc\bfw{{\boldsymbol w}}\nc\bfW{{\boldsymbol W}}\nc\cW{{\mathcal W}}
\nc\bfx{{\boldsymbol x}}\nc\bfX{{\boldsymbol X}}\nc\cX{{\mathcal X}}
\nc\bfy{{\boldsymbol y}}\nc\bfY{{\boldsymbol Y}}\nc\cY{{\mathcal Y}}
\nc\bfz{{\boldsymbol z}}\nc\bfZ{{\boldsymbol Z}}\nc\cZ{{\mathcal Z}}

\nc\diff{{\mathrm d}}
\nc\e{{\mathrm e}}
\nc\calC{{\mathcal C}}

\newcommand{\remove}[1]{}
\newcommand{\correct}[1]{\textcolor{black}{#1}}

\newcommand{\avg}{{\mathbb E}}

\newcommand{\dist}{d_{L}}

\numberwithin{equation}{section}
\theoremstyle{plain}
\newtheorem{theorem}{Theorem}
\newtheorem{lemma}[theorem]{Lemma}
\newtheorem*{theorem*}{Theorem}
\newtheorem*{lemma*}{Lemma}
\newtheorem{corollary}{Corollary}
\newtheorem*{corollary*}{Corollary}
\newtheorem{observation}[theorem]{Observation}
\newtheorem{definition}{Definition}[section]

\newtheorem{remark}{Remark}

\def\DEBUG{true}

\ifdefined\DEBUG

  \def\rem#1{{\marginpar{\raggedright\scriptsize #1}}}
  \newcommand{\barnr}[1]{\rem{\textcolor{red}{$\bullet$ #1}}}
  \newcommand{\aryar}[1]{\rem{\textcolor{green}{$\bullet$ #1}}}
\else

  \newcommand{\barnr}[1]{}
  \newcommand{\aryar}[1]{}

\fi

\newcommand\reals{{\mathbb R}}

\newcommand\integers{{\mathbb Z}}

\newcommand{\norm}[1]{||#1||}

\allowdisplaybreaks

\title{Connectivity in Random Annulus Graphs and\\ the Geometric Block Model\thanks{Sainyam Galhotra and Barna Saha are supported in part by NSF CAREER 1652303, a Google faculty award and an Alfred P. Sloan fellowship. Arya Mazumdar and Soumyabrata Pal are supported in part by NSF Awards 1642658 and 1642550.}}
\usepackage{authblk}
\author[1]{Sainyam Galhotra}
\author[1]{Arya Mazumdar}
\author[1]{Soumyabrata Pal}
\author[2]{Barna Saha}
\affil[1]{College of Information and Computer Sciences,\protect\\
  University of Massachusetts Amherst,\protect\\
  Amherst, MA 01003\protect\\
  \texttt{\{sainyam,arya,spal\}@cs.umass.edu}}
\affil[2]{University of California Berkeley, USA\protect\\
        \texttt{barnas@berkeley.edu}}

\begin{document}
\maketitle

\begin{abstract}
We provide new connectivity results for  {\em vertex-random graphs} or {\em random annulus graphs} which are significant generalizations of random geometric graphs. Random geometric graphs (RGG) are one of the most basic models of random graphs for spatial networks proposed by Gilbert in 1961, shortly after the introduction of the Erd\H{o}s-R\'{en}yi random graphs. They resemble social networks in many ways (e.g. by spontaneously creating cluster of nodes with high modularity). The connectivity properties of RGG have been studied since its introduction, and analyzing them has been significantly harder than their  Erd\H{o}s-R\'{en}yi counterparts due to correlated edge formation.

An Erd\H{o}s-R\'{en}yi random graph $G(n,p), n \in \integers_+, p \in [0,1]$ has $n$ vertices, and each pair of vertices form an edge with probability $p$. This is the simplest model of random graphs where the randomness lies in the edges. It is natural to define a similar (in simplicity) model of random graphs where the randomness lies in the vertices. Consider a vertex-random graph $G(n, [r_1, r_2]), 0 \le r_1 <r_2\le 1$  with $n$ vertices.  Each vertex of the graph is assigned a real number  in $[0,1]$ randomly and uniformly. There is an edge between two vertices if the difference between the corresponding two random numbers is between $r_1$ and $r_2$. For the special case of $r_1=0$, this corresponds to random geometric graph in one dimension. We can extend this model to higher dimensions where each vertex is associated with a uniform random vector on a $t$-dimensional unit sphere and an edge gets formed if and only if the Euclidean (or geodesic) distance between the two vertices is between $r_1$ and $r_2$. Again, when $r_1=0$, this reduces to high-dimensional RGGs. We call such graphs {\em random annulus graphs} (RAG). In this paper we study the connectivity properties of such graphs, providing both necessary and sufficient conditions. We show a surprising {\em long edge phenomena} for vertex-random graphs: the minimum gap for connectivity between $r_1$ and $r_2$ is significantly less when $r_1 >0$ vs when $r_1=0$ (RGG). We then extend the connectivity results to high dimensions. 

Our next contribution is in using the connectivity of random annulus graphs to provide necessary and sufficient conditions for efficient recovery of communities for {\em the geometric block model} (GBM). The GBM is a probabilistic model for community detection defined over an RGG in a similar spirit as the popular {\em stochastic block model}, which is defined over an Erd\H{o}s-R\'{en}yi random graph. The geometric block model inherits the transitivity properties of RGGs and thus models communities better than a stochastic block model.  However, analyzing them requires fresh perspectives as all prior tools fail due to correlation in edge formation. We provide a simple and efficient algorithm that can recover communities in GBM exactly with high probability in the regime of connectivity.

\end{abstract}

\section{Introduction}
Models of random graphs are ubiquitous with Erd\H{o}s-R\'{en}yi graphs \citep{erdos1959random,gilbert1959random} at the forefront. Studies of the properties of random graphs have led to many fundamental theoretical observations as well as many engineering applications. 
In an Erd\H{o}s-R\'{en}yi graph $G(n,p), n \in \integers_+, p \in [0,1]$, the randomness lies in how the edges are chosen: each possible pair of vertices forms an edge independently with probability $p$. It is also possible to consider models of graphs where randomness lies in the vertices. 

\paragraph{Vertex Random Graphs.} Keeping up with the simplicity of the  Erd\H{o}s-R\'{en}yi model, let us define a vertex-random graph (VRG) in the following way. Given two reals $0 \le r_1\le r_2\le 1/2$, the vertex-random graph  ${\rm VRG}(n, [r_1,r_2])$ is a random graph with $n$ vertices. Each vertex $u$ is assigned a random number $X_u$ selected randomly and uniformly from $[0,1]$. Two vertices $u$ and $v$ are connected by an edge, if and only if $r_1 \le d_L(X_u,X_v) \le r_2$, where $d_L(X_u,X_v)$ can be taken to be the absolute difference $|X_u-X_v|$, however to curtail the boundary effect we define   $d_L(X_u,X_v) \equiv \min\{|X_u-X_v|, 1-|X_u-X_v|\}$ here.

This definition is by no means new. For the case of $r_1=0$, this is the random geometric graphs (RGG) in one dimension. 
Random Geometric graphs were defined first  by \citep{gilbert1961random} and constitute the first and simplest model of spatial networks. The definition of VRG  has been previously mentioned in \citep{dettmann2016random}. The interval $[r_1,r_2]$ is called the connectivity interval in VRG.
Random geometric graphs have several desirable properties that model real human social networks, such as vertices with high modularity and the degree associativity property (high degree nodes tend to connect). This has led RGGs to be used as models of disease outbreak in social network \citep{eubank2004modelling} and flow of opinions \citep{zhang2014opinion}. RGGs are an excellent model for wireless (ad-hoc) communication networks \citep{dettmann2016random,haenggi2009stochastic}. From a more mathematical stand-point, RGGs act as a bridge between the theory of classical random graphs and that of percolation \citep{b:01,b:06}. Recent works on RGGs also include hypothesis testing between an Erd\H{o}s-R\'{en}yi graph and a random geometric graph \citep{bubecktriangle}.

Threshold properties of   Erd\H{o}s-R\'{en}yi graphs have been at the center of much theoretical interest, and in particular it is known  that many graph properties exhibit sharp phase transition phenomena \citep{friedgut1996every}. Random geometric graphs also exhibit similar threshold properties  \citep{penrose2003random}. 

Consider a ${\rm VRG}(n,[0,r])$ defined above with $r = \frac{a\log n}{n}$. It is known that ${\rm VRG}(n,[0,r])$ is connected with high probability if and only if $a > 1$\footnote{That is, ${\rm VRG}(n,[0,\frac{(1+\epsilon)\log n}{n}])$ is connected for any $\epsilon >0$. We will ignore this $\epsilon$ and just mention connectivity threshold as $\frac{\log{n}}{n}$.}. Now let us consider the graph ${\rm VRG}(n,[\frac{\delta\log n}{n},\frac{\log n}{n}]), \delta>0$. Clearly this graph has less edges than ${\rm VRG}(n,[0,\frac{\log n}{n}])$.  {\bf Is this graph still connected?}
Surprisingly, we show that the above modified graph remains connected as long as $\delta \le 0.5$. Note that, on the other hand,  ${\rm VRG}(n,[0,\frac{(1-\epsilon)\log n}{n}])$ is not connected for any $\epsilon >0$.


To elaborate, consider a  ${\rm VRG}(n,[r_1,r_2])$ when $r_1 = \frac{b\log n}{n}$ and $r_2 = \frac{a \log n}{n}$. We show that when $b >0$, the vertex-random graph  ${\rm VRG}(n,[\frac{b\log n}{n},\frac{a\log n}{n}])$ is connected with high probability if and only if $a - b > 0.5$ and $a > 1$. This means the graphs ${\rm VRG}(n, [0, \frac{0.99 \log n}{n}])$ and ${\rm VRG}(n, [\frac{0.49 \log n}{n}, \frac{0.99 \log n}{n}])$ are not connected with high probability, whereas ${\rm VRG}(n, [\frac{0.50 \log n}{n}, \frac{\log n}{n}])$ is connected. 
For a depiction of the connectivity regime for the vertex-random graph $G(n,[\frac{b\log n}{n},\frac{a\log n}{n}])$ see Figure~\ref{fig:region}.
\vspace{-0.1in}
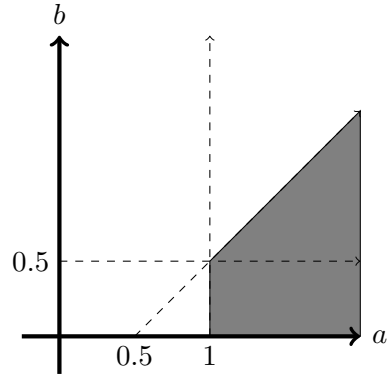
\begin{wrapfigure}{r}{0.4\textwidth}
\centering
\begin{tikzpicture}

\draw[fill=gray]  (2,1) -- (4,3) -- (4,0) -- (2,0) -- cycle;
\draw[->,ultra thick] (-0.5,0)--(4,0) node[right]{$a$};
\draw[->,ultra thick] (0,-0.5)--(0,4) node[above]{$b$};

\draw[->,dashed] (0,1)--(4,1); 
\draw[->,dashed] (2,0)--(2,4); 
\draw[->,dashed] (1,0)--(4,3); 

\node[below] at (2,0) {$1$};
\node[below] at (1,0) {$0.5$};
\node[left] at (0,1) {$0.5$};

\draw[fill=gray,transparent]  (2,1) -- (4,3) -- (4,0) -- (2,0) -- cycle;
\end{tikzpicture}
\caption{The shaded area in  the $a$-$b$ plot shows the regime where an VRG $G(n,[\frac{b \log n}{n}, \frac{a\log n}{n}])$ is connected with high probability. \label{fig:region}}
\end{wrapfigure}

Can we explain this seemingly curious shift in connectivity interval, when one goes from $b=0$ to $b >0$? Compare the VRG$(n, [\frac{0.50 \log n}{n}, \frac{\log n}{n}])$ with the ${\rm VRG}(n, [0, \frac{\log n}{n}])$. The former one can  be thought of being obtained by deleting all the `short-distance' edges from the latter. It turns out  the `long-distance' edges are  sufficient to maintain connectivity, because they can connect points over multiple hops in the graph. Another possible explanation is that connectivity threshold for {\rm VRG} is not dictated by isolated nodes as is the case in Erd\H{o}s-R\'{en}yi graphs. Thus, after the connectivity threshold has been achieved, removing certain short edges still retains connectivity. 

\paragraph{The Geometric Block Model.}
We are motivated to study the threshold phenomena of vertex-random graphs, because it appears naturally in the analysis of the geometric block model (GBM) \citep{galhotra2017geometric}.
The geometric block model is a  probabilistic generative model of communities in a variety of networks and is a spatial analogue to the popular stochastic block model (SBM) \citep{holland1983stochastic,dyer1989solution,decelle2011asymptotic,abbe2015community,abh:16,DBLP:conf/colt/HajekWX15,chin2015stochastic,mossel2015consistency}. The SBM 
 generalizes the
Erd\H{o}s-R\'{en}yi graphs in the following way. Consider a graph $G(V,E)$, where $V = V_1 \sqcup V_2 \sqcup \dots \sqcup V_k$ is a disjoint union of $k$ clusters denoted by $V_1, \dots, V_k.$ The edges of the graph are drawn randomly: there is an edge between $u \in V_i$ and $v \in V_j$ with probability $q_{i,j}, 1\le i,j \le k.$
Given the adjacency matrix of such a graph, the task is to find exactly (or approximately) the partition $V_1 \sqcup V_2 \sqcup \dots \sqcup V_k$ of $V$.

This model has been incredibly popular both in theoretical and practical domains of community detection. Recent theoretical works  focus on characterizing sharp threshold of recovering the partition in the SBM. For example, when there are only two communities of exactly equal sizes, and the inter-cluster edge probability is $\frac{b\log n}{n}$ and intra-cluster edge probability is $\frac{a\log n}{n}$, it is known that exact recovery is possible if and only if $\sqrt{a} - \sqrt{b} > \sqrt{2}$ 
\citep{abh:16,mossel2015consistency}. The regime of the probabilities being $\Theta\Big(\frac{\log n}{n}\Big)$ has been put forward as one of most interesting ones, because  in an Erd\H{o}s-R\'{en}yi random graph, this is the threshold for graph connectivity \citep{bollobas1998random}. Note that the results are not only of theoretical interest, many real-world networks exhibit a ``sparsely connected'' community feature \citep{leskovec2008statistical}, and any efficient recovery algorithm for sparse SBM has many potential applications.  

While SBM is a popular model (because of its apparent simplicity), there are many aspects of real social networks, such as  ``transitivity rule'' (`friends having common friends') inherent to many social and other community structures, are not accounted for in SBM.  Defining a block model over a random geometric graph, the geometric block model (GBM), circumvents this since GBM naturally inherits the transitivity property of a random geometric graph. In a previous work  \citep{galhotra2017geometric}, we showed GBM models community structures better than an SBM in many real world networks (e.g.  DBLP, Amazon purchase network etc.). The GBM depends on the basic definition of the random geometric graph in  the same way the SBM depends on  Erd\H{o}s-R\'{en}yi graphs. The two-cluster GBM with vertex set $V= V_1\sqcup V_2$, $V_1 = V_2$ is a random graph defined in the following way. Suppose, $0 \le r_d < r_s \le 1/2$ be two  real numbers. For each vertex $u \in V$ randomly and independently choose a number $X_u \in [0,1]$ according to uniform distribution. There will be an edge between $u,v$ if and only if,
\begin{align*}
d_L(X_u, X_v) \le r_s & \text{ when } u, v \in V_1 \text{ or } u,v \in V_2\\
 d_L(X_u, X_v) \le r_d & \text{ when } u \in V_1, v \in V_2 \text{ or } u\in V_2,v \in V_1.
\end{align*}
Let us denote this random graph as ${\rm GBM}(r_s,r_d)$. Given this graph  ${\rm GBM}(r_s,r_d)$, the main problem of community detection is to find the parts $V_1$ and $V_2$. It has been shown in  \citep{galhotra2017geometric} that GBM accurately represents (more so than SBM) many real world networks.
Given a geometric random graph our main  objective  is to recover the partition (i.e., $V_1$ and $V_2$).

Motivated by SBM literature, we here also look at GBM in the connectivity regime, i.e., when $r_s = \frac{a\log n}{n}, r_d = \frac{b\log n}{n}$. Our first contribution in this part is to provide a lower bound that shows that it is impossible to recover the parts from ${\rm GBM}(\frac{a\log n}{n},\frac{b\log n}{n})$ when 
$
a - b < 1/2.
$
We also derive a relation between $a$ and $b$ that defines a sufficient condition of recovery in ${\rm GBM}(\frac{a\log n}{n},\frac{b\log n}{n})$ (see, Theorem~\ref{gbm:upper}). 
To analyze the algorithm proposed, we need to  crucially use the results obtained for the connectivity of vertex-random graphs. 




It is possible to generalize the GBM to include different distributions, different metric spaces and multiple parts. It is also possible to construct other type of spatial block models such as the one very recently being put forward in \citep{sankararaman2018community} which rely on the random dot product graphs \citep{young2007random}. In \citep{sankararaman2018community}, edges are drawn between vertices randomly and independently as a function  of the distance between the corresponding vertex random variables. In contrast, in GBM edges are drawn deterministically given the vertex random variables, and edges are dependent unconditionally. \citep{sankararaman2018community} also considers the recovery scenario where in addition to the graph, values of the vertex random variables are provided. In GBM, we only observe the graph. In particular, it will be later clear  that if we are given the corresponding random variables (locations) to the variables in addition to the graph, then recovery of the partitions in ${\rm GBM}(\frac{a\log n}{n},\frac{b\log n}{n})$ is possible if and only if $a - b > 0.5, a>1$.

\paragraph{VRG in Higher Dimension: The Random Annulus Graphs.} It is natural to ask similar question of connectivity for VRGs in higher dimension. In a VRG at dimension $t$, we may assign $t$-dimensional random vectors to each of the vertices, and use a standard metric such as the Euclidean distance to decide whether there should be an edge between two vertices. Formally, let us define the $t$-dimensional sphere as $S^t\equiv \{x \in \reals^{t+1} \mid \norm{x}_2=1\}$. Given two reals $0 \le r_1\le r_2\le 2$, the random annulus graph  ${\rm RAG}_t(n, [r_1,r_2])$ is a random graph with $n$ vertices. Each vertex $u$ is assigned a random vector $X_u$ selected randomly and uniformly from $S^t$. Two vertices $u$ and $v$ are connected by an edge, if and only if $r_1 \le d(u,v) \equiv \|X_u-X_v\|_2 \le r_2.$ Note that, for $t=1$ an ${\rm RAG}_1(n, [r_1,r_2])$ is nothing but a VRG as defined above, where we need to convert the Euclidean distance to the geodesic distance and scale the probabilities by a factor of $2\pi$. The ${\rm RAG}_t(n, [0,r])$ gives the standard definition of random geometric graphs in $t$ dimensions (for example, see \citep{bubecktriangle} or \citep{penrose2003random}).

We give the name random annulus graph (RAG) because two vertices are connected if one is within an `annulus' centered at the other. For the high dimensional random annulus graphs we extend our connectivity results of $t=1$ to general $t$. {In particular we show that there exists an isolated vertex in the ${\rm RAG}_t(n, [b(\frac{\log n}{n})^{\frac1t},a(\frac{\log n}{n})^{\frac1t}])$ with high probability if and only if 
$$a^t -b^t < \frac{\sqrt{\pi}(t+1)\Gamma(\frac{t+2}{2})}{\Gamma(\frac{t+3}{2})}\equiv \psi(t),$$ where $\Gamma(\cdot)$ is the gamma function.} Computing the connectivity threshold of RAG exactly is highly challenging, and we have to use several approximations of high dimensional geometry. Our arguments crucially rely on VC dimensions of sets of geometric objects such as intersections of high dimensional annuluses and hyperplanes. {Overall we find that the  ${\rm RAG}_t(n, [b(\frac{\log n}{n})^{\frac1t},a(\frac{\log n}{n})^{\frac1t}])$ is connected with high probability if 
$$
(a/2)^t-b^t \ge {8(t+1)\psi(t)} \text{  and  }  a>2b.
$$}
Using the connectivity result for ${\rm RAG}_t$, the results for the geometric block model can be extended to high dimensions. The latent feature space of nodes in most networks are high-dimensional. For example, road networks are two-dimensional whereas the number of features used in a social network may have much higher dimensions. In a `high-dimensional' GBM: for any $t >1$, instead of assigning a random variable from $[0,1]$ we assign a random vector $X_u \in S^t$  to each vertex $u$; and two vertices in the same part is connected if and only if their Euclidean distance is less than $r_s$, whereas two vertices from different parts are connected if and only if their distance is less than $r_d$. We show the algorithm developed for one dimension, extends to higher dimensions as well with nearly tight lower and upper bounds.

In this paper, we consistently refer to the $t=1$ case for RAG as vertex-random graph.

The paper is organized as follows.   In Section \ref{sec:notation}, we provide the formal definitions and the main results of the paper formally.  In Section~\ref{sec:rag}, the sharp connectivity phase transition results for vertex-random graphs are proven (details in Section~\ref{sec:VRG-detail}). In Section~\ref{sec:hrag}, the connectivity results are proven for high dimensional random annulus graphs (details in Section~\ref{sec:hrag-detail}). Finally, in Section~\ref{sec:gbm}, a lower bound for the geometric block model as well as the main recovery algorithm are presented (details in Section~\ref{sec:gbm-detail}).

\section{Main Results}
\label{sec:notation}
We formally define the random graph models, and state our results here.

\begin{definition}[Vertex-Random Graph]
A vertex-random graph ${\rm VRG}(n,[r_1,r_2])$ on $n$ vertices has parameters $n$,  and  a pair of  real numbers $r_1, r_2 \in [0,1/2], r_1 \le r_2$. It is defined by assigning a number $X_i \in \reals$ to vertex $i, 1 \le i \le n,$ where $X_i$s are independent and identical random variables uniformly distributed in $[0,1]$. There will be an edge between vertices $i$ and $j, i \ne j,$ if and only if $r_1 \le d_L(X_i,X_j)  \le r_2$ where $d_L(X_i,X_j) \equiv \min\{|X_i - X_j|, 1 - |X_i - X_j|\} $. 
\end{definition}
One can think of the random variables $X_i, 1\le i \le n$, to be uniformly distributed on the perimeter of a circle with radius $\frac1{2\pi}$ and the distance $d_L(\cdot,\cdot)$ to be the geodesic distance.
It will be helpful to consider vertices as just random points on $[0,1]$. Note that every point has a natural left direction (if we think of them as points on a circle then this is the counterclockwise direction) and a right direction. As a shorthand, for any two vertices $u,v$,  let $d(u,v)$ denote $d_L(X_u,X_v)$ where $X_u,X_v$ are corresponding random values to the vertices respectively. We can extend this notion to denote the distance $d(u,v)$ between a vertex $u$ (or the embedding of that vertex in $[0,1]$) and a point $v \in [0,1]$ naturally.

Our main result regarding vertex-random graphs is given in the following theorem. The base of the logarithm is $e$ here and everywhere else in the paper unless otherwise mentioned. 

\begin{theorem}[Connectivity threshold of vertex-random  graphs]\label{thm:rag}
The ${\rm VRG}(n,[\frac{b\log n}{n},\frac{a\log n}{n}])$  is connected with  probability $1-o(1)$ if $a >1$ and $a - b >0.5$. On the other hand, 
the ${\rm VRG}(n,[\frac{b\log n}{n},\frac{a\log n}{n}])$  is not connected with  probability $1-o(1)$ if $a <	1$ or $a-b < 0.5$. 
\end{theorem}
For the special case of $b=0$, the result was known (\citep{muthukrishnan2005bin,penrose2003random} See also \citep{penrose2016}). However, note that the case of $b >0$ is neither a straightforward generalization (i.e., the connectivity region is not defined by $a -b =1$) nor intuitive. 


%

\begin{definition}[The Random Annulus Graph]
Let us define the $t$-dimensional unit sphere as $S^t\equiv \{x \in \reals^{t+1} \mid \norm{x}_2=1\}$.
A random annulus graph ${\rm RAG}_t(n,[r_1,r_2])$ on $n$ vertices has parameters $n,t \in \integers_+$,  and a pair of  real numbers $r_1, r_2 \in [0,2], r_1 \le r_2$. It is defined by assigning a number $X_i \in S^t$ to vertex $i, 1 \le i \le n,$ where $X_i$'s are independent and identical random vectors uniformly distributed in $S^t$. There will be an edge between vertices $i$ and $j, i \ne j,$ if and only if $r_1 \le \|X_i-X_j\|_2  \le r_2$ where $\|\cdot\|_2$ denote the $\ell_2$ norm. 
\end{definition}
When from the context it is clear that we are in high dimensions, we use $d(u,v)$ to denote $\|X_u-X_v\|_2$ or just the $\ell_2$ distance between the arguments.

If we substitute $t=1$, then ${\rm RAG}_1(n,[r_1,r_2])$ is a random graph where each vertex is associated with a random variable uniformly distributed in the unit circle.  The distance between two vertices is the length of the chord connecting the random variables corresponding to the two vertices. If the length of the chord is $r \le 2$, then the length of the corresponding (smaller) chord {length of the corresponding arc between the vertices along the circumference of the circle} is $2\sin^{-1}\frac{r}{2}$. If we normalize the circumference of the circle by $2\pi$ we obtain a random graph model that is equivalent to our definition of the vertex-random graphs. Since handling geodesic distances is more cumbersome in the higher dimensions, we resorted to Euclidean distance.

We derived the following results about the existence of isolated vertices in random annulus graphs.
\begin{theorem}[Zero-One law for Isolated Vertex in RAG]\label{th:lb}
For a random annulus graph ${\rm RAG}_t(n,[r_1,r_2])$ where $r_2=a \Big(\frac{ \log n}{n}\Big)^{\frac{1}{t}}$ and $r_1=b\Big(\frac{\log n}{n}\Big)^{\frac{1}{t}}$, there exists isolated nodes with  probability $1-o(1)$ if 
$$a^t -b^t <  \frac{\sqrt{\pi}(t+1)\Gamma(\frac{t+2}{2})}{\Gamma(\frac{t+3}{2})}\equiv \psi(t),$$ where $\Gamma(x) = \int_0^\infty y^{x-1}e^{-y}dy$ is the gamma function, and there does not exist an isolated vertex with probability $1-o(1)$ if $a^t -b^t > \psi(t)$.
\end{theorem}

An obvious deduction from this theorem is that an ${\rm RAG}_t(n,[b\Big(\frac{\log n}{n}\Big)^{\frac{1}{t}},a \Big(\frac{ \log n}{n}\Big)^{\frac{1}{t}}])$ is not connected with probability $1-o(1)$ if $a^t -b^t < \psi(t)$. Our main result here gives a condition that guarantees connectivity in this regime.

\begin{theorem}\label{thm:highdem1}
A $t$ dimensional random annulus graph ${\rm RAG}_t(n,[b\Big(\frac{\log n}{n}\Big)^{\frac{1}{t}},a \Big(\frac{ \log n}{n}\Big)^{\frac{1}{t}}])$ is connected with  probability $1-o(1)$ if 
\begin{align*}
(a/2)^t-b^t \ge {8(t+1)\psi(t)}\text{  and  }  a>2b.
\end{align*} 
\end{theorem}

All these connectivity results find immediate application in analyzing the algorithm that we propose for the geometric block model (GBM). 
A GBM is a generative model for networks (graphs) with underlying community structure.
\begin{definition}[Geometric Block Model]
Given $V = V_1\sqcup V_2, |V_1|=|V_2| = \frac{n}2$,  choose a random variable $X_u$ uniformly distributed in $[0,1]$ for all $u \in V$.
The geometric block model  ${\rm GBM}(r_s, r_d)$ with parameters $r_s> r_d$ is a random graph where an edge exists between vertices $u$ and $v$  if and only if,
\begin{align*}
d_L(X_u, X_v) \le r_s & \text{ when } u, v \in V_1 \text{ or } u,v \in V_2\\
 d_L(X_u, X_v) \le r_d & \text{ when } u \in V_1, v \in V_2 \text{ or } u\in V_2,v \in V_1.
\end{align*}
 \end{definition}
As a consequence of the connectivity lower bound on  VRG, we are able to show that  recovery of the partition is not possible with high probability in  ${\rm GBM}(\frac{a\log n}{n}, \frac{b \log n}{n})$  by any means whenever $a - b <0.5$ or $a<1$ (see, Theorem~\ref{gbm:lower}). Another consequence of the vertex-random graph results is that we show that if in  addition to a GBM graph, all the locations of the vertices are also provided, then recovery is possible if and only if $a - b >0.5$ or $a>1$ (formal statement in Theorem~\ref{thm:gbmplus}).

Coming back to the actual recovery problem, our main contribution for GBM is to provide a simple and efficient algorithm that performs  well in the sparse regime (see, Algorithm~\ref{alg:alg1}). 

\begin{theorem}[Recovery algorithm for GBM]\label{gbm:upper}
Suppose we have the graph $G(V,E)$ generated according to ${\rm GBM}(r_s \equiv \frac{a\log n}{n},r_d\equiv \frac{b\log n}{n}), a \ge 2b$. 
Define
\begin{align*}
t_1&=\min\{t: (2b+t)\log \frac{2b+t}{2b}-t > 1\},~~~
t_2=\min\{t : (2b-t)\log \frac{2b-t}{2b}+t > 1\}\\
\theta_1 &= \max\{\theta:\frac{1}{2}\Big((4b+2t_1)\log \frac{4b+2t_1}{2a-\theta}+2a-\theta-4b-2t_1\Big) > 1  \text{ and } 0 \le \theta \le 2a-4b-2t_1\}\\
\theta_2 &= \min\{\theta: \frac{1}{2}\Big((4b-2t_2\log \frac{4b-2t_2}{2a-\theta}+2a-\theta-4b+2t_2\Big) > 1
 \text{ and }  
a \ge \theta \ge \max\{2b,2a-4b+2t_2\}\}.
\end{align*}
Then there exists an efficient algorithm which will recover the correct partition in the GBM with  probability $1-o(1)$  if   $a-\theta_2+\theta_1>2$ OR $a-\theta_2 > 1, a>2$. 
\end{theorem}

Some example of the parameters when the proposed algorithm (Algorithm~\ref{alg:alg1}) can successfully recover is given  in Table~\ref{tab:per}.
 
\begin{table}[h!]\label{tab:per}
\centering
\begin{tabular}{ |c|c|c|c|c|c|c|c|c| } 
 \hline
$b$ & 0.01 & 1 & 2 & 3 & 4 & 5 & 6 & 7 \\ 
\hline
Minimum value of $a$ & 3.18 & 8.96 & 12.63 & 15.9 & 18.98 & 21.93 & 24.78 & 27.57 \\
 \hline 
\end{tabular}
\caption{Minimum value of $a$, given $b$ for which Algorithm~\ref{alg:alg1} resolves clusters correctly in ${\rm GBM}(\frac{a\log n}{n},\frac{b\log n}{n})$.\label{tab:per}}
\end{table}

As can be anticipated, the connectivity results for RAG applies to the `high dimensional' geometric block model. 
\begin{definition}[The GBM in High Dimensions]
Given $V = V_1\sqcup V_2, |V_1|=|V_2| = \frac{n}2$,  choose a random vector $X_u$ independently uniformly distributed in $S^t$ for all $u \in V$.
The geometric block model  ${\rm GBM}_t(r_s, r_d)$ with parameters $r_s> r_d$ is a random graph where an edge exists between vertices $u$ and $v$  if and only if,
\begin{align*}
\norm{X_u-X_v}_2 \le r_s & \text{ when } u, v \in V_1 \text{ or } u,v \in V_2\\
\norm{X_u-X_v}_2 \le r_d & \text{ when } u \in V_1, v \in V_2 \text{ or } u\in V_2,v \in V_1.
\end{align*}
\end{definition}
We extend the algorithmic results to high dimensions.
\begin{theorem}
\label{theorem:intro-1}
There exists a polynomial time efficient algorithm that recovers the partition from ${\rm GBM}_t(r_s, r_d)$ with probability $1-o(1)$ if $r_s=\Theta((\frac{\log{n}}{n})^{\frac{1}{t}})$ and $r_s-r_d=\Omega( (\frac{\log n}{n})^{\frac{1}{t}})$. Moreover, any  algorithm fails to recover the parts with probability at least $1/2$ if $r_s-r_d=o( (\frac{\log n}{n})^{\frac{1}{t}})$ or $r_s=o((\frac{\log{n}}{n})^{\frac{1}{t}})$.
\end{theorem}




\section{Connectivity of Vertex-Random Graphs}
\label{sec:rag}
In this section we give a sketch of the proof of  sufficient condition for connectivity of VRG (as part of proving Theorem~\ref{thm:rag}). The full details along with the proof of the necessary condition (Theorem~\ref{thm:lower_bound}) are given in Section~\ref{sec:VRG-detail}. 
\subsection{Sufficient condition for connectivity of VRG}
\begin{theorem}\label{thm:upper}
The vertex-random graph ${\rm VRG}(n,[\frac{b\log n}{n},\frac{a\log n}{n}])$  is connected with  probability $1-o(1)$ if $a >1$ and $a - b >0.5$.
\end{theorem}
To prove this theorem we use two main technical lemmas that show two different events happen with high probability simultaneously.

\begin{lemma}\label{lem:lemma1}
A set of vertices $\cC \subseteq V$ is called a cover of $[0,1]$, if for any point $y$ in $[0,1]$ there exists a vertex $v\in \cC$ such that $d(v,y) \le \frac{a\log n}{2n}$.  A ${\rm VRG}(n, [\frac{b\log n}{n}, \frac{a\log n}{n}])$ is a union of cycles  
such that every cycle forms a cover of $[0,1]$  (see Figure~\ref{fig:arr0}) as long as $a-b >0.5$ and $a >1$ with  probability $1-o(1)$.
\end{lemma}

This lemma also shows effectively the fact that `long-edges' are able to connect vertices over multiple hops.
Note that, the statement of Lemma~\ref{lem:lemma1} would be easier to prove if the condition were $a-b>1$. 
In that case what we prove is that every vertex  has neighbors (in the VRG) on both of the left and right directions. To  see this
for each vertex $u$ , assign two indicator $\{0,1\}$-random variables $A_{u}^{l}$ and $A_{u}^{r}$, with $A_{u}^{l}=1$ if and only if there is no node $x$ to the left of node $u$ such that $d(u,x) \in [\frac{b\log n}{n},\frac{a\log n}{n}]$. Similarly, let $A_{u}^{r}=1$ if and only if there is no node $x$ to the right of node $u$ such that $d(u,x)\in [\frac{b\log n}{n},\frac{a\log n}{n}]$. Now define $A=\sum_{u}( A_{u}^{l}+A_{u}^{r})$.  We have, 

\begin{wrapfigure}{r}{0.35\textwidth}
\vspace{-25pt}
\centering
\begin{tikzpicture}[thick, scale=0.5]
 \filldraw[color=black, fill=red!0,  line width=2pt](0,0) circle (3.5);
 \draw [line width=2pt, color=black!60,->](0,0) -- (3.5,0)node[anchor=north east] at (2,0) {$\frac{1 }{2\pi}$};
 \filldraw [gray] (0,0) circle (2pt);
\filldraw [red] (0,3.5) circle (2pt);
\filldraw [red] (3.5,0) circle (2pt);
\filldraw [red] (0,-3.5) circle (2pt);
\filldraw [red] (-3.5,0) circle (2pt);
\filldraw [red] (2.47,2.47) circle (2pt);
\filldraw [red] (2.47,-2.47) circle (2pt);
\filldraw [red] (-2.47,2.47) circle (2pt);
\filldraw [red] (-2.47,-2.47) circle (2pt);
\draw[red]    (0,3.5) to[out=80,in=40] (2.47,2.47);
\draw[red]    (2.47,2.47) to[out=40,in=60](3.5,0) ;
\draw[red]    (3.5,0) to[out=40,in=-50](2.47,-2.47) ;
\draw[red]    (2.47,-2.47) to[out=-20,in=-60](0,-3.5) ;
\draw [red]   (0,-3.5) to[out=-100,in=-100](-2.47,-2.47) ;
\draw [red]   (-2.47,-2.47) to[out=-120,in=-200](-3.5,0) ;
\draw [red]   (-3.5,0) to[out=150,in=-180](-2.47,2.47) ;
\draw [red]   (-2.47,2.47) to[out=150,in=100](0,3.5) ;
\filldraw [blue] (3.23,1.339) circle (2pt);
\filldraw [blue] (3.23,-1.339) circle (2pt);
\filldraw [blue] (-3.23,1.339) circle (2pt);
\filldraw [blue] (-3.23,-1.339) circle (2pt);
\filldraw [blue] (1.339,3.23) circle (2pt);
\filldraw [blue] (1.339,-3.23) circle (2pt);
\filldraw [blue] (-1.339,3.23) circle (2pt);
\filldraw [blue] (-1.339,-3.23) circle (2pt);
\draw[blue]   (3.23,1.339) to[out=10,in=10, , distance=2cm](3.23,-1.339);
\draw[blue]    (3.23,-1.339) to[out=-20,in=-60, distance=2cm](1.339,-3.23) ;
\draw[blue]    (1.339,-3.23) to[out=-80,in=-80,distance=2cm](-1.339,-3.23) ;
\draw[blue]    (-1.339,-3.23) to[out=-100,in=-150,distance=2cm](-3.23,-1.339) ;
\draw[blue]    (-3.23,-1.339) to[out=-150,in=150,distance=2cm](-3.23,1.339) ;
\draw[blue]    (-3.23,1.339) to[out=120,in=150,distance=2cm](-1.339,3.23) ;
\draw[blue]    (-1.339,3.23) to[out=120,in=50,distance=2cm](1.339,3.23) ;
\draw[blue]    (1.339,3.23) to[out=20,in=50,distance=2cm](3.23,1.339) ;
\end{tikzpicture}
\vspace{-20pt}
\caption{Each vertex having two neighbors on either direction implies the graph is a union of cycles. The cycles can be interleaving in $[0,1]$.\label{fig:arr0}}
\end{wrapfigure}
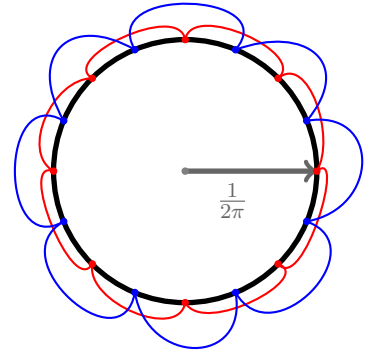
$$\Pr(A_{u}^{l}=1)=\Pr(A_{u}^{r}=1)=(1-\frac{(a-b)\log n}{n})^{n-1},$$ 
and, 
$$\avg[A]=2n(1-\frac{(a-b)\log n}{n})^{n-1} \le 2n^{1-(a-b)}.$$
 If $a-b >1$ then $\avg[A]=o(1)$ which implies, by invoking Markov inequality, that with high probability every node will have neighbors (connected by an edge in the VRG) on either side. This results in the interesting conclusion that every vertex will lie in  a cycle that covers  $[0,1]$. This is true for every vertex, hence the graph is simply a union of cycles each of which is a cover of $[0,1]$.
The main technical challenge is to show that this conclusion remains valid even when $a-b >0.5$, which is proved in Lemma~\ref{lem:lemma1} in Section~\ref{sec:VRG-detail}. 

\begin{lemma}\label{lem:lemma2}
Set two real numbers $k\equiv \lceil b/(a-b)\rceil+1$ and $\epsilon < \frac1{2k}$. In an ${\rm VRG}(n, [\frac{b\log n}{n}, \frac{a\log n}{n}]), 0 <b <a$, with  probability $1-o(1)$ there exists a vertex $u_0$ and  $k$ nodes $\{u_1, u_2 ,\ldots, u_k\}$ to the right of $u_0$ such that $d(u_0,u_i) \in [\frac{(i(a-b)-2i\epsilon)\log n}{n},\frac{(i(a-b)-(2i-1)\epsilon) \log n}{n}]$ and  $k$ nodes $\{v_1, v_2 ,\ldots, v_k\}$ to the right of $u_0$ such that $d(u_0,v_i) \in [\frac{((i(a-b)+b-(2i-1)\epsilon)\log n}{n},\frac{(i(a-b)+b -(2i-2)\epsilon)\log n}{n}]$, for $i =1,2,\ldots,k$. The arrangement of the vertices is shown  in Figure~\ref{fig:arr} (pg. 18).
\end{lemma}

With the help of these two lemmas, we are in a position to  prove Theorem~\ref{thm:upper}. The proofs of  the two lemmas are given in Section~\ref{sec:VRG-detail} and contain the technical essence of this section.

\begin{proof}[Proof of Theorem~\ref{thm:upper}]
We have shown that the two events mentioned in Lemmas \ref{lem:lemma1} and \ref{lem:lemma2} happen with high probability. Therefore they simultaneously happen under the condition $a>1$ and $a-b >0.5$.
Now we will show that these events together imply that the graph is connected. To see this, consider the vertices $u_0,\{u_1, u_2, \ldots, u_k\}$ and $\{v_1, v_2, \ldots, v_k\}$ that satisfy the conditions of Lemma \ref{lem:lemma2}. We can observe that each vertex $v_i$ has an edge with $u_i$ and $u_{i-1}$, $i =1, \ldots,k$. This is because (see Figure~\ref{fig:arr} for a depiction)
$$
d(u_i,v_i) 
 \ge \frac{((i(a-b)+b-(2i-1)\epsilon)\log n}{n} - \frac{i(a-b)-(2i-1)\epsilon) \log n}{n}= \frac{b \log n}{n} \quad \text{and}
$$ 
\begin{align*}
 d(u_i,v_i) &\le 
   \frac{i(a-b)+b -(2i-2)\epsilon\log n}{n} - \frac{(i(a-b)-2i\epsilon)\log n}{n} = \frac{(b + 2\epsilon)\log n}{n}.
\end{align*}

Similarly, 
\vspace{-10pt}
\begin{align*}
d(u_{i-1},v_i)& 
\ge \frac{((i(a-b)+b-(2i-1)\epsilon)\log n}{n} - \frac{(i-1)(a-b)-(2i-3)\epsilon) \log n}{n}\\
& = \frac{(a-2\epsilon)\log n}{n} \quad \text{and}
\end{align*}
\begin{align*}
d(u_{i-1},v_i) 
&\le \frac{i(a-b)+b -(2i-2)\epsilon\log n}{n} - \frac{((i-1)(a-b)-2(i-1)\epsilon)\log n}{n} = \frac{a\log n}{n}.
\end{align*}
 This implies that $u_0$ is connected to $u_i$ and $v_i$ for all $i=1,\dots,k$. Using Lemma \ref{lem:lemma1}, the first event implies  that 
 the connected components are cycles spanning the entire line $[0,1]$.
 Now consider two such disconnected components, one of which consists of the nodes $u_0,\{u_1, u_2, \ldots, u_k\}$ and $\{v_1, v_2, \ldots, v_k\}$. There must exist a node $t$ in the other component (cycle) such that $t$ is on the right of $u_0$ and $d(u_0,t) \equiv \frac{x \log n}{n}  \le \frac{a \log n}{n}$. If $x\leq b$, $\exists i \mid  i\leq k \text{ and } i(a-b)+b - a-(2i-2)\epsilon \le x\leq i(a-b)-(2i-1)\epsilon$ (see Figure~\ref{fig:arr1d}). 
When $x\leq b$, we can calculate the distance between $t$ and $v_i$ as 
\begin{align*}
d(t,v_i)
&\ge \frac{(i(a-b)+b-(2i-1)\epsilon)\log n}{n} - \frac{(i(a-b)-(2i-1)\epsilon) \log n}{n} = \frac{b \log n}{n}
\end{align*}
and 
\begin{align*}
d(t,v_i)
&\le \frac{(i(a-b)+b -(2i-2)\epsilon)\log n}{n} - \frac{(i(a-b)+b-a-(2i-2)\epsilon)\log n}{n} = \frac{a\log n}{n}.
\end{align*}
Therefore  $t$ is connected to $v_i$ when  $x\leq b.$ If $x >b$ then $t$ is already connected to  $u_0$. Therefore the two components (cycles) in question are connected.
This is true for all cycles and hence there is only a single component in the entire graph. Indeed, if we consider the cycles to be disjoint super-nodes, then we have shown  that there must be a star configuration. 
\end{proof}

\section{Connectivity of High Dimensional Random Annulus Graphs: Proof of Theorem \ref{thm:highdem1}}
\label{sec:hrag}
In this section we show a proof sketch of Theorem \ref{thm:highdem1} to establish the sufficient condition of connectivity of random annulus graphs. The details of the proof and the necessary conditions are provided in Section~\ref{sec:hrag-detail}.


 
Note, here $r_1\equiv   b\left(\frac{\log n}{n}\right)^{1/t}$ and $r_2 \equiv a\left(\frac{\log n}{n}\right)^{1/t}$.
We show the upper bound for connectivity of a Random Annulus Graphs in $t$ dimension as shown in Theorem \ref{thm:highdem1}. 
For this we first define a \emph{pole} as a vertex which is connected to all vertices within a distance of $r_2$ from itself. 
In order to prove  Theorem \ref{thm:highdem1}, we first show the existence of a pole with high probability in Lemma \ref{lem:pole}.
\begin{lemma}\label{lem:pole}
In a ${\rm RAG}_t\left(n, \left[b\left(\frac{\log n}{n}\right)^{1/t},a\left(\frac{\log n}{n}\right)^{1/t}\right]\right), 0 <b <a$, with  probability $1-o(1)$ there exists a pole.
\end{lemma}
 Next, Lemma \ref{lem:high_stuff1} shows that for every vertex $u$ and every hyperplane $L$ passing through $u$ and not too close to the tangent hyperplane at $u$, there will be a neighbor of $u$ on either side of the plane. Therefore, there should be a neighbor towards the direction of the pole. In order to formalize this, let us define a few regions associated with a node $u$ and a hyperplane $L:w^{T}x=\beta$ passing through $u$.
\begin{align*}
\mathcal{R}_{L}^1 &\equiv \{x \in S^t \mid r_1 \le  d(u,x) \le r_2, w^{T}x \le \beta \} \\
\mathcal{R}_{L}^2 &\equiv \{x \in S^t \mid r_1 \le  d(u,x) \le r_2, w^{T}x \ge \beta \} \\
\mathcal{A}_{L} & \equiv \{x \mid x \in \mathcal{S}^t, \quad w^{T}x=\beta \}.
\end{align*}
Informally, $\mathcal{R}_{L}^1$ and $\mathcal{R}_{L}^2$ represent the partition of the annulus on either side of the hyperplane $L$ and $\mathcal{A}_L$ represents the region on the sphere lying on $L$. 
\begin{lemma}\label{lem:high_stuff1}
If we sample $n$ nodes from $S^t$ according to ${\rm RAG}_t\left(n,\left[b\left(\frac{\log n}{n}\right)^{1/t},a\left(\frac{\log n}{n}\right)^{1/t}\right]\right)$, then for every node $u$ and every hyperplane $L$ passing through $u$ such that $\mathcal{A}_L$ is not all within  distance $r_2$ of $u$, node $u$ has a neighbor on both sides of the hyperplane $L$ with probability at least $1-\frac{1}{n}$ provided 
$
(a/2)^t-b^t \ge   \frac{8\sqrt{\pi}(t+1)^2\Gamma(\frac{t+2}{2})}{\Gamma(\frac{t+3}{2})}
$ and $a>2b$.
\end{lemma}
We found the proof of this lemma to be challenging. Since, we do not know the location of the pole, we need to show that every point has a neighbor on both sides of the plane $L$ no matter what the orientation of the plane. Since the number of possible orientations is uncountably infinite, we cannot use a union-bound type argument. To show this we have to rely on the VC Dimension of the family of sets $\{x \in S^t \mid r_1 \le  \|u-x\|_2 \le r_2, w^{T}x \ge \beta , \mathcal{A}_{L:w^{T}x=\beta}  \text{ not all within } r_2 \text{ of } u\}$ for all hyperplanes $L$ (which can be shown to be less than $t+1$). We rely on the celebrated result of \cite{haussler1987} (we derived a continuous version of it), see Theorem~\ref{thm:VC_dim}, to deduce our conclusion.

For a node $u$ { and its corresponding location $X_u=(u_1,u_2,\dots,u_{t+1})$}, define the particular hyperplane $L^{\star}_u : x_1=u_1$ which is normal to the line joining $u_0 \equiv (1,0,\dots,0)$ and the origin and passes through $u$. We now need one more lemma that will help us prove Theorem \ref{thm:highdem1}.
\begin{lemma}\label{lem:cut_twice}
For a particular node $u$ and corresponding hyperplane $L^{\star}_u$, if every point in $\mathcal{A}_{L^{\star}_u}$ is within distance $r_2$ from $u$,  then $u$ must be within $r_2$ of $u_0$.
\end{lemma}

For now, we assume that the Lemmas \ref{lem:pole}, \ref{lem:high_stuff1} and \ref{lem:cut_twice} are true and show why these lemmas together imply the proof of Theorem \ref{thm:highdem1}.
\begin{proof}[Proof of Theorem \ref{thm:highdem1}] 
We consider an alternate (rotated but not shifted) coordinate system by multiplying every vector by an orthonormal matrix such that the new position of the pole is the $t+1$-dimensional vector $(1,0,\dots,0)$ where only the first co-ordinate is non-zero. Let the $t+1$ dimensional vector describing any node $u$ in this new coordinate system be $\hat{u}=(\hat{u}_1,\hat{u}_2,\dots,\hat{u}_{t+1})$. Now consider the hyperplane $L: x_1=\hat{u}_1$ and if $u$ is not connected to the pole already, then by Lemma \ref{lem:high_stuff1} and Lemma \ref{lem:cut_twice}, the node $u$ has a neighbor $u_2$ which has a higher first coordinate ($\hat{u}_2 >\hat{u}_1$). The same analysis applies for $u_2$ and hence we have a path where the first coordinate of every node is higher than the previous node. Since the number of nodes is finite, this path cannot go on indefinitely and at some point, one of the nodes is going to be within $r_2$ of the pole and will be connected to the pole. Therefore every node is going to be connected to the pole and hence our theorem is proved. 
\end{proof}
\section{The Geometric Block Model}
\label{sec:gbm}
In this section, we prove the necessary condition for exact cluster recovery of GBM and give an efficient algorithm that matches that within a constant factor. The details are provided in Section~\ref{sec:gbm-detail}.
\subsection{Immediate consequence of VRG connectivity}
The following lower bound for GBM can be obtained as a  consequence of Theorem~\ref{thm:rag}. 

\begin{theorem}[Impossibility in GBM]\label{gbm:lower}
Any algorithm to recover the partition in ${\rm GBM}(\frac{a \log n}{n},\frac{b \log n}{n})$ will give incorrect output with probability $1-o(1)$ if $a-b < 0.5$ or $a < 1$.
\end{theorem}
\begin{proof}
Consider the scenario that not only the geometric block model graph  ${\rm GBM}(\frac{a \log n}{n},\frac{b \log n}{n})$ was provided to us, but also the random values $X_u \in [0,1]$ for all vertex $u$ in the graph were provided. We will show that we will still not be able to recover the correct partition of the vertex set $V$ with probability at least $0.5$ (with respect to choices of $X_u,~u, v\in V$ and any randomness in the algorithm).

In this situation, the edge $(u,v)$ where $d_L(X_u,X_v) \le \frac{b \log n}{n}$ does not give any new information than $X_u,X_v$. However the edges $(u,v)$ where $\frac{b \log n}{n} \le d_L(X_u,X_v) \le \frac{a \log n}{n}$ are informative, as existence of such an edge will imply that $u$ and $v$ are in the same part. These edges constitute a vertex-random graph ${\rm VRG}(n, [\frac{b \log n}{n},\frac{a \log n}{n}])$. But if there are more than two components in this vertex-random graph, then it is impossible to separate out the vertices into the correct two parts, as the connected components can be assigned to any of the two parts and the VRG along with the location values ($X_u, u \in V$) will still be consistent. 

What remains to be seen that ${\rm VRG}(n, [\frac{b \log n}{n},\frac{a \log n}{n}])$ will have $\omega(1)$ components with high probability  if $a-b < 0.5$ or $a < 1$. This is certainly true when $a-b < 0.5$ as we have seen in Theorem~\ref{thm:lower_bound}, there can indeed be $\omega(1)$ isolated nodes with high probability. On the other hand, when $a<1$, just by using an analogous argument it is possible to show that there are $\omega(1)$ vertices that do not have any neighbors on the left direction (counterclockwise). We delegate the proof of this claim as Lemma \ref{lem:disc} in the appendix.  
If there are $k$ such vertices, there must be at least $k-1$ disjoint candidates. This completes the proof.
\end{proof}

Indeed, when the locations $X_u$ associated with every vertex $u$ is provided, it is also possible to recover the partition exactly when $a-b > 0.5$ and $a > 1$, matching the above lower bound exactly (see Theorem~\ref{thm:gbmplus}).

Similar impossibility result extends to higher dimensional GBM from the necessary condition on connectivity of RAG.

\subsection{A recovery algorithm for GBM}
We now turn our attention to an efficient recovery algorithm for GBM. Intriguingly, we show a simple triangle counting algorithm works well for GBM and recovers the clusters in the sparsity regime. Triangle counting algorithms are popular heuristics applied to social networks for clustering \cite{easley2012networks}, however they fail in SBM. Hence, this serves as another validation why GBM are well-suited to model community structures in social networks.

The algorithm is as follows. Given a graph $G=(V:|V|=n,E)$ with two disjoint parts, $V_1, V_2 \subseteq V$ generated according to ${\rm GBM}(r_s, r_d)$, the algorithm (see Algorithm ~\ref{alg:alg1}) goes over  all edges  $(u,v)\in E$. It counts the number of triangles  containing the edge  $(u,v)$ and leave the edge intact if and only if the number of triangles are {\em not} within two specified thresholds $E_S$ and $E_D$. It then returns the connected components of the redacted graph. Having two thresholds is somewhat non-intuitive. We show two vertices in different components can only have number of common neighbors within $E_S$ and $E_D$, and thus all those edges get removed during the first iteration. In this process, some intra-cluster edges also get removed, but using the connectivity property of VRG, we are able to show the clusters still remain connected. 

The same algorithm extends to higher dimensions as well, showing irrespective of underlying dimensionality, there exists a good algorithm. The proof here relies on the connectivity of random annulus graphs.

\section{Connectivity of Vertex-Random Graphs: Details}
\label{sec:VRG-detail}
In this section, we prove the necessary and sufficient condition for connectivity of VRG in full details.
\subsection{Necessary condition for connectivity of VRG}

\begin{theorem}[VRG connectivity lower bound]\label{thm:lower_bound}
The ${\rm VRG}(n,[\frac{b\log n}{n},\frac{a\log n}{n}])$  is not connected
with  probability $1-o(1)$ if $a <	1$ or $a-b < 0.5$. 
\end{theorem}
\begin{proof}
First of all,  it is known that ${\rm VRG}(n,[0,\frac{a\log n}{n}])$ is not connected  with  high probability when $a <1$ \citep{muthukrishnan2005bin,penrose2003random}. Therefore  ${\rm VRG}(n,[\frac{b\log n}{n},\frac{a\log n}{n}])$ must not be connected with high probability when $a<1$ as the connectivity interval is a strict subset of the previous case, and ${\rm VRG}(n,[\frac{b\log n}{n},\frac{a\log n}{n}])$ can be obtained from ${\rm VRG}(n,[0,\frac{a\log n}{n}])$ by deleting all the edges that has the two corresponding random variables separated by distance less than $\frac{b\log n}{n}$.

Next we will show that if $a-b < 0.5$ then there exists an isolated vertex with high probability. It would be easier to think of each vertex as a uniform random point in $[0,1]$. 
Define an indicator variable $A_{u}$ for every node $u$ which is 1 when node $u$ is isolated and $0$ otherwise. We have,
$$\Pr(A_{u}=1)=\bigg(1-\frac{2(a-b)\log n}{n}\bigg)^{n-1}.$$ 
Define $A=\sum_{u} A_{u}$, and hence
$$\avg[A]=n\Big(1-\frac{2(a-b)\log n}{n} \Big)^{n-1}= n^{1-2(a-b)-o(1)}.$$ Therefore, when $a-b < 0.5$, $\avg[A]=\Omega(1)$. To prove this statement with high probability we can show that the variance of $A$ is bounded. Since $A$ is a sum of indicator random variables, we have that 
$${\rm Var}(A) \le \avg[A]+\sum_{u \neq v} {\rm Cov}(A_u,A_v)=\avg[A]+\sum_{u \neq v} (\Pr(A_u=1 \cap A_v=1)-\Pr(A_u=1)\Pr(A_v=1)).$$ 
Now, consider the scenario when the  vertices $u$ and $v$ are at a distance more than $\frac{2a \log n}{n}$ apart (happens with probability $1-\frac{4a\log n}{n})$. Then the region in $[0,1]$ that is between  distances $\frac{b \log n}{n}$ and $\frac{a\log n}{n}$ from both of the vertices is empty and therefore $\Pr(A_u=1 \cap A_v=1) =  
\Big(1- \frac{4(a-b) \log n}{n}\Big)^{n-2}.$
When the vertices  are within distance $\frac{2a \log n}{n}$ of one another, then
$
\Pr(A_u=1 \cap A_v=1) \le \Pr(A_u=1).
$
Therefore,
\begin{align*}
\Pr(A_u=1 \cap A_v=1) \le &(1-\frac{4a\log n}{n}) \Big(1- \frac{4(a-b) \log n}{n}\Big)^{n-2} + \frac{4a\log n}{n}\Pr(A_u=1)\\ 
&\le  (1-\frac{4a\log n}{n}) n^{-4(a-b)+o(1)}+ \frac{4a\log n}{n} n^{-2(a-b)+o(1)}.
\end{align*}
Consequently for large enough $n$,
\begin{align*}
\Pr(A_u=1 \cap A_v=1)-\Pr(A_u=1)\Pr(A_v=1) &\le (1-\frac{4a\log n}{n}) n^{-4(a-b)+o(1)} \\+ \frac{4a\log n}{n} n^{-2(a-b)+o(1)}  -& n^{-4(a-b) +o(1)} 
\le  \frac{8a\log n}{n}\Pr(A_u=1).
\end{align*}
Now,
$$
{\rm Var}(A) \le \avg[A] + \binom{n}{2}\frac{8a\log n}{n}\Pr(A_u=1) \le \avg[A](1+ 4a\log n).
$$
By using Chebyshev bound, with probability at least $1-\frac{1}{\log n}$, 
$$A >n^{1-2(a-b)}-\sqrt{n^{1-2(a-b)}(1+4a\log n)\log n},$$
which imply for $a-b <0.5$, there will exist isolated nodes with high probability.
\end{proof}

\subsection{Sufficient condition for connectivity of VRG}

\begin{theorem*}[\ref{thm:upper}]
The vertex-random graph ${\rm VRG}(n,[\frac{b\log n}{n},\frac{a\log n}{n}])$  is connected with  probability $1-o(1)$ if $a >1$ and $a - b >0.5$.
\end{theorem*}
To prove this theorem we use two main technical lemmas that show two different events happen with high probability simultaneously.

\begin{lemma*}[\ref{lem:lemma1}]
A set of vertices $\cC \subseteq V$ is called a cover of $[0,1]$, if for any point $y$ in $[0,1]$ there exists a vertex $v\in \cC$ such that $d(v,y) \le \frac{a\log n}{2n}$.  A ${\rm VRG}(n, [\frac{b\log n}{n}, \frac{a\log n}{n}])$ is a union of cycles  
such that every cycle forms a cover of $[0,1]$  (see Figure~\ref{fig:arr0}) as long as $a-b >0.5$ and $a >1$ with  probability $1-o(1)$.
\end{lemma*}
This lemma also shows effectively the fact that `long-edges' are able to connect vertices over multiple hops.
Note that, the statement of Lemma~\ref{lem:lemma1} would be easier to prove if the condition were $a-b>1$. 
In that case what we prove is that every vertex  has neighbors (in the VRG) on both of the left and right directions. To  see this
for each vertex $u$ , assign two indicator $\{0,1\}$-random variables $A_{u}^{l}$ and $A_{u}^{r}$, with $A_{u}^{l}=1$ if and only if there is no node $x$ to the left of node $u$ such that $d(u,x) \in [\frac{b\log n}{n},\frac{a\log n}{n}]$. Similarly, let $A_{u}^{r}=1$ if and only if there is no node $x$ to the right of node $u$ such that $d(u,x)\in [\frac{b\log n}{n},\frac{a\log n}{n}]$. Now define $A=\sum_{u}( A_{u}^{l}+A_{u}^{r})$.  We have, 
$$\Pr(A_{u}^{l}=1)=\Pr(A_{u}^{r}=1)=(1-\frac{(a-b)\log n}{n})^{n-1},$$ 
and, 
$$\avg[A]=2n(1-\frac{(a-b)\log n}{n})^{n-1} \le 2n^{1-(a-b)}.$$
 If $a-b >1$ then $\avg[A]=o(1)$ which implies, by invoking Markov inequality, that with high probability every node will have neighbors (connected by an edge in the VRG) on either side. This results in the interesting conclusion that every vertex will lie in  cycle that covers  $[0,1]$. This is true for every vertex, hence the graph is simply a union of cycles each of which is a cover of $[0,1]$.
The main technical challenge is to show that this conclusion remains valid even when $a-b >0.5$, which is proved after we describe the other components of the result in this section. 

\begin{lemma*}[\ref{lem:lemma2}]
Set two real numbers $k\equiv \lceil b/(a-b)\rceil+1$ and $\epsilon < \frac1{2k}$. In an ${\rm VRG}(n, [\frac{b\log n}{n}, \frac{a\log n}{n}]), 0 <b <a$, with  probability $1-o(1)$ there exists a vertex $u_0$ and  $k$ nodes $\{u_1, u_2 ,\ldots, u_k\}$ to the right of $u_0$ such that $d(u_0,u_i) \in [\frac{(i(a-b)-2i\epsilon)\log n}{n},\frac{(i(a-b)-(2i-1)\epsilon) \log n}{n}]$ and  $k$ nodes $\{v_1, v_2 ,\ldots, v_k\}$ to the right of $u_0$ such that $d(u_0,v_i) \in [\frac{((i(a-b)+b-(2i-1)\epsilon)\log n}{n},\frac{(i(a-b)+b -(2i-2)\epsilon)\log n}{n}]$, for $i =1,2,\ldots,k$. The arrangement of the vertices is shown  in Figure~\ref{fig:arr}.
\end{lemma*}

With  the help of these two lemmas, we are in a position to  prove Theorem~\ref{thm:upper}. The proof of  the two lemmas are given immediately after the proof of the theorem.

\begin{proof}[Proof of Theorem~\ref{thm:upper}]
We have shown that the two events mentioned in Lemmas \ref{lem:lemma1} and \ref{lem:lemma2} happen with high probability. Therefore they simultaneously happen under the condition $a>1$ and $a-b >0.5$.
Now we will show that these events together imply that the graph is connected. To see this, consider the vertices $u_0,\{u_1, u_2, \ldots, u_k\}$ and $\{v_1, v_2, \ldots, v_k\}$ that satisfy the conditions of Lemma \ref{lem:lemma2}. We can observe that each vertex $v_i$ has an edge with $u_i$ and $u_{i-1}$, $i =1, \ldots,k$. This is because (see Figure~\ref{fig:arr} for a depiction)
$$
d(u_i,v_i) 
 \ge \frac{((i(a-b)+b-(2i-1)\epsilon)\log n}{n} - \frac{i(a-b)-(2i-1)\epsilon) \log n}{n}= \frac{b \log n}{n} \quad \text{and}
$$ 
\begin{align*}
 d(u_i,v_i) &\le 
   \frac{i(a-b)+b -(2i-2)\epsilon\log n}{n} - \frac{(i(a-b)-2i\epsilon)\log n}{n} = \frac{(b + 2\epsilon)\log n}{n}.
\end{align*}

Similarly, 
\vspace{-10pt}
\begin{align*}
d(u_{i-1},v_i)& 
\ge \frac{((i(a-b)+b-(2i-1)\epsilon)\log n}{n} - \frac{(i-1)(a-b)-(2i-3)\epsilon) \log n}{n}\\
& = \frac{(a-2\epsilon)\log n}{n} \quad \text{and}
\end{align*}
\begin{align*}
d(u_{i-1},v_i) 
&\le \frac{i(a-b)+b -(2i-2)\epsilon\log n}{n} - \frac{((i-1)(a-b)-2(i-1)\epsilon)\log n}{n} = \frac{a\log n}{n}.
\end{align*}
 This implies that $u_0$ is connected to $u_i$ and $v_i$ for all $i=1,\dots,k$. The first event implies that 
 the connected components are cycles spanning the entire line $[0,1]$.
 Now consider two such disconnected components, one of which consists of the nodes $u_0,\{u_1, u_2, \ldots, u_k\}$ and $\{v_1, v_2, \ldots, v_k\}$. There must exist a node $t$ in the other component (cycle) such that $t$ is on the right of $u_0$ and $d(u_0,t) \equiv \frac{x \log n}{n}  \le \frac{a \log n}{n}$. If $x\leq b$, $\exists i \mid  i\leq k \text{ and } i(a-b)+b - a-(2i-2)\epsilon \le x\leq i(a-b)-(2i-1)\epsilon$ (see Figure~\ref{fig:arr1}). 
When $x\leq b$, we can calculate the distance between $t$ and $v_i$ as 
\begin{align*}
d(t,v_i)
&\ge \frac{(i(a-b)+b-(2i-1)\epsilon)\log n}{n} - \frac{(i(a-b)-(2i-1)\epsilon) \log n}{n} = \frac{b \log n}{n}
\end{align*}
and 
\begin{align*}
d(t,v_i)
&\le \frac{(i(a-b)+b -(2i-2)\epsilon)\log n}{n} - \frac{(i(a-b)+b-a-(2i-2)\epsilon)\log n}{n} = \frac{a\log n}{n}.
\end{align*}
Therefore  $t$ is connected to $v_i$ when  $x\leq b.$ If $x >b$ then $t$ is already connected to  $u_0$. Therefore the two components (cycles) in question are connected.
This is true for all cycles and hence there is only a single component in the entire graph. Indeed, if we consider the cycles to be disjoint super-nodes, then we have shown  that there must be a star configuration. 
\end{proof}

We will now provide the proof of Lemma~\ref{lem:lemma2}.

\begin{proof}[Proof of Lemma~\ref{lem:lemma2}]
Recall that we want to show that
 there exists a node $u_0$ and  $k$ nodes $\{u_1, u_2 ,\ldots, u_k\}$ to the right of $u_0$ such that $d(u_0,u_i) \in [\frac{(i(a-b)-2i\epsilon)\log n}{n},\frac{(i(a-b)-(2i-1)\epsilon) \log n}{n}]$ and exactly $k$ nodes $\{v_1, v_2 ,\ldots, v_k\}$ to the right of $u_0$ such that $d(u_0,v_i) \in [\frac{((i(a-b)+b-(2i-1)\epsilon)\log n}{n},\frac{(i(a-b)+b -(2i-2)\epsilon)\log n}{n}]$, for $i =1,2,\ldots,k$ and $\epsilon$ is a constant less than $\frac1{2k}$ (see Figure~\ref{fig:arr} for a depiction). 
 Let $A_u$ be an indicator $\{0,1\}$-random variable for every node $u$ which is $1$ if $u$ satisfies the above conditions and $0$ otherwise. We will show $\sum_{u} A_u \ge 1$ with high probability. 
 
 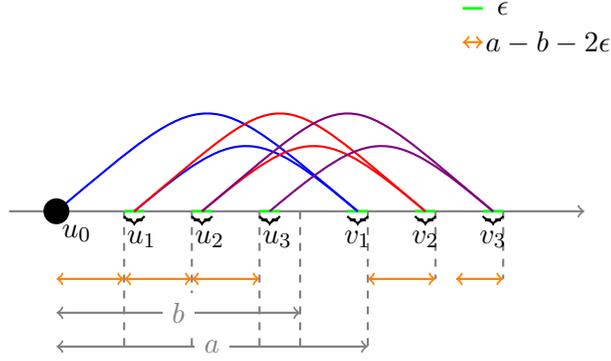
\begin{figure}
 \centering
\begin{tikzpicture}[thick, scale=0.9]

\draw [gray,->] (-1,0) -- (7.5,0);
\draw [gray,dashed] (0.7,0) -- (0.7,-2);
\draw [gray,dashed] (1.7,0) -- (1.7,-2);
\draw [gray,dashed] (2.7,0) -- (2.7,-2);
\draw [gray,dashed] (3.3,0) -- (3.3,-2);
\draw [gray,dashed] (4.3,0) -- (4.3,-2);
\draw [gray,dashed] (5.3,0) -- (5.3,-1);
\draw [gray,dashed] (6.3,0) -- (6.3,-1);

\draw [green, line width=1pt] (4,0) -- (4.3,0);
\draw [decorate,decoration={brace,amplitude=3pt,mirror},xshift=0pt,yshift=-2pt](4.,0) -- (4.3,0) node [black,midway,xshift=0cm,yshift=-0.31cm] {$v_1$};

\draw [green, line width=1pt] (5,0) -- (5.3,0);
\draw [decorate,decoration={brace,amplitude=3pt,mirror},xshift=0pt,yshift=-2pt](5,0) -- (5.3,0) node [black,midway,xshift=0cm,yshift=-0.31cm] {$v_2$};

\draw [green, line width=1pt] (6,0) -- (6.3,0);
\draw [decorate,decoration={brace,amplitude=3pt,mirror},xshift=0pt,yshift=-2pt](6,0) -- (6.3,0) node [black,midway,xshift=0cm,yshift=-0.31cm] {$v_3$};

\draw [green, line width=1pt] (0.7,0) -- (1,0);
\draw [decorate,decoration={brace,amplitude=3pt,mirror},xshift=0pt,yshift=-2pt](0.7,0) -- (1,0) node [black,midway,xshift=0.1cm,yshift=-0.31cm] {$u_1$};

\draw [green, line width=1pt] (1.7,0) -- (2,0);
\draw [decorate,decoration={brace,amplitude=3pt,mirror},xshift=0pt,yshift=-2pt](1.7,0) -- (2,0) node [black,midway,xshift=0.1cm,yshift=-0.31cm] {$u_2$};

\draw [green, line width=1pt] (2.7,0) -- (3,0);
\draw [decorate,decoration={brace,amplitude=3pt,mirror},xshift=0pt,yshift=-2pt](2.7,0) -- (3,0) node [black,midway,xshift=0.1cm,yshift=-0.31cm] {$u_3$};

\draw [orange,<->] (-0.3,-1) -- (0.7,-1);
\draw [orange,<->] (0.7,-1) -- (1.7,-1);
\draw [orange,<->] (1.7,-1) -- (2.7,-1);
\draw [orange,<->] (4.3,-1) -- (5.3,-1);
\draw [orange,<->] (5.6,-1) -- (6.3,-1);

\draw[gray,<->] (-0.3,-2) -- node[midway,fill=white] {${a}$} (4.3,-2);
\draw[gray,<->] (-0.3,-1.5) -- node[midway,fill=white] {${b}$} (3.3,-1.5);

\draw [green, line width=1pt] (5.7,3) -- (6,3)node [black,midway,xshift=0.4cm,yshift=0cm] {$\epsilon$};
\draw [orange,<->] (5.7,2.5) -- (6,2.5)node [black,midway,xshift=1cm,yshift=0cm] {$a-b-2\epsilon$};


\draw[blue]    (-0.3,0) to[out=40,in=140,distance=3cm](4.15,0) ;
\draw[blue]    (0.85,0) to[out=40,in=140,distance=2cm](4.15,0) ;
\draw[red]    (0.85,0) to[out=40,in=140,distance=3cm](5.15,0) ;
\draw[red]    (1.85,0) to[out=40,in=140,distance=2cm](5.15,0) ;
\draw[violet]    (1.85,0) to[out=40,in=140,distance=3cm](6.15,0) ;
\draw[violet]    (2.85,0) to[out=40,in=140,distance=2cm](6.15,0) ;

\filldraw [black] (-0.3,0) circle (5pt)node[anchor=south] at (0,-0.6) {$u_0$};
\end{tikzpicture}
\caption{The location of $u_i$ and $v_i$ relative to $u$ scaled by $\frac{\log n}{n}$ in Lemma \ref{lem:lemma2}. Edges stemming put of  $v_1,v_2, v_3$ are shown as blue, red and violet respectively. \label{fig:arr}}
\end{figure}

 We  have,
\begin{align*}
\Pr(A_u=1) & = n(n-1)\dots (n-(2k-1))\Big(\frac{\epsilon \log n}{n}\Big)^{2k} \Big(1-2k\epsilon \frac{\log n}{n}\Big)^{n-2k}\\
& = c_0 n^{-2k\epsilon} (\epsilon \log n)^{2k}  \prod_{i=0}^{2k-1} (1-i/n)\\
&= c_1 n^{-2k\epsilon} (\epsilon \log n)^{2k}
\end{align*}
where $c_0,c_1$ are just absolute constants independent of $n$ (recall $k$ is a constant). 
Hence,
\begin{align*}
\sum_{u} \avg A_u= c_1 n^{1-2k\epsilon} (\epsilon \log n)^{2k} \ge 1
\end{align*} 
as long as $\epsilon \leq \frac{1}{2k}$.  Now, in order to prove $\sum_{u}  A_u\ge 1$ with high probability, we will show that the variance of $\sum_{u}  A_u$  is bounded from above.  This calculation is very similar to the one in the proof of Theorem~\ref{thm:lower_bound}. 
Recall that if $A =\sum_{u}  A_u$ is a sum of indicator random variables, we must have 
$${\rm Var}(A) \le \avg[A]+\sum_{u \neq v} {\rm Cov}(A_u,A_v)=\avg[A]+\sum_{u \neq v} \Pr(A_u=1 \cap A_v=1)-\Pr(A_u=1)\Pr(A_v=1).$$ 
Now first consider the case when vertices $u$ and $v$ are at a distance of at least $\frac{2(a+b) \log n}{n}$ apart (happens with probability $1- \frac{4(a+b)\log n}{n}$). 
Then the region in $[0,1]$ that is within distance $\frac{(a+b)\log n}{n}$ from both $u$ and $v$ is the empty-set. In this case, $\Pr(A_u=1 \cap A_v=1) = n(n-1)\dots (n-(4k-1))\Big(\frac{\epsilon \log n}{n}\Big)^{4k} \Big(1-4k\epsilon \frac{\log n}{n}\Big)^{n-4k} = c_2 n^{-4k\epsilon} (\epsilon \log n)^{4k},$ where $c_2$ is a constant.

In all other cases, $\Pr(A_u=1 \cap A_v=1) \le \Pr(A_u =1)$.
Therefore,
\begin{align*}
\Pr(A_u=1 \cap A_v=1)\leq \Big(1-\frac{4(a+b) \log n}{n}\Big) c_2 n^{-4k\epsilon} (\epsilon \log n)^{4k}+ \frac{4(a+b) \log n}{n} c_1n^{-2k\epsilon} (\epsilon \log n)^{2k} 
\end{align*}
and
\begin{align*}
{\rm Var(A)} &\le c_1n^{1-2k\epsilon} (\epsilon \log n)^{2k} +{n \choose 2}\Big(\Pr(A_u=1 \cap A_v=1)-\Pr(A_u=1)\Pr(A_v=1)\Big) \\
&\le c_1n^{1-2k\epsilon} (\epsilon \log n)^{2k}+ c_3 n^{1-2k\epsilon} (\log n)^{2k+1}\\
& \le c_4  n^{1-2k\epsilon} (\log n)^{2k+1}
\end{align*}
where $c_3,c_4$ are constants. Again invoking Chebyshev's inequality, with probability at least $1-\frac{1}{\log n}$ 
$$
A > c_1n^{1-2k\epsilon} (\epsilon \log n)^{2k} - \sqrt{c_4  n^{1-2k\epsilon} (\log n)^{2k+2}}.
$$  

\end{proof}

\begin{figure}
\centering
\begin{tikzpicture}[thick, scale=0.9]

\draw [gray,->] (-1,0) -- (10.5,0);
\draw [gray,dashed] (0.7,0.3) -- (0.7,-2);
\draw [gray,dashed] (1.7,0.7) -- (1.7,-2);
\draw [gray,dashed] (2,0.35) -- (2,0);
\draw [gray,dashed] (3,0.7) -- (3,0);
\draw [gray,dashed] (2.7,0) -- (2.7,-2);
\draw [gray,dashed] (3.3,0) -- (3.3,-2);
\draw [gray,dashed] (4.3,0) -- (4.3,-2);
\draw [gray,dashed] (5.3,0) -- (5.3,-1);
\draw [gray,dashed] (6.3,0) -- (6.3,-1);

\draw [blue, line width=1pt] (4,0) -- (4.3,0);
\draw [decorate,decoration={brace,amplitude=3pt,mirror},xshift=0pt,yshift=-2pt](4.,0) -- (4.3,0) node [black,midway,xshift=0cm,yshift=-0.31cm] {$v_1$};

\draw [blue, line width=1pt] (5,0) -- (5.3,0);
\draw [decorate,decoration={brace,amplitude=3pt,mirror},xshift=0pt,yshift=-2pt](5,0) -- (5.3,0) node [black,midway,xshift=0cm,yshift=-0.31cm] {$v_2$};

\draw [blue, line width=1pt] (6,0) -- (6.3,0);
\draw [decorate,decoration={brace,amplitude=3pt,mirror},xshift=0pt,yshift=-2pt](6,0) -- (6.3,0) node [black,midway,xshift=0cm,yshift=-0.31cm] {$v_3$};

\draw [blue, line width=1pt] (0.7,0) -- (1,0);
\draw [blue,decorate,decoration={brace,amplitude=3pt},xshift=0pt,yshift=-2pt](-0.3,0.1) -- (1,0.1);
\draw [red,decorate,decoration={brace,amplitude=3pt},xshift=0pt,yshift=-2pt](.7,0.4) -- (2,0.4);
\draw [violet,decorate,decoration={brace,amplitude=3pt},xshift=0pt,yshift=-2pt](1.7,0.8) -- (3,0.8);

\draw [blue, line width=1pt] (1.7,0) -- (2,0);

\draw [blue, line width=1pt] (2.7,0) -- (3,0);

\draw [orange,<->] (-0.3,-1) -- (0.7,-1);
\draw [orange,<->] (0.7,-1) -- (1.7,-1);
\draw [orange,<->] (1.7,-1) -- (2.7,-1);
\draw [orange,<->] (4.3,-1) -- (5.3,-1);
\draw [orange,<->] (5.6,-1) -- (6.3,-1);

\draw[gray,<->] (-0.3,-2) -- node[midway,fill=white] {${a}$} (4.3,-2);
\draw[gray,<->] (-0.3,-1.5) -- node[midway,fill=white] {${b}$} (3.3,-1.5);

\draw [blue, line width=1pt] (5.7,3) -- (6,3)node [black,midway,xshift=0.4cm,yshift=0cm] {$\epsilon$};
\draw [orange,<->] (5.7,2.5) -- (6,2.5)node [black,midway,xshift=1cm,yshift=0cm] {$a-b-2\epsilon$};


\draw[blue]    (0.35,0.17) to[out=40,in=140,distance=2cm](4.15,0) ;
\draw[red]    (1.3,0.4) to[out=50,in=140,distance=3cm](5.15,0) ;

\draw[violet]    (2.3,0.8) to[out=40,in=140,distance=2cm](6.15,0) ;

\filldraw [black] (-0.3,0) circle (5pt)node[anchor=south] at (-0.5,-0.6) {$u$};
\end{tikzpicture}
\caption{The line segments where $v_1,v_2,v_3$ can have neighbors (scaled by $\frac{\log n}{n}$) in the proof of Theorem \ref{thm:upper}. The point $t$ has to lie in one of these regions.\label{fig:arr1d}}
\end{figure}
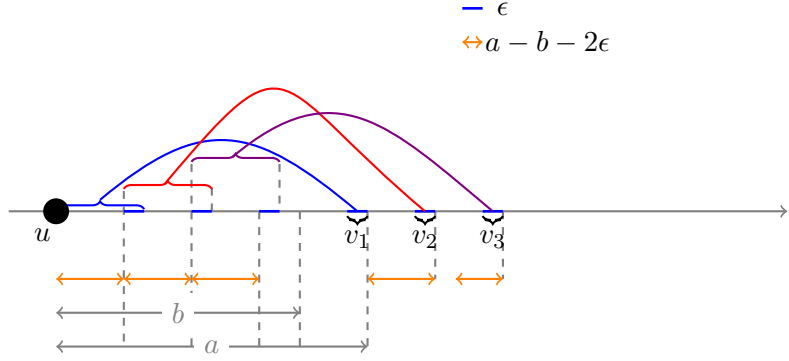

It remains to prove Lemma~\ref{lem:lemma1}. 
 
 \begin{proof}[Proof of Lemma~\ref{lem:lemma1}]
 The proof of this lemma is somewhat easily explained if we consider a weaker result (a stronger condition) with $a -b >2/3$. Let us first briefly describe this case.

Consider a node $u$ and assume without loss of generality that the position of $u$ is $0$ (i.e. $X_u=0$). Associate four indicator $\{0,1\}$-random variables $A_u^i, i =1,2,3,4$ which take the value of $1$ if and only if there does not exist any node $x$ such that 
\begin{enumerate}
\item $d(u,x) \in [b\frac{\log n}{n},a\frac{\log n}{n}] \cup [0,\frac{a-b}{2}\frac{\log n}{n}]\} \text{  for  } i=1$
\item $d(u,x) \in [b\frac{\log n}{n},a\frac{\log n}{n}] \cup [\frac{-a-b}{2}\frac{\log n}{n},-b\frac{\log n}{n}]\}  \text{  for  } i=2$
\item $d(u,x) \in [-a\frac{\log n}{n},-b\frac{\log n}{n}] \cup [\frac{-a+b}{2}\frac{\log n}{n},0]\}  \text{  for  } i=3$
\item $d(u,x) \in [-a\frac{\log n}{n},-b\frac{\log n}{n}] \cup [b\frac{\log n}{n},\frac{a+b}{2}\frac{\log n}{n}]\}  \text{  for  } i=4.$
\end{enumerate}
The intervals representing these random variables are shown in Figure~\ref{fig:four}.

Notice that $\Pr(A^i_u=1)=\max\{\Big(1-1.5(a-b)\frac{\log n}{n}\Big)^{n-1},\Big(1-a\frac{\log n}{n}\Big)^{n-1}\}$ and therefore $\sum_{i,u} \avg A^i_u \approx 4\max\{n^{1-1.5(a-b)},n^{1-a}\}$.  This means that for $a-b \ge 0.67$ and $a \ge 1$, $\sum_{i,u} \avg A^i_u=o(1)$.  Hence there exist vertices in all the regions described above for every node $u$ with high  probability.
 
Now, $A_u^1$ and $A_u^2$ being zero implies that either there is a vertex in $ [b\frac{\log n}{n},a\frac{\log n}{n}] $ or there exists two vertices $v_1,v_2$ in $ [0,\frac{a-b}{2}\frac{\log n}{n}]$ and $[\frac{-a-b}{2}\frac{\log n}{n},-b\frac{\log n}{n}]$ respectively (see, Figure~\ref{fig:four}). In the second case, $u$ is connected to $v_2$ and $v_2$ is connected to $v_1$. Therefore $u$ has nodes on left ($v_2$) and right ($v_1$) and $u$ is connected to both of them through one hop in the graph. 

Similarly, $A_u^3$ and $A_u^4$ being zero implies that either there exists a vertex in $[-a\frac{\log n}{n},-b\frac{\log n}{n}]$ or again $u$ will have vertices on left and right and will be connected to them. So, when all the four $A_u^i, i=1,2,3,4$ are zero together:
{
\begin{itemize}
    \item $A_u^1=A_u^2=0$ implies there is a neighbor of $u$ on either sides or there is a single node in $[b\frac{\log n}{n},a\frac{\log n}{n}]$
    \item $A_u^3=A_u^4=0$ implies there is a neighbor of $u$ on either sides or there is a single node in $[-a\frac{\log n}{n},-b\frac{\log n}{n}]$
\end{itemize}
This shows that when $A_u^1=A_u^2=0$ and  $A_u^3=A_u^4=0$ guarantee a node on only one side of $u$, there are nodes in $[b\frac{\log n}{n},a\frac{\log n}{n}] $ and $[-a\frac{\log n}{n},-b\frac{\log n}{n}]$.} But in that case $u$ has direct neighbors on both its left and right. 
We can conclude that every vertex $u$ is connected to  a vertex  $v$ on  its right and a vertex  $w$ on its left  such that $d(u,v) \in [0,a \frac{\log n}{n}]$ and  $d(u,w) \in [-a\frac{\log n}{n},0]$; therefore every vertex is part of a cycle that covers $[0,1]$.

 \begin{figure}
   \centering
\begin{tikzpicture}[thick, scale=0.9]
\draw [gray,<->] (0,1.5) -- (0.5,1.5)node at (1,1.5) {$\frac{a-b}{2}$};;
\draw[pattern=north west lines,pattern color=red][preaction={fill=blue!20}] (3.5,1.9) -- (4.5,1.9) -- (4.5,1.7) -- (3.5,1.7) -- cycle node at (4.8,1.8) {$A_u^3$};;
\draw[pattern=vertical lines,pattern color=red][preaction={fill=blue!50}]  (3.5,1.4) -- (4.5,1.4) -- (4.5,1.2) -- (3.5,1.2) -- cycle node at (4.8,1.2) {$A_u^4$};;
\draw[pattern=dots,pattern color=blue] [preaction={fill=gray}]  (1.5,1.9) -- (2.5,1.9) -- (2.5,1.7) -- (1.5,1.7) -- cycle node at (2.8,1.8) {$A_u^1$};;
\draw[pattern=grid,pattern color=violet][preaction={fill=red!50}]   (1.5,1.4) -- (2.5,1.4) -- (2.5,1.2) -- (1.5,1.2) -- cycle node at (2.8,1.2) {$A_u^2$};;

\draw[pattern=dots,pattern color=blue][preaction={fill=gray}]  (3.5,0) -- (4.5,0) -- (4.5,-0.3) -- (3.5,-0.3) -- cycle;
\draw[pattern=grid,pattern color=violet] [preaction={fill=red!50}]  (3.5,-0.3) -- (4.5,-0.3) -- (4.5,-0.6) -- (3.5,-0.6) -- cycle;
\draw[pattern=north west lines,pattern color=red][preaction={fill=blue!20}]  (-3.5,0) -- (-4.5,0) -- (-4.5,-0.3) -- (-3.5,-0.3) -- cycle;
\draw[pattern=vertical lines,pattern color=red][preaction={fill=blue!50}]   (-3.5,-0.3) -- (-4.5,-0.3) -- (-4.5,-0.6) -- (-3.5,-0.6) -- cycle;
\draw[pattern=dots,pattern color=blue][preaction={fill=gray}]  (0,0) -- (.5,0) -- (.5,-0.3) -- (0,-0.3) -- cycle;
\draw[pattern=north west lines,pattern color=red][preaction={fill=blue!20}]  (0,0) -- (-.5,0) -- (-.5,-0.3) -- (0,-0.3) -- cycle;
\draw[pattern=vertical lines,pattern color=red][preaction={fill=blue!50}]   (3.5,-0.6) -- (4,-0.6) -- (4,-0.9) -- (3.5,-0.9) -- cycle;
\draw[pattern=grid,pattern color=violet] [preaction={fill=red!50}]  (-3.5,0) -- (-3,0) -- (-3,-0.3) -- (-3.5,-0.3) -- cycle;

\draw [gray,<->] (3.5,-1.3) -- (4,-1.3);
\draw [gray,<->] (0,-1.3) -- (0.5,-1.3);
\draw [gray,<->] (0,-1.3) -- (-.5,-1.3);
\draw [gray,<->] (-3.5,-1.3) -- (-3,-1.3);

\draw [gray,<->] (-5,0) -- (5,0);
\draw [gray,dashed] (3.5,0) -- (3.5,-1.4);
\draw [gray,dashed] (0.5,0) -- (0.5,-1.4);
\draw [gray,dashed] (0,0) -- (0,-1.4);
\draw [gray,dashed] (-0.5,0) -- (-0.5,-1.4);
\draw [gray,dashed] (-3.5,0) -- (-3.5,-1.4);
\draw [gray,dashed] (-3,0) -- (-3,-1.4);
\draw [gray,dashed] (4,0) -- (4,-1.4);
\draw node[anchor=south] at (3.5,0) {$b$};
\draw node[anchor=south] at (4.5,0) {$a$};
\draw node[anchor=south] at (-3.5,0) {$-b$};
\draw node[anchor=south] at (-4.5,0) {$-a$};

%
%
%
%
\filldraw [black] (0,0) circle (3pt)node[anchor=south] at (0,0) {$u$};
\end{tikzpicture}
\caption{Representation of four different random variables for Lemma \ref{lem:lemma1}. \label{fig:four}}
\end{figure}

We can  now extend this proof to the case when $a-b> 0.5.$

Let $c$ be large number to  be chosen specifically later.
Consider a node $u$ and assume that the position of $u$ is $0$. Now consider the {four different regions $[-a\frac{\log n}{n},-b\frac{\log n}{n}]$, $[-(a-b)\frac{\log n}{n},0]$, $[b\frac{\log n}{n},a\frac{\log n}{n}]$ and $[0,a-b\frac{\log n}{n}]$ around $u$ each divided into $L\equiv 2^{c}$ patches (intervals) of size $\theta=\frac{a-b}{2^c}$  in  the following way}:
\begin{enumerate}
\item $I_u^i=[\frac{(-a + (i-1)\theta )\log n}{n},\frac{(-a + i\theta)  \log n}{n}]$
\item $J_u^i=[\frac{(-(a-b) + (i-1)\theta )\log n}{n},\frac{(-(a-b) + i\theta)  \log n}{n}]$
\item $K_u^i= [\frac{(b + (i-1)\theta )\log n}{n},\frac{(b + i\theta)  \log n}{n}]$
\item $M_u^i= [\frac{( (i-1)\theta )\log n}{n},\frac{ i\theta  \log n}{n}]$
\end{enumerate}
where $i=1,2,3,\ldots,L$. Note that any vertex in $\cup I_u^i \cup K_u^i$ is connected to $u$. See, Figure~\ref{fig:patch} for a depiction.

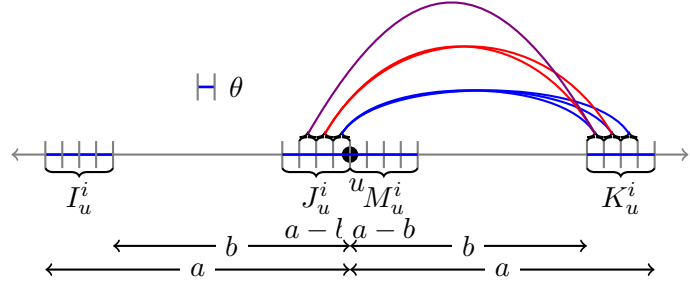
\begin{figure}
\centering
\begin{tikzpicture}[thick, scale=0.9]
\draw [gray,<->] (-5,0) -- (5,0);\filldraw [black] (0,0) circle (3pt)node[anchor=south] at (0.1,-0.7) {$u$};
\draw[<->] (0,-1.1) -- node[midway,fill=white] {$a-b$} (-1,-1.1);
\draw[<->] (0,-1.1) -- node[midway,fill=white] {$a-b$} (1,-1.1);
\draw[<->] (0,-1.35) -- node[midway,fill=white] {$b$} (-3.5,-1.35);
\draw[<->] (0,-1.7) -- node[midway,fill=white] {$a$} (-4.5,-1.7);
\draw[<->] (0,-1.35) -- node[midway,fill=white] {$b$} (3.5,-1.35);
\draw[<->] (0,-1.7) -- node[midway,fill=white] {$a$} (4.5,-1.7);

\draw [blue, line width=1pt] (-4.5,0) -- (-4.25,0);
\draw [blue, line width=1pt] (-4.25,0) -- (-4,0);
\draw [blue, line width=1pt] (-4,0) -- (-3.75,0);
\draw [blue, line width=1pt] (-3.75,0) -- (-3.5,0);

\draw [blue, line width=1pt] (-2.00,1) -- (-2.25,1) node [black,midway,xshift=0.4cm,yshift=0cm] {$\theta$};;
\draw [gray] (-2.00,1.2) -- (-2.00,0.8);
\draw [gray] (-2.25,1.2) -- (-2.25,0.8);

\draw [gray] (4.5,0.2) -- (4.5,-.2);
\draw [gray] (4.25,0.2) -- (4.25,-.2);
\draw [gray] (4,0.2) -- (4,-.2);
\draw [gray] (3.75,0.2) -- (3.75,-.2);
\draw [gray] (3.5,0.2) -- (3.5,-.2);

\draw [gray] (-4.5,0.2) -- (-4.5,-.2);
\draw [gray] (-4.25,0.2) -- (-4.25,-.2);
\draw [gray] (-4,0.2) -- (-4,-.2);
\draw [gray] (-3.75,0.2) -- (-3.75,-.2);
\draw [gray] (-3.5,0.2) -- (-3.5,-.2);

\draw [gray] (-1,0.2) -- (-1,-.2);
\draw [gray] (-.75,0.2) -- (-.75,-.2);
\draw [gray] (-0.5,0.2) -- (-0.5,-.2);
\draw [gray] (-.25,0.2) -- (-.25,-.2);
\draw [gray] (0,0.2) -- (0,-.2);
\draw [gray] (1,0.2) -- (1,-.2);
\draw [gray] (.75,0.2) -- (.75,-.2);
\draw [gray] (0.5,0.2) -- (0.5,-.2);
\draw [gray] (.25,0.2) -- (.25,-.2);
\draw [decorate,decoration={brace,amplitude=3pt,mirror},xshift=0pt,yshift=0pt](3.5,-0.2) -- (4.5,-0.2) node [black,midway,xshift=0cm,yshift=-0.4cm] {$K_u^i$};
\draw [decorate,decoration={brace,amplitude=3pt},xshift=0pt,yshift=0pt](0,-0.2) -- (-1,-0.2) node [black,midway,xshift=0cm,yshift=-0.4cm] {$J_u^i$};
\draw [decorate,decoration={brace,amplitude=3pt,mirror},xshift=0pt,yshift=0pt](0,-0.2) -- (1,-0.2) node [black,midway,xshift=0cm,yshift=-0.4cm] {$M_u^i$};

\draw [decorate,decoration={brace,amplitude=3pt},xshift=0pt,yshift=0pt](-3.5,-0.2) -- (-4.5,-0.2) node [black,midway,xshift=0cm,yshift=-0.4cm] {$I_u^i$};

\draw [blue, line width=1pt] (-1,0) -- (-0.75,0);
\draw [blue, line width=1pt] (-0.75,0) -- (-0.5,0);
\draw [blue, line width=1pt] (-0.5,0) -- (-0.25,0);
\draw [blue, line width=1pt] (-0.25,0) -- (0,0);

\draw [blue, line width=1pt] (1,0) -- (0.75,0);
\draw [blue, line width=1pt] (0.75,0) -- (0.5,0);
\draw [blue, line width=1pt] (0.5,0) -- (0.25,0);
\draw [blue, line width=1pt] (0.25,0) -- (0,0);

\draw [blue, line width=1pt] (4.5,0) -- (4.25,0);
\draw [blue, line width=1pt] (4.25,0) -- (4,0);
\draw [blue, line width=1pt] (4,0) -- (3.75,0);
\draw [blue, line width=1pt] (3.75,0) -- (3.5,0);

\draw [decorate,decoration={brace,amplitude=3pt},xshift=0pt,yshift=0pt](3.5,0.2) -- (3.75,0.2);
\draw [decorate,decoration={brace,amplitude=3pt},xshift=0pt,yshift=0pt](3.75,0.2) -- (4,0.2);
\draw [decorate,decoration={brace,amplitude=3pt},xshift=0pt,yshift=0pt](4,0.2) -- (4.25,0.2);
\draw [decorate,decoration={brace,amplitude=3pt},xshift=0pt,yshift=0pt](-0.75,0.2) -- (-0.5,0.2);
\draw [decorate,decoration={brace,amplitude=3pt},xshift=0pt,yshift=0pt](-0.5,0.2) -- (-0.25,0.2);
\draw [decorate,decoration={brace,amplitude=3pt},xshift=0pt,yshift=0pt](-0.25,0.2) -- (0,0.2);

\draw[blue]    (-0.125,0.3) to[out=60,in=120,distance=1cm](3.625,0.3) ;
\draw[blue]    (-0.125,0.3) to[out=60,in=120,distance=1cm](3.875,0.3) ;
\draw[blue]    (-0.125,0.3) to[out=60,in=120,distance=1cm](4.125,0.3) ;

\draw[red]    (-0.375,0.3) to[out=60,in=120,distance=2cm](3.875,0.3) ;
\draw[red]    (-0.375,0.3) to[out=60,in=120,distance=2cm](3.625,0.3) ;

\draw[violet]    (-0.625,0.3) to[out=60,in=120,distance=3cm](3.625,0.3) ;

\end{tikzpicture}
\caption{Pictorial representation of $I_u^i, J_u^i, K_u^i, M_u^i$ and their connectivity as described in Lemma \ref{lem:lemma1}. The colored lines show the regions that are connected to each other.\label{fig:patch}}
\end{figure}

Consider a $\{0,1\}$-indicator random variable $X_u$ that is $1$ if and only if there {does not exist any node in a region formed by union of any $2L-1$ patches amongst the ones described above.} Notice that { when $a<2b$, the patches do not overlap and the total size of $2L-1$ patches is $\frac{2^{c+1}-1}{2^c}\frac{(a-b)\log n}{n}$ and when $a\geq 2b$, the patches can overlap and the total size of the $2L-1$ patches is going to be more than $\min\{\frac{2^{c+1}-1}{2^c}\frac{(a-b)\log n}{n},\frac{a \log n}{n}\}$. }
Since there are ${4L \choose 2L-1}\le n^{\frac{4L}{\log n}}$ possible regions that consists of $2L-1$ patches, 
\begin{align*}
 \sum_u \avg X_u &\leq n {4L \choose 2L-1}\Big(1-\min\{\frac{2^{c+1}-1}{2^c}\frac{(a-b)\log n}{n},\frac{a \log n}{n}\}\Big)^{n-1}\\
 & \le \max\{n^{1-\frac{2^{c+1}-1}{2^c}(a-b) +\frac{4L}{\log n}}, n^{1-a+\frac{4L}{\log n}}\}.
\end{align*}
At this point we can choose $c= c_n = o(\log n)$ such that $\lim_n c_n = \infty$.
 Hence when $a-b > \frac12$ 
 and $a > 1$, for every vertex $u$ there exists at least one patch amongst every $2L-1$ patches in $\cup I_u^i \cup J_u^j \cup K_u^k, i,j,k =1,2, \dots, L$ that contains a vertex.

Consider a collection of patches $\cup_iI_u^i \cup_j J_u^j,  i,j=1, 2, \dots, L$. We know that there exist two patches amongst these $I_u^i$s and $K_u^j$s that contain at least one vertices. If one of $I_u^i$s and one of $K_u^j$s contain two vertices, we found one neighbor of $u$ on both left and right directions (see, Figure~\ref{fig:patch}).

We consider the other case now.
Without loss of generality assume that there are no vertex in all $I_u^i$s and there exist at least two patches in $K_u^i$s that contain at least one vertex each. Hence, there exists at least one of $\{K_u^i \mid i\in \{1,2, \ldots, L-1\}\}$ that contains a vertex. Similarly, we can also conclude in this case that there exists at least one of $\{J_u^i \mid i\in \{2,3 \ldots, L\}\}$ which contain a node. Assume $J_u^{\phi}$ to be the left most patch in $\cup J_u^i \mid i \in \{1, 2, \dots, L\}$ that contains a vertex (see, Figure~\ref{fig:patch}) . From our previous observation, we can conclude that $\phi \ge 2$.

We can observe that any vertex in $J_u^j$ is connected to the vertices in patches $K_u^k,  \forall k<j$. 
This is because for two vertices $v\in J_u^j$ and $w\in K_u^k$, we have
\begin{align*}
d(v,w) &\ge \frac{(b + (k-1)\theta )\log n}{n} - \frac{(-(a-b) + j\theta)  \log n}{n}
= \frac{(a + (k-j-1)\theta )\log n}{n}; \\
d(v,w) &\le \frac{(b + k\theta)  \log n}{n} - \frac{(-(a-b) + (j-1)\theta )\log n}{n}
= \frac{(a + (k-j+1)\theta )\log n}{n}. 
\end{align*}


Consider a collection of $2L-1$ patches $\{\cup I_u^i  \cup J_u^j \cup K_u^k \mid i,j,k\in \{1,\dots, L\}, j>\phi, k\le \phi-1\}$ where $\phi \ge 2$. This is a collection of $2L-1$ patches out of which one must have a vertex and since none of $\{J_u^j \mid j>\phi\}$ and $I_u^i$ can contain a vertex, one of $\{K_u^k \mid k \le \phi -1\}$  must contain the vertex. Recall that the vertex in $J_u^{\phi}$ is connected to any node in $K_u^k$ for any $k\le \phi-1$ and therefore $u$ has a node to the right direction and left direction that are connected to $u$. Therefore every vertex is part of a cycle and each of the circles  covers $[0,1]$.
 \end{proof}

The following result is an immediate corollary of the connectivity upper bound.
\begin{corollary}\label{cor:patch}
Consider a random  graph $G(V,E)$  is being  generated as a variant of the VRG where each  $u,v\in V$ forms an edge if and only if  $d(u,v)\in \left[0,c\frac{\log n}{n}\right]\cup \left[b\frac{\log n}{n},a\frac{\log n}{n}\right], 0 <c<b<a$. This graph is connected with  probability $1-o(1)$ if $a-b + c>1$ or if $a-b > 0.5, a>1$. \label{cor:patches}
\end{corollary}

The above corollary can be further improved for some regimes of $a,b,c$. In particular, we can get the following result (proof delegated to appendix).

\begin{corollary}\label{cor:extra}
Consider a random  graph $G(V,E)$  is being  generated as a variant of the VRG where each  $u,v\in V$ forms an edge if and only if  $d(u,v)\in \left[0,c\frac{\log n}{n}\right]\cup \left[b\frac{\log n}{n},a\frac{\log n}{n}\right], 0 <c<b<a$. This graph is connected with  probability $1-o(1)$ if
either of the following conditions are true:
\begin{enumerate}
\item $2(a-b)+ c/2 > 1 \text{ when } a-b<c \text{ and }b>3c/2$
\item $b-c > 1 \text{ when } a-b<c \text{ and }b\le 3c/2$
\item $a > 1 \text{ when } a-b\ge c \text{ and } b\le 3c/2 $
\item $(a-b)+ 3c/2  > 1 \text{ when } a-b\ge c \text{ and } b>3c/2$.
\end{enumerate}
\end{corollary}


 \section{Connectivity of High Dimensional Random Annulus Graphs: Detailed Proofs of Theorems \ref{th:lb} and \ref{thm:highdem1}}
\label{sec:hrag-detail}
In this section we first prove an impossibility result on the connectivity of random annulus graphs in $t$ dimensions by showing a sufficient condition of existence of isolated nodes.
 Next, we show that if the  gap between $r_1$ and $r_2$ is large enough then the RAG is fully connected. We will start by introducing a few notations. Let us define the regions $B_t(u,r)$ and $B_t(u,[r_1,r_2])$ for the any  $u\in S^t$ in the following way:
\begin{align*}
B_t(u,r)&=\{x \in S^t \mid \|u-x\|_2 \le r \} \\
B_t(u,[r_1,r_2])&=\{x \in S^t \mid  r_1 \le \|u-x\|_2 \le r_2 \}.
\end{align*}
First, we calculate  $|B_t(u,r)|$ and show that it is proportional to $r^t$. 
\begin{lemma}
$|B_t(u,r)|=c_tr^t$ for $r=o(1)$ where $c_t \approx \frac{\pi^{t/2}}{\Gamma (\frac{t}{2}+1) }$.
\label{lem:area}
\end{lemma}
\begin{proof}
We use the following fact from (\cite{larsen2017improved,li2011concise}) for the proof. For a $t$-dimensional unit sphere, the hyperspherical cap of angular radius $\theta= \max_{x \in B_t(u,r)} \arccos \langle x,u \rangle$ has a surface area $C_t(\theta)$ given by 
\begin{align*}
C_t(\theta)=\int_{0}^{\tan \theta} \frac{S_{t-1}(r)}{(1+r^{2})^{2}}dr 
\end{align*}
where $S_{t-1}(\theta)=\frac{t\pi^{t/2}}{\Gamma (\frac{t}{2}+1) }\theta^{t-1}$. Note that $C_t(\theta)$ is nothing but $|B_t(u,r)|$ where $\cos \theta=1-\frac{r^{2}}{2}$ and therefore $\tan \theta=\frac{r \sqrt{4-r^{2}}}{2-r^{2}} \approx r$ for small $r$. Now since $r=o(1)$ and $1+r^{2}$ is an increasing function of $r$, we must have that 
\begin{align*}
\int_{0}^{r} \frac{t\pi^{t/2}}{(1+o(1))\Gamma (\frac{t}{2}+1) }\theta^{t-1} d\theta  < C_t(\theta) <  \int_{0}^{r} \frac{t\pi^{t/2}}{\Gamma (\frac{t}{2}+1) }\theta^{t-1} d\theta
\end{align*}
and therefore $C_t(\theta)$ can be expressed as $c_t r^t$ where $c_t$ lies in $\Big(\frac{\pi^{t/2}}{(1+o(1))\Gamma (\frac{t}{2}+1) },\frac{\pi^{t/2}}{\Gamma (\frac{t}{2}+1) } \Big)$.
\end{proof}

\subsection{Impossibility Result}
The following theorem proves the impossibility result for the connectivity of a random annulus graph by proving a tight threshold for the presence of an isolated node with high probability.
\begin{theorem*}(\ref{th:lb})
For a random annulus graph ${\rm RAG}_t(n,[r_1,r_2])$ where $r_1=b \Big(\frac{ \log n}{n}\Big)^{\frac{1}{t}}$ and $r_2=a\Big(\frac{\log n}{n}\Big)^{\frac{1}{t}}$, there exists isolated nodes with high probability if and only if 
$a^t -b^t < \frac{\sqrt{\pi}(t+1)\Gamma(\frac{t+2}{2})}{\Gamma(\frac{t+3}{2})}$.
\end{theorem*}
\begin{proof}
Consider the random annulus graph ${\rm RAG}_t(n,[r_1,r_2])$ in $t$ dimensions. In this graph, a node $u$ is isolated if there are no nodes $v$ such that $ r_1 \le \|u-v\|_2 \le r_2$. Since all nodes are uniformly and randomly distributed on $S^t$, the probability of a node $v$ being connected to a node $u$ is the volume of $B_t(u,[r_1,r_2])$. Define the indicator random variable $A_u \in \{0,1\}$ which is $1$ if and only if the node $u$ is isolated. Also define the random variable $A=\sum_u A_u$ which denotes the total number of isolated nodes. Since $|B_t(u,[r_1,r_2])|=c_t(r_2^t-r_1^t)$, we must have
\begin{align*}
\Pr(A_u=1)=\Bigg(1-\frac{c_t\Big(r_2^t-r_1^t \Big)}{|S^t|}\Bigg)^{n-1}.
\end{align*}
Now, we know from (\cite{surfacensphere}) that  
$|S^t|=\frac{(t+1)\pi^{\frac{t+1}{2}}}{\Gamma(\frac{t+3}{2})}$. Plugging in, we get that $\frac{c_t}{|S^t|}=\frac{\Gamma(\frac{t+3}{2})}{\sqrt{\pi}(t+1)\Gamma(\frac{t+2}{2})}$.
Hence, the expected number of isolated nodes $\avg A$ is going to be 
\begin{align*}
n\Bigg(1-\Big(a^t-b^t\Big)\frac{c_t}{|S^t|}\frac{\log n}{n}\Bigg)^{n-1} \approx n^{1- \frac{c_t(a^t-b^t)}{|S^t|}}.
\end{align*}
Therefore $\avg [A] \ge 1$ if $a^t-b^t <\frac{|S^t|}{c_t}$. In order to show that $A =\omega(1)$ with high probability we are going to show that the variance of $A$ is bounded from above. Since $A$ is a sum of indicator random variables, we have that 

$${\rm Var}(A) \le \avg[A]+\sum_{u \neq v} {\rm Cov}(A_u,A_v)=\avg[A]+\sum_{u \neq v} (\Pr(A_u=1 \cap A_v=1)-\Pr(A_u=1)\Pr(A_v=1)).$$ 
Now, consider the scenario when the vertices $u$ and $v$ are at a distance more than $2r_2$ apart which happens with probability $1-\frac{c_t(2r_2)^t}{|S^t|}$. Then the region  in which every point is within a distance of $r_2$ and $r_1$ from both $u,v$ is empty and therefore $\Pr(A_u=1 \cap A_v=1) =  
\Bigg(1-\frac{2c_t}{|S^t|}\Big(a^t-b^t\Big)\frac{\log n}{n}\Bigg)^{n-2}$.
When the vertices  are within distance $2r_2$ of one another, then
$
\Pr(A_u=1 \cap A_v=1) \le \Pr(A_u=1).
$
Therefore,
\begin{align*}
\Pr(A_u=1 \cap A_v=1) \le &\Big(1-\frac{c_t(2r_2)^t}{|S^t|} \Big) \Big(1- \frac{2c_t}{|S^t|}\Big(a^t-b^t\Big)\frac{\log n}{n}\Big)^{n-2} + \frac{c_t(2r_2)^t}{|S^t|} \Pr(A_u=1)\\ 
&\le  (1-\frac{c_t(2r_2)^t}{|S^t|}) n^{-\frac{2c_t}{|S^t|}(a^t-b^t)+o(1)}+ \frac{c_t(2r_2)^t}{|S^t|} n^{-\frac{c_t(a^t-b^t)}{|S^t|}+o(1)}.
\end{align*}
Consequently for large enough $n$,
\begin{align*}
\Pr(A_u=1 \cap A_v=1)-\Pr(A_u=1)\Pr(A_v=1) &\le (1-\frac{c_t(2r_2)^t}{|S^t|}) n^{-\frac{2c_t(a^t-b^t)}{|S^t|}+o(1)} \\+ \frac{c_t(2r_2)^t}{|S^t|} n^{-\frac{c_t(a^t-b^t)}{|S^t|}+o(1)}  -& n^{-\frac{2c_t(a^t-b^t)}{|S^t|} +o(1)} 
\le  \frac{2c_t(2r_2)^t}{|S^t|}\Pr(A_u=1).
\end{align*}
Now,
$$
{\rm Var}(A) \le \avg[A] + 2\binom{n}{2}\frac{c_t(2r_2)^t}{|S^t|}\Pr(A_u=1) \le \avg[A](1+ \frac{c_t(2a)^t}{|S^t|}  \log n).
$$
By using Chebyshev bound, with probability at least $1-O\Big(\frac{1}{\log n}\Big)$, 
$$A >n^{1-\frac{c_t(a^t-b^t)}{|S^t|}}-\sqrt{n^{1-\frac{c_t(a^t-b^t)}{|S^t|}}(1+\frac{c_t(2a)^t}{|S^t|} \log n)\log n},$$
which implies that for $a^t-b^t < \frac{|S^t|}{c_t}$, $A>1$ and hence there will exist isolated nodes with high probability.  
\end{proof}

 \subsection{Connectivity Bound}
We show the upper bound for connectivity of a Random Annulus Graphs in $t$ as per Theorem \ref{thm:highdem1}, rewritten below.

\begin{theorem*}(\ref{thm:highdem1})
For $t$ dimensional random annulus graph ${\rm RAG}_t(n,[r_1,r_2])$ where $r_2=a\Big(\frac{\log n}{n}\Big)^{t}$ and $r_1=b\Big(\frac{\log n}{n}\Big)^{t}$ with $a\ge b$ and $t$ is a constant, the graph is connected with high probability if 
\begin{align*}
a^t-b^t \ge \frac{8|S^t|(t+1)}{c_t\Big(1 -  \frac{1}{{2^{1+1/t} - 1}}\Big)} \text{  and  }  a>2^{{1}+\frac{1}{t}}b.
\end{align*} 
\end{theorem*}
Let us define a \emph{pole} to be a vertex  which is connected to all vertices within a distance of $r_2$ from itself. 
In order to prove  Theorem \ref{thm:highdem1}, we first show the existence of a pole with high probability in Lemma \ref{lem:pole}. Next, Lemma \ref{lem:high_stuff1} shows that for every vertex $u$ and every hyperplane $L$ passing through $u$ and not too close to the tangential hyperplane at $u$, there will be a neighbor of $u$ on either side of the plane. In order to formalize this, let us define a few regions associated with a node $u$ and a hyperplane $L:w^{T}x=\beta$ passing through $u$.

\begin{align*}
\mathcal{R}_{L}^1 &\equiv \{x \in S^t \mid b\Big(\frac{\log n}{n}\Big)^{1/t} \le  \|u-x\|_2 \le a\Big(\frac{\log n}{n}\Big)^{1/t}, w^{T}x \le \beta \} \\
\mathcal{R}_{L}^2 &\equiv \{x \in S^t \mid b\Big(\frac{\log n}{n}\Big)^{1/t} \le  \|u-x\|_2 \le a\Big(\frac{\log n}{n}\Big)^{1/t}, w^{T}x \ge \beta \} \\
\mathcal{A}_{L} & \equiv \{x \mid x \in \mathcal{S}^t, \quad w^{T}x=\beta \}.
\end{align*}

Informally, $\mathcal{R}_{L}^1$ and $\mathcal{R}_{L}^2$ represents the partition of the region $B_t(u,[r_1,r_2])$ on either side of the hyperplane $L$ and $\mathcal{A}_L$ represents the region on the sphere lying on $L$. 

\begin{lemma*}(\ref{lem:pole})
In  ${\rm RAG}_t\left(n, \left[b\left(\frac{\log n}{n}\right)^{1/t},a\left(\frac{\log n}{n}\right)^{1/t}\right]\right), 0 <b <a$, with  probability $1-o(1)$ there exists a pole.
\end{lemma*}

\begin{lemma*}(\ref{lem:high_stuff1})
If we sample $n$ nodes from $S^t$ according to ${\rm RAG}_t\left(n,\left[b\left(\frac{\log n}{n}\right)^{1/t},a\left(\frac{\log n}{n}\right)^{1/t}\right]\right)$, then for every node $u$ and every hyperplane $L$ passing through $u$ such that $\mathcal{A}_L \not \subset B_t(u,a\left(\frac{\log n}{n}\right)^{1/t})$, node $u$ has a neighbor on both sides of the hyperplane $L$ with probability at least $1-\frac{1}{n}$ provided 
\begin{align*}
a^t-b^t \ge \frac{8|S^t|(t+1)}{c_t\Big(1 -  \frac{1}{{2^{1+1/t} - 1}}\Big)}
\end{align*}  and $a>2^{1+\frac{1}{t}}b$.
\end{lemma*}

For a node $u \equiv (u_1,u_2,\dots,u_{t+1})$, define the particular hyperplane $L^{\star}_u : x_1=u_1$ which is normal to the line joining $u_0 \equiv (1,0,\dots,0)$ and the origin and passes through $u$. We now have the following lemma.
\begin{lemma*}(\ref{lem:cut_twice})
For a particular node $u$ and corresponding hyperplane $L^{\star}_u$, if $\mathcal{A}_{L^{\star}_u}  \subseteq B_t(u,r_2)$ then $u$ must be within $r_2$ of $u_0$.
\end{lemma*}

For now, we assume that the Lemmas \ref{lem:pole}, \ref{lem:high_stuff1} and \ref{lem:cut_twice} are true and show why these lemmas together imply the proof of Theorem \ref{thm:highdem1}.
\begin{proof}[Proof of Theorem \ref{thm:highdem1}] 
We consider an alternate (rotated but not shifted) coordinate system by multiplying every vector by a orthonormal matrix $R$ such that the new position of the pole is the $t+1$-dimensional vector $(1,0,\dots,0)$ where only the first co-ordinate is non-zero. Let the $t+1$ dimensional vector describing a node $u$ in this new coordinate system be $\hat{u}$. Now consider the hyperplane $L: x_1=\hat{u}_1$ and if $u$ is not connected to the pole already, then by Lemma \ref{lem:high_stuff1} and Lemma \ref{lem:cut_twice} the node $u$ has a neighbor $u_2$ which has a higher first coordinate. The same analysis applies for $u_2$ and hence we have a path where the first coordinate of every node is higher than the previous node. Since the number of nodes is finite, this path cannot go on indefinitely and at some point, one of the nodes is going to be within $r_2$ of the pole and will be connected to the pole. Therefore every node is going to be connected to the pole and hence our theorem is proved. 

\end{proof}
We show the proofs of Lemma \ref{lem:pole}, \ref{lem:high_stuff1} and \ref{lem:cut_twice} in the following sections.

\subsection{Proof of Lemma \ref{lem:pole}}
Lemma \ref{lem:pole_helper} is a helper lemma that shows the region of connectivity for a small ball of radius $\epsilon (\frac{\log n}{n})^{1/t}$. Lemma \ref{lem:lemma2pole} uses this  lemma  to show the existence of a point $u_0$ which is connected to various balls of radius $\epsilon (\frac{\log n}{n})^{1/t}$.
\begin{lemma}\label{lem:pole_helper}
For a $t$ dimensional random annulus graph ${\rm RAG}_t(n,[r_1,r_2])$ where $r_1=b\left(\frac{\log n}{n}\right)^{1/t}, r_2=a\left(\frac{\log n}{n}\right)^{1/t}$ and $a\ge b$, consider the region ${B_t}(O,\theta)$ centered at $O$ and radius $\theta = \epsilon \left(\frac{\log n}{n}\right)^{1/t}$. Then, every vertex in ${B_t}(O,\theta)$ is connected to all vertices present in  ${B_t}(O,[\theta_1,\theta_2])$ where $\theta_1 = (b+\epsilon)\left(\frac{\log n}{n}\right)^{1/t}$ and $\theta_2 = (a-\epsilon)\left(\frac{\log n}{n}\right)^{1/t}$.
\end{lemma}
\begin{proof}
For any point $A\in {B_t}(O,\theta)$, we have $0<\|A-O\|_2\leq \theta$ and for any point $X \in {B_t}(O,[\theta_1,\theta_2])$, we must have $\theta_1 \leq \|X-O\|_2\leq \theta_2$. Hence 
\begin{align*}
\|A-X\|_2&\leq \|A-O\|_2 + \|X-O\|_2   \\
&\leq \theta + \theta_2\\
&=  a \left(\frac{\log n}{n}\right)^{1/t},\\
\|A-X\|_2&\geq \|X-O\|_2 - \|A-O\|_2\\
&\geq \theta_1 -  \theta\\
&=  b \left(\frac{\log n}{n}\right)^{1/t},
\end{align*}
and therefore the claim of the lemma is proved.
\end{proof}

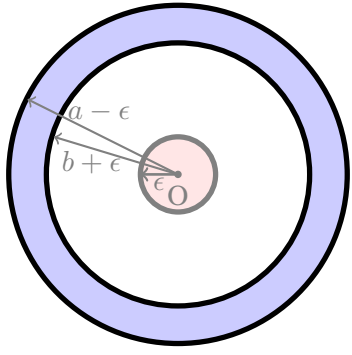
\begin{figure}
\vspace{-20pt}
\centering
\begin{tikzpicture}[thick, scale=0.5]
 \filldraw[color=black, fill=blue!20,  line width=2pt](0,0) circle (4.5);
 \filldraw[color=black, fill=red!0,  line width=2pt](0,0) circle (3.5);
 \filldraw[color=gray, fill=red!10,  line width=2pt](0,0) circle (1);
  \filldraw [gray] (0,0) circle (2pt)node[anchor=north ] at (0,0){O};
  \draw[line width=1pt,gray,->] (0,0)--(-1,0)node[anchor=north ] at (-0.5,0.2){$\epsilon$};
    \draw[gray,->] (0,0)--(-3.3,1)node[anchor=north ] at (-2.3,0.9){$b+\epsilon$};
      \draw[gray,->] (0,0)--(-4,2)node[anchor=north ] at (-2.1,2.2){$a-\epsilon$};
\end{tikzpicture}
\caption{Any node in the red region is connected to any node in the blue region.\label{fig:arr0high}}
\end{figure}


\begin{lemma}
Set two real numbers $k\equiv \lceil b/(a-b)\rceil+1$ and $\epsilon < \left( \frac{|S_t|}{2kc_t} \right)^{1/t}$. In an ${\rm RAG}_t\left(n, \left[b\left(\frac{\log n}{n}\right)^{1/t},a\left(\frac{\log n}{n}\right]\right)^{1/t}\right), 0 <b <a$, with  probability $1-o(1)$ there exists a vertex $u_0 \in S^t$  with the following property. 
Consider a homogeneous hyperplane $L$ in $\reals^{t+1}$ that pass through $u_0$.
There are $k$ nodes $\cA = \{u_1, u_2 ,\ldots, u_k\}$  with $u_i \in {B_t}\left(O_{u_i}, \epsilon \left(\frac{\log n}{n}\right)^{1/t}\right)$ for some $O_{u_i}\in L \cap S^t$ such that $\|O_{u_i}-u_0\|_2 = (i(a-b)-(4i-1)\epsilon)\left(\frac{\log n}{n}\right)^{1/t}$ and  $k$ nodes $\cB = \{v_1, v_2 ,\ldots, v_k\}$ with $v_i \in {B_t}\left(O_{v_i}, \epsilon \left(\frac{\log n}{n}\right)^{1/t}\right)$ for some $O_{v_i}\in L\cap S^t$ such that $\|O_{v_i}-u_0\|_2= (i(a-b)+b-(4i-3)\epsilon)\left(\frac{\log n}{n}\right)^{1/t}$, for $i =1,2,\ldots,k$
with $\cA$ and $\cB$ separated by $L$.
~\label{lem:lemma2pole}
\end{lemma}

\begin{proof}[Proof of Lemma~\ref{lem:lemma2pole}]
 Let $A_u$ be an indicator $\{0,1\}$-random variable for every node $u$ which is $1$ if $u$ satisfies the conditions stated in the lemma and $0$ otherwise. We will show $\sum_{u} A_u \ge 1$ with high probability.

We  have,
\begin{align*}
\Pr(A_u=1) & = \frac{1}{2k!} n(n-1)\dots (n-(2k-1))\Big(\frac{\epsilon^tc_t \log n}{n|S^t|}\Big)^{2k} \Big(1-2kc_t\epsilon^t\frac{\log n}{n|S^t|}\Big)^{n-2k}\\
& = c_1 n^{-2k\epsilon^t c_t/|S^t|} (\epsilon^t \log n)^{2k}  \prod_{i=0}^{2k-1} (1-i/n)\\
&= c_2 n^{-2k\epsilon^t c_t/|S^t|} (\epsilon^t \log n)^{2k}
\end{align*}
where $c_1=\frac{c_t^{2k}}{2k!|S^t|^{2k}},c_2$ are just absolute constants independent of $n$ (recall $k$ is a constant). 
Hence,
\begin{align*}
\sum_{u} \avg A_u= c_2 n^{1-2k\epsilon^t c_t/|S^t|} (\epsilon^t \log n)^{2k} \ge 1
\end{align*} 
as long as $\epsilon \leq \left( \frac{|S^t|}{2kc_t} \right)^{1/t}$.  Now, in order to prove $\sum_{u}  A_u\ge 1$ with high probability, we will show that the variance of $\sum_{u}  A_u$  is bounded from above.  
Recall that if $A =\sum_{u}  A_u$ is a sum of indicator random variables, we must have 
$${\rm Var}(A) \le \avg[A]+\sum_{u \neq v} {\rm Cov}(A_u,A_v)=\avg[A]+\sum_{u \neq v} \Pr(A_u=1 \cap A_v=1)-\Pr(A_u=1)\Pr(A_v=1).$$ 
Now first consider the case when vertices $u$ and $v$ are at a distance of at least $2(a+b) \left(\frac{ \log n}{n}\right)^{1/t}$ apart (happens with probability $1- 4^t(a+b)^t c_{t} \left(\frac{\log n}{n|S^t|}\right)$). 
Then the region that is within distance $(a+b) \left(\frac{\log n}{n}\right)^{1/t}$ from both $u$ and $v$ is the empty-set. In that case, $\Pr(A_u=1 \cap A_v=1) = n(n-1)\dots (n-(4k-1)) c_3 \Big(\frac{\epsilon^t c_t \log n}{n|S^t|}\Big)^{4k} \Big(1-4k\epsilon^t \frac{c_t\log n}{n|S^t|}\Big)^{n-4k} = c_4 n^{-4k\epsilon^tc_t/|S^t|} (\epsilon^t \log n)^{4k},$ where $c_3,c_4$ are constants.

In all other cases, $\Pr(A_u=1 \cap A_v=1) \le \Pr(A_u =1)$.
Therefore,\\
\begin{align*}
\Pr(A_u=1 \cap A_v=1)&\leq &\Big(1- 4^t(a+b)^t c_t \left(\frac{\log n}{n|S^t|}\right)\Big) c_4 n^{-4k\epsilon^tc_t/|S^t|} (\epsilon^t \log n)^{4k}+\\ 
&&\frac{4^t(a+b)^tc_t \log n}{n|S^t|} c_2 n^{-2k\epsilon^t c_t/|S^t|} (\epsilon^t \log n)^{2k}
\end{align*}
and
\begin{align*}
{\rm Var(A)} &\le c_2n^{1-2k\epsilon^tc_t/|S^t|} (\epsilon^t \log n)^{2k} +{n \choose 2}\Big(\Pr(A_u=1 \cap A_v=1)-\Pr(A_u=1)\Pr(A_v=1)\Big) \\
&\le c_2n^{1-2k\epsilon^tc_t/|S^t|} (\epsilon^t \log n)^{2k}+ c_5 n^{1-2k\epsilon^tc_t/|S^t|} (\log n)^{2k+1}\\
& \le c_6  n^{1-2k\epsilon^tc_t/|S^t|} (\log n)^{2k+1}
\end{align*}
where $c_5,c_6$ are constants. Again invoking Chebyshev's inequality, with probability at least $1-O\Big(\frac{1}{\log n}\Big)$ 
$$
A > c_2n^{1-2k\epsilon^tc_t/|S^t|} (\epsilon^t \log n)^{2k} - \sqrt{c_6  n^{1-2k\epsilon^tc_t/|S^t|} (\log n)^{2k+2}}
$$  
which implies that $A>1$ with high probability.
\end{proof}

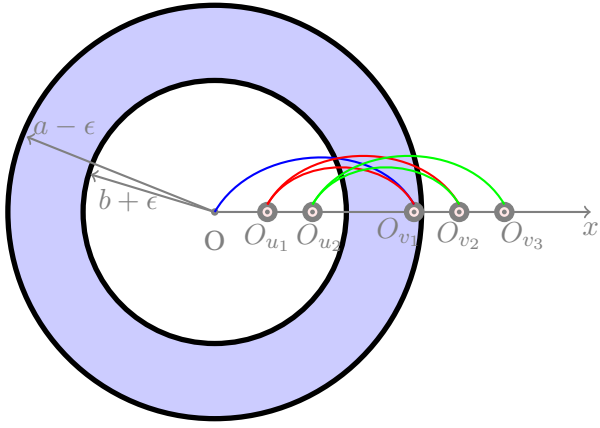
\begin{figure}
\vspace{-20pt}
\centering
\begin{tikzpicture}[thick, scale=0.5]

 \filldraw[color=black, fill=blue!20,  line width=2pt](0,0) circle (5.5);
 \filldraw[color=black, fill=red!0,  line width=2pt](0,0) circle (3.5);
    \draw[gray,->] (0,0)--(-3.3,1)node[anchor=north ] at (-2.3,0.9){$b+\epsilon$};
      \draw[gray,->] (0,0)--(-5,2)node[anchor=north ] at (-4,2.8){$a-\epsilon$};
   \draw[gray,->] (0,0)--(10,0)node[anchor=north ] at (10,0){$x$};

   \filldraw[color=gray, fill=red!10,  line width=2pt](5.3,0) circle (0.2)node[anchor=north ] at (4.9,0.15){$O_{v_1}$};
      \filldraw[color=gray, fill=red!10,  line width=2pt](6.5,0) circle (0.2)node[anchor=north ] at (6.5,0.05){$O_{v_2}$};
            \filldraw[color=gray, fill=red!10,  line width=2pt](7.7,0) circle (0.2)node[anchor=north ] at (8.2,0.05){$O_{v_3}$};
      \filldraw[color=gray, fill=red!10,  line width=2pt](1.4,0) circle (0.2)node[anchor=north ] at (1.4,0){$O_{u_1}$};
	\filldraw[color=gray, fill=red!10,  line width=2pt](2.6,0) circle (0.2)node[anchor=north ] at (2.8,0){$O_{u_2}$};
    \filldraw [gray] (5.3,0) circle (1pt);
      \filldraw [gray] (0,0) circle (2pt)node[anchor=north ] at (0,-0.2){O};
      \filldraw [gray] (1.4,0) circle (1pt);
            \filldraw [gray] (2.6,0) circle (1pt);
                        \filldraw [gray] (6.5,0) circle (1pt);
                                    \filldraw [gray] (7.7,0) circle (1pt);
              \draw[blue]    (0,0) to[out=60,in=120] (5.3,0.2);
  \draw[red]    (1.4,0.2) to[out=60,in=120] (5.3,0.2);
    \draw[red]    (1.4,0.2) to[out=60,in=120] (6.5,0.2);
        \draw[green]    (2.6,0.2) to[out=60,in=120] (6.5,0.2);
                \draw[green]    (2.6,0.2) to[out=60,in=120] (7.7,0.2);
  

\end{tikzpicture}
\caption{Representation of $u_i$ and $v_i$ in the $t+1$-dimensional sphere with respect to $u_0$.\label{fig:arr1}}
\end{figure}

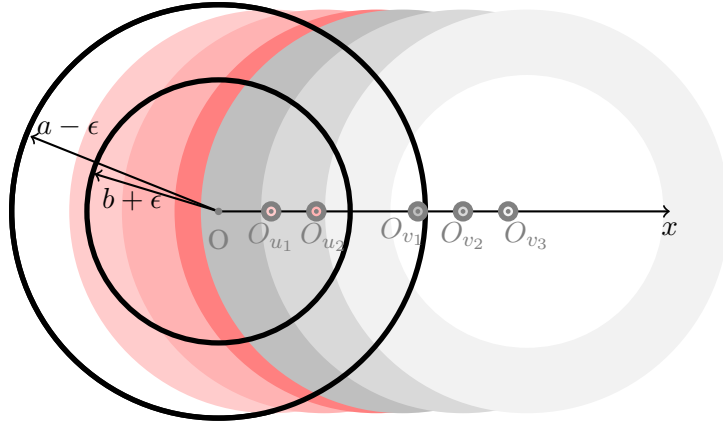
\begin{figure}
\vspace{-10pt}
\centering
\begin{tikzpicture}[thick, scale=0.5]

 \filldraw[color=black, fill=none,  line width=2pt](0,0) circle (5.5);
 
  \filldraw[color=red!20, fill=red!20,  line width=2pt] (1.4,0) circle (5.3);
 \filldraw[color=red!20, fill=red!0,  line width=2pt] (1.4,0) circle (3.7);
  
  \filldraw[color=red!30, fill=red!30,  line width=2pt] (2.8,0) circle (5.3);
    \filldraw[color=red!30, fill=red!0,  line width=2pt] (2.8,0) circle (3.7);
    
        \filldraw[color=red!50, fill=red!50,  line width=2pt] (4.2,0) circle (5.3);    
        \filldraw[color=red!50, fill=red!0,  line width=2pt] (4.2,0) circle (3.7);

                \filldraw[color=gray!50, fill=gray!50,  line width=2pt] (4.9,0) circle (5.3);  
                \filldraw[color=gray!50, fill=gray!0,  line width=2pt] (4.9,0) circle (3.7);
                
                                                \filldraw[color=gray!30, fill=gray!30,  line width=2pt] (6.5,0) circle (5.3);
                                \filldraw[color=gray!30, fill=gray!0,  line width=2pt] (6.5,0) circle (3.7);
                                
                                                                                \filldraw[color=gray!10, fill=gray!10,  line width=2pt] (8.2,0) circle (5.3);
                                \filldraw[color=gray!10, fill=blue!0,  line width=2pt] (8.2,0) circle (3.7);
                                
 \filldraw[color=black, fill=none,  line width=2pt](0,0) circle (3.5);

 \filldraw[color=black, fill=none,  line width=2pt](0,0) circle (5.5);

    \draw[black,->] (0,0)--(-3.3,1)node[anchor=north ] at (-2.3,0.9){$b+\epsilon$};
      \draw[black,->] (0,0)--(-5,2)node[anchor=north ] at (-4,2.8){$a-\epsilon$};
   \draw[black,->] (0,0)--(12,0)node[anchor=north ] at (12,0){$x$};

   \filldraw[color=gray, fill=gray!50,  line width=2pt](5.3,0) circle (0.2)node[anchor=north ] at (4.9,0.15){$O_{v_1}$};
      \filldraw[color=gray, fill=gray!30,  line width=2pt](6.5,0) circle (0.2)node[anchor=north ] at (6.5,0.05){$O_{v_2}$};
            \filldraw[color=gray, fill=gray!10,  line width=2pt](7.7,0) circle (0.2)node[anchor=north ] at (8.2,0.05){$O_{v_3}$};
      \filldraw[color=gray, fill=red!20,  line width=2pt](1.4,0) circle (0.2)node[anchor=north ] at (1.4,0){$O_{u_1}$};
	\filldraw[color=gray, fill=red!30,  line width=2pt](2.6,0) circle (0.2)node[anchor=north ] at (2.8,0){$O_{u_2}$};
    \filldraw [gray] (5.3,0) circle (1pt);
      \filldraw [gray] (0,0) circle (2pt)node[anchor=north ] at (0,-0.2){O};
      \filldraw [gray] (1.4,0) circle (1pt);
            \filldraw [gray] (2.6,0) circle (1pt);
                        \filldraw [gray] (6.5,0) circle (1pt);
                                    \filldraw [gray] (7.7,0) circle (1pt);
\end{tikzpicture}
\caption{Shaded regions represent the region of connectivity with $u_i$ and $v_i$ (red for $u_i$'s and gray for $v_i$'s).\label{fig:arr2}}
\end{figure}
\begin{lemma*}(\ref{lem:pole})
In a ${\rm RAG}_t\left(n, \left[b\left(\frac{\log n}{n}\right)^{1/t},a\left(\frac{\log n}{n}\right)^{1/t}\right]\right), 0 <b <a$, with  probability $1-o(1)$ there exists a vertex $u_0$ such that any node $v$ that satisfies $\|u-v\|_2 \le a\left(\frac{\log n}{n}\right)^{1/t}$ is connected to $u_0$.
\end{lemma*}
\begin{proof}
Consider the vertices $u_0,\{u_1, u_2, \ldots, u_k\}$ and $\{v_1, v_2, \ldots, v_k\}$ that satisfy the conditions of Lemma \ref{lem:lemma2pole} as shown in Fig \ref{fig:arr1}. We can observe that each vertex $v_i$ has an edge with $u_i$ and $u_{i-1}$, $i =1, \ldots,k$. 
\begin{align*}
\|u_i-v_i\|_2 &\ge \|u_i-O_{v_i}\|_2 - \|v_i-O_{v_i}\|_2\\
&\ge \|O_{v_i}-O_{u_i}\|_2  - \|u_i-O_{u_i}\|_2 - \|v_i-O_{v_i}\|_2\\
& \ge (b+2\epsilon)\left(\frac{\log n}{n}\right)^{1/t} - 2\epsilon \left( \frac{\log n}{n}\right)^{1/t}= b\left(\frac{ \log n}{n}\right)^{1/t} \quad \text{and}
\end{align*}
\begin{align*}
 \|u_i-v_i\|_2 &\le  \|O_{v_i}-O_{u_i}\|_2 + \|u_i-O_{u_i}\|_2 + \|v_i-O_{v_i}\|_2\\
  & =  (b+2\epsilon)\left(\frac{\log n}{n}\right)^{1/t} + 2\epsilon  \left(\frac{\log n}{n}\right)^{1/t}= (b+4\epsilon)\left(\frac{ \log n}{n} \right)^{1/t}
\end{align*}

Similarly, 
\vspace{-10pt}
\begin{align*}
\|u_{i-1}-v_i\|_2& \ge \|u_{i-1}-O_{v_i}\|_2 - \|v_i-O_{v_i}\|_2\\
&\ge \|O_{v_i}-O_{u_{i-1}}\|_2 - \|u_{i-1}-O_{u_{i-1}}\|_2 - \|v_i-O_{v_i}\|_2\\
&\ge (a-2\epsilon)\left( \frac{\log n}{n}\right)^{1/t} - 2\epsilon\left(\frac{ \log n}{n}\right)^{1/t}\\
& = (a-4\epsilon)\left(\frac{\log n}{n}\right)^{1/t} \quad \text{and}
\end{align*}
\begin{align*}
\|u_{i-1}-v_i\|_2 &\le \|O_{v_i}-O_{u_{i-1}}\|_2  + \|u_{i-1}-O_{u_{i-1}}\|_2 + \|v_i-O_{v_i}\|\\
&\le (a-2\epsilon)\left(\frac{\log n}{n}\right)^{1/t} - 2\epsilon\left(\frac{ \log n}{n}\right)^{1/t}= a\left(\frac{\log n}{n}\right)^{1/t}.
\end{align*}
 This implies that $u_0$ is connected to $u_i$ and $v_i$ for all $i=1,\dots,k$.  Next, we show that any point in the region ${B_t}\Big(u_0,r_s=a(\frac{\log n}{n})^{1/t}\Big)$ is connected to $u_0$. Now recall that any point in the region ${B_t}\Big(x,\left[(b+\epsilon)\Big(\frac{\log n}{n}\Big)^{t},(a-\epsilon)\Big(\frac{\log n}{n}\Big)^{t}\right]\Big)$ is connected to any point in the region ${B_t}(x,\epsilon)$.
 We can observe that the nodes $u_1,\ldots, u_k,v_1$,\ldots, $v_k$ form a cover of ${B_t}\Big(u_0,r_s=a(\frac{\log n}{n})^{1/t}\Big)$ in the form of these annulus regions (A region corresponding to a particular node implies the portion of the hypersphere such that any other node in that region is connected to it) translated by $(a-b-4\epsilon)\left(\frac{\log n}{n}\right)^{1/t}$. This is because any node in ${B_t}\Big(O_{u_i},\left[(b+\epsilon)\Big(\frac{\log n}{n}\Big)^{t},(a-\epsilon)\Big(\frac{\log n}{n}\Big)^{t}\right]\Big)$ is connected to $u_i$ and any node in ${B_t}\Big(O_{v_i},\left[(b+\epsilon)\Big(\frac{\log n}{n}\Big)^{t},(a-\epsilon)\Big(\frac{\log n}{n}\Big)^{t}\right]\Big)$  is connected to $v_i$ respectively (Figure\ref{fig:arr2}). Therefore any node falling in any of the aforementioned regions is connected with $u_0$. Since the width of each region is $(a-b-2\epsilon)\left(\frac{\log n}{n}\right)^{1/t}$, the regions overlap with each other. Additionally, the inner radius of a particular region is $(b+\epsilon)\left(\frac{\log n}{n}\right)^{1/t}$ which is greater than $b\left(\frac{\log n}{n}\right)^{1/t}$. Hence, there can not exist any point in ${B_t}\Big(u_0,r_2=a(\frac{\log n}{n})^{1/t}\Big)$ which is not covered by the union of these regions. 
\end{proof}

\subsection{Proofs of Lemma \ref{lem:high_stuff1} and Lemma \ref{lem:cut_twice}} 
Assume that the $(t+1)$-dimensional space is described by a coordinate system whose center coincides with the center of the sphere. In this coordinate system, let us denote the point $(1,0,0,\dots,0)$ by $u_0$. 

Lemma \ref{conjec:ratio} shows that for any plane $L$ with  $\mathcal{A}_{L} \not \subset B_t(u,r_2)$, the region of connectivity of $u$ on both sides differ by a constant fraction.
\begin{lemma}\label{conjec:ratio}
For a particular node $u$ in ${\rm RAG}_t\left(n,\left[r_1,r_2\right]\right)$ where $r_1=b\Big(\frac{\log n}{n}\Big)^{1/t}$, $r_2=a\Big(\frac{\log n}{n}\Big)^{1/t}$, consider a hyperplane $L$ passing through $u$ such that $\mathcal{A}_{L} \not \subset B_t(u,a\Big(\frac{\log n}{n}\Big)^{1/t})$, then $ \frac{\min (|R_{L}^{1}|,|R_{L}^{2}|)}{|R_{L}^{1}|+|R_{L}^{2}|} \ge \delta$ if $a > 2b$, where $\delta$ is a constant.
\end{lemma}
\begin{proof}

First, for a node $u$ and a given hyperplane $L:w^{T}x=\beta$ passing through $u$, we try to evaluate the surface area of the region corresponding to
\begin{align*}
\{x \in S^t \mid  \|u-x\|_2 \le r_2=a\Big(\frac{\log n}{n}\Big)^{1/t}, w^{T}x \ge \beta \}
\end{align*}
such that  the farthest point from $u$ on the plane $L$ and $S^t$ is at distance $r_2$. This region is a spherical cap corresponding to $B_t(u',r')$ where $u'$ is the intersection of $S^t$ with the normal from the origin to the plane and $r'=\|u-u'\|_2$. Suppose $h$ is the height of this cap (perpendicular distance from $u'$ to the hyperplane $L$). Using pythagoras theorem, we can see that $r'^2=h^2+\left(r_2/2\right)^2$ and $(1-h)^2+ r_2^2/4 = 1$. Simplifying this, we get $h=\frac{r'^2}{2}$ and hence $r'\approx  \frac{r_2}{2}$.
Hence the area of this region is $ c_t \left(r_2/{2}\right)^t$.



Without loss of generality, assume $|\mathcal{R}_{L}^1| \ge |\mathcal{R}_{L}^2|$. Now,
\begin{eqnarray*}
|\mathcal{R}_{L}^1| &=& |\{x \in S^t \mid b\Big(\frac{\log n}{n}\Big)^{1/t} \le  \|u-x\|_2 \le a\Big(\frac{\log n}{n}\Big)^{1/t}, w^{T}x \le \beta \}|\\
&=&  |\{x \in S^t \mid b\Big(\frac{\log n}{n}\Big)^{1/t} \le  \|u-x\|_2 \le a\Big(\frac{\log n}{n}\Big)^{1/t} \}| - |\mathcal{R}_{L}^2|\\
&=& c_t(r_2^t-r_1^t) -  |\mathcal{R}_{L}^2|.\\\\
|\mathcal{R}_{L}^2| &\equiv& |\{x \in S^t \mid b\Big(\frac{\log n}{n}\Big)^{1/t} \le  \|u-x\|_2 \le a\Big(\frac{\log n}{n}\Big)^{1/t}, w^{T}x \ge \beta \}|\\
&\equiv& |\{x \in S^t \mid  \|u-x\|_2 \le a\Big(\frac{\log n}{n}\Big)^{1/t}, w^{T}x \ge \beta \}| - |\{x \in S^t \mid   \|u-x\|_2 \le b\Big(\frac{\log n}{n}\Big)^{1/t}, w^{T}x \ge \beta \}|\\
 &\ge&   |\{x \in S^t \mid  \|u-x\|_2 \le a\Big(\frac{\log n}{n}\Big)^{1/t}, w^{T}x \ge \beta \}| -\\
&& \left[ |\{x \in S^t \mid   \|u-x\|_2 \le b\Big(\frac{\log n}{n}\Big)^{1/t},w^{T}x < \beta \}| + |\{x \in S^t \mid   \|u-x\|_2 \le b\Big(\frac{\log n}{n}\Big)^{1/t},w^{T}x \ge \beta \}|\right]\\
&\equiv& |\{x \in S^t \mid  \|u-x\|_2 \le a\Big(\frac{\log n}{n}\Big)^{1/t}, w^{T}x \ge \beta \}| - |\{x \in S^t \mid   \|u-x\|_2 \le b\Big(\frac{\log n}{n}\Big)^{1/t} \}|\\
&\equiv& c_t \left(r_2/2\right)^t -c_tr_1^t
\end{eqnarray*}
If $c_t \left(r_2/{2}\right)^t -c_tr_1^t > 0$, then,
\begin{eqnarray*}
1\le \frac{|\mathcal{R}_{L}^1|}{|\mathcal{R}_{L}^2|} &\le&  \frac{c_t(r_2^t-r_1^t)}{c_t \left(r_2/{2}\right)^t -c_tr_1^t}-1\\
&=& \frac{ a^t -b^t}{ \left(a/{2}\right)^t -b^t} -1 = \delta'\\
\end{eqnarray*}
Hence, 
\begin{eqnarray*}
2\le \frac{|\mathcal{R}_{L}^1|+ |\mathcal{R}_L^2|}{|\mathcal{R}_{L}^2|} &\le&  1+\delta'\\
\end{eqnarray*}
This gives us that $\frac{\min (|\mathcal{R}_{L}^{1}|,|\mathcal{R}_{L}^{2}|)}{|\mathcal{R}_{L}^{1}|+|\mathcal{R}_{L}^{2}|} =\frac{|\mathcal{R}_{L}^2|}{|\mathcal{R}_{L}^1|+|\mathcal{R}_{L}^2|} \geq \frac{1}{1+\delta'} = \delta$. 
Hence, the claim of the lemma is satisfied   if $\left(a/{2}\right)^t -b^t > 0$ i.e. $a>2b.$
\end{proof}

\begin{corollary}\label{cor:ratio}
For a particular node $u$ in ${\rm RAG}_t\left(n,\left[b\Big(\frac{\log n}{n}\Big)^{1/t},a\Big(\frac{\log n}{n}\Big)^{1/t}\right]\right)$, consider a hyperplane $L$ passing through $u$ such that $\mathcal{A}_{L} \not \subset B_t(u,a\Big(\frac{\log n}{n}\Big)^{1/t})$, then $ \frac{\min (|\mathcal{R}_{L}^{1}|,|\mathcal{R}_{L}^{2}|)}{|\mathcal{R}_{L}^{1}|+|\mathcal{R}_{L}^{2}|} \ge  \frac{ \left(a/2\right)^t -b^t}{ a^t -b^t}$.
\end{corollary}
\begin{proof}
Using Lemma \ref{conjec:ratio}, $ \frac{\min (|\mathcal{R}_{L}^{1}|,|\mathcal{R}_{L}^{2}|)}{|\mathcal{R}_{L}^{1}|+|\mathcal{R}_{L}^{2}|} \ge \delta = \frac{1}{1+ \delta'}$ where $\delta' = \frac{ \left(a\right)^t -b^t}{ \left(a/{2}\right)^t -b^t} -1 $.
\end{proof}

For a node $u$, recall that the hyperplane $L^{\star}_u : x_1=u_1$ is normal to the line joining $u_0$ and the origin and passes through $u$. We now have the following lemma, which tries to show that if the plane satisfies $\mathcal{A}_{L^{\star}_u}  \subseteq B_t(u,r_2)$ then the node $u$ must be within $r_2$  distance of $u_0$. 
\begin{lemma*}(\ref{lem:cut_twice})
For a particular node $u$ and corresponding hyperplane $L^{\star}_u$, if $\mathcal{A}_{L^{\star}_u}  \subseteq B_t(u,r_2)$  then $u$ must be within $r_2$ of $u_0$.
\end{lemma*}
\begin{proof}
The reflection of  $u$ in x-axis (say $v$) is the farthest from $u$ that lies on both $\mathcal{A}_{L^{\star}_u}$ and $\mathcal{S}^t$.  Now we want to show that $v = (u_1,-u_2,\ldots,-u_{t+1})$ has the following property:  if $\|u-v\|_2\le r_2$ then $\|u-u_0\|_2\le r_2$.
We are given that
\begin{eqnarray*}
u_1^2+u_2^2+\dots+u_{t+1}^2 = 1\\
d(u,v) = \sqrt{4(u_2^2+\ldots+u_{t+1}^2)}\le r_2\\
4(1-u_1^2)\leq r_2^2
\end{eqnarray*}
We need to show that, 
\begin{eqnarray*}
d(u,u_0)^2 &=& {(1-u_1)^2 + (u_2^2+\ldots+u_{t+1}^2) }\\
&=&(1-u_1)^2 + (1-u_1^2) \\
&=&  {2-2u_1}\\
&\leq & r_2^2
\end{eqnarray*}
which holds if $\sqrt{\frac{4-r_2^2}{4}}\ge \frac{2-r_2^2}{2}$. Notice that, 
\begin{eqnarray*}
&\sqrt{\frac{4-r_2^2}{4}}\ge \frac{2-r_2^2}{2}\\
 &\implies 4-r_2^2 \geq4-4r_2^2+r_2^4\\
 & \implies r_2^4-3r_2^2\leq 0 \\
 & \implies r_2^2(r_2^2-3) \le 0
\end{eqnarray*}
which is true since $0 \le  r_2 \le 1$.
\end{proof}
Since, we do not know the location of the pole, we need to show that every point has a neighbor on both sides of the plane $L$ no matter what the orientation of the plane given that $\mathcal{A}_{L} \not \subset B_t(u,r_2)$. For this we need to introduce the concept of VC Dimension. Define $(X,R)$ to be a range space if $X$ is a set (possibly infinite) and $R$ is a family of subsets of $X$. For any set $A \subseteq X$, we define $P_{R}(A)=\{r \cap A \mid r \in R\}$ to be the projection of $R$ on $A$. Finally we define the VC dimension $d$ of a range space $(X,R)$ to be $d=\sup_{A \subseteq X} \{|A| \mid |P_R(A)|=2^{A} \}$. Next we give a modified version of a well-known theorem about VC-dimension (\cite{haussler1987}).

\begin{theorem}\label{thm:VC_dim}
Let $(X,R)$ be a range space of VC dimension $d$ and let $U$ be a uniform probability measure defined on $X$. In that case, if we sample a set $\mathcal{M}$ of $m$ points according to $U$ such that
\begin{align*}
m \ge \max \Big(\frac{8d}{\epsilon}\log \frac{8d}{\epsilon}, \frac{4}{\epsilon}\log \frac{2}{\eta} \Big)
\end{align*}
then with probability $1-\eta$ for any set $r \in R$ such that $\Pr_{x \sim_U X}(x \in r) \ge \epsilon $, we have $|r \cap \mathcal{M}| \neq \Phi$.  \\
\end{theorem}
\begin{proof}
Define a set $r \in R$ to be \emph{heavy} if $\Pr_{x \sim_{U} X} (x \in r) \ge \epsilon$. We pick two random samples $N$ and $T$ each of size $m$ according to the uniform distribution defined on $X$. Consider the event $E_1$ (bad event) for which there exists a heavy $r \in R$ such that $r \cap N=\Phi$. Consider another event 
$E_2$ for which there exists a heavy $r \in R$ such that $r \cap N=\Phi$ and $|r \cap T| \ge \frac{\epsilon m}{2}$. Now, since $r$ is heavy, assume that $\Pr_{x \sim_{U} X} (x \in r)=\alpha$ such that $\alpha>\epsilon$. In that case, $|r \cap T|$ is a Binomial random variable with mean $\alpha m$ and variance at most $\alpha m$ as well. Hence, we have that 
\begin{align*}
\Pr (E_2 \mid E_1)&=\Pr(|r \cap T| \ge \frac{\epsilon m}{2})=1- \Pr(|r \cap T| \le \frac{\epsilon m}{2})  \\
&\ge 1- \Pr(|r \cap T| \le \frac{\alpha m}{2})  \ge 1- \frac{\alpha m}{ (\frac{\alpha m}{2})^{2}} \ge 1-\frac{4}{m \alpha} 
\end{align*} 
Now for $m \ge \frac{8}{\epsilon} \ge \frac{8}{\alpha}$, we conclude that $\Pr (E_2 \mid E_1) \ge \frac{1}{2}$. Now consider the same experiment in a different way. Consider picking $2m$ samples according to the uniform distribution from $X$ and then equally partition them randomly between $N$ and $T$. Consider the following event for a particular set $r$.
\begin{align*}
E_r: r \cap N= \Phi \text{ and } |r \cap T| \ge \frac{\epsilon m}{2} 
\end{align*}
and therefore
\begin{align*}
E_2=\bigcup _{r:\textup{ heavy}} E_r 
\end{align*}
Let us fix $N \cup T$ and define $p=|r \cap (N \cup T)|$. In that case, we have 
\begin{align*}
\Pr(r \cap N= \Phi \mid |r \cap (N \cup T)| \ge \frac{\epsilon m}{2})=\frac{(2m-p)(2m-p-1)\dots(m-p+1)}{2m(2m-1) \dots (2m-p+1)} \le 2^{-p} \le 2^{-\frac{\epsilon m}{2}}
\end{align*}
The last inequality holds since $p \ge \frac{\epsilon m}{2}$. Now, since the VC dimension of the range space $(X,R)$ is $d$, the cardinality of the set $\{ r \cap (N \cup T) \mid r \in R\}$ is at most $\sum_{i \le d} {2m \choose i} \le (2m)^{d}$ ( see \cite{shalev2014understanding}). Notice that
\begin{align*}
\Pr(E_r)=\Pr(r \cap N= \Phi \mid |r \cap (N \cup T)| \ge \frac{\epsilon m}{2}) \Pr(|r \cap (N \cup T)| \ge \frac{\epsilon m}{2}) \le 2^{-\frac{\epsilon m}{2}}
\end{align*}
Therefore by using the union bound over the possible number of distinct events $E_r$, we have
\begin{align*}
\Pr(E_2) \le (2m)^{d} 2^{-\frac{\epsilon m}{2}}
\end{align*}
Since $\Pr(E_2 \mid E_1) \ge \frac{1}{2}$ and $\Pr( E_1 \mid E_2)=1$, we must have 
\begin{align*}
\Pr(E_1) \le 2(2m)^{d} 2^{-\frac{\epsilon m}{2}} \le \delta
\end{align*}
which is ensured by the statement of the theorem.
\end{proof}

In order to use this theorem consider the range space $(X,\mathcal{R}_{u})$ where $X$ is the set of points in $\mathcal{S}^t$ and $\mathcal{R}_u$ be the family of sets $\{x \in S^t \mid b\Big(\frac{\log n}{n}\Big)^{1/t} \le  \|u-x\|_2 \le a\Big(\frac{\log n}{n}\Big)^{1/t}, w^{T}x \ge \beta , \mathcal{A}_{L:w^{T}x=\beta} \not \subset B_t(u,r_2)\}$. We now have the following lemma about the VC Dimension of the above range space which is a straightforward extension of VC dimension of half-spaces  (\cite{shalev2014understanding}):
\begin{lemma}\label{lem:VC_dimension}
VC dimension of the range space $(X,\mathcal{R}_u) \le t+1$. 
\end{lemma}
\begin{proof}
In order to show this, consider a set $\mathcal{S}$ of $t+2$ points. Recall that the convex hull of a set $S$ of points $\{x_i\}_{i=1}^{n}$ is the set 
\begin{align*}
C(S)=\{ \sum \lambda_{i}x_{i} \mid \sum \lambda_i=1, \lambda_i \ge 0  \}.
\end{align*}
By Radon's lemma (\cite{shalev2014understanding}) we have that the set of points $\mathcal{S}$ can be partitioned into two sets $\mathcal{S}_1$ and $\mathcal{S}_2$ such that their convex hulls intersect.
Let $p \in \mathcal{S}_1$ be a point in that intersection. Assume there exist a hyperplane such that
\begin{align*}
w^{T}x_i  \le w_0,  \forall x_i  \in \mathcal{S}_1 \\
 w^{T}x_i \ge w_0, \forall x_i  \in \mathcal{S}_2.
\end{align*}
Since $p$ is in the convex hull of $\mathcal{S}_1$ we must have that $w ^{T}p \le w_0$. But then,
\begin{align*}
w^{T}p= \sum_{i:x_i \in \mathcal{S}_2} \lambda_i w^{T} x_i  >( \sum_{i \in \mathcal{S}_2} \lambda_i ) \min_{i: x_i \in \mathcal{S}_2} w^{T}x_i= \min (w^{T} x_i)>w_0.
\end{align*}
which is a contradiction. Hence it is not possible to shatter $t+2$ elements and therefore the VC dimension of this range space is at most $t+1$.
\end{proof}

Using the results shown above, we are ready to prove the following Lemma.

\begin{lemma*}(\ref{lem:high_stuff1})
If we sample $n$ nodes from $S^t$ according to ${\rm RAG}_t(n,[r_1,r_2])$ with $r_1=b\Big(\frac{\log n}{n}\Big)^{1/t}$, $r_2=a\Big(\frac{\log n}{n}\Big)^{1/t}$ , then for every node $u$ and every hyperplane $L$ passing through $u$ such that $\mathcal{A}_L \not \subset B(u,r_2)$, node $u$ has a neighbor on both sides of the hyperplane $L$ with probability at least $1-\frac{1}{n}$ provided 
\begin{align*}
(a/2)^t-b^t \ge \frac{8|S^t|(t+1)}{c_t}
\end{align*} and $a > 2b.$
\end{lemma*}

\begin{proof}
Recall that the volume of $B(u,r_1,r_2)$ is $c_t(r_2^t-r_1^t)=\frac{c_t \log n}{n}(a^t-b^t)$. According to Corollary \ref{cor:ratio}, whenever $a>2b$, $ \frac{\min (|\mathcal{R}_{L}^{1},\mathcal{R}_{L}^{2}|)}{|\mathcal{R}_{L}^{1}|+|\mathcal{R}_{L}^{2}|} \ge \delta$ where $\delta =  \frac{ \left(a/2\right)^t -b^t}{ a^t -b^t}$ when the hyperplane $L$ satisfies the conditions of the lemma. In that case, we have that for all $r \in R_u$,
\begin{align*}
 \Pr_{x \sim_U X}(x \in r) \ge \frac{\delta c_t \log n}{ n |S^t|}(a^t-b^t).
\end{align*}
Since VC Dimension of $(X,R_u) \le t+1$ and $n$ points are sampled from $X$, the conditions of Theorem \ref{thm:VC_dim} is satisfied for $\eta=\frac{2}{n^{2}}$ if
\begin{align*}
n \ge \max \Big ( \frac{8n|S^t|(t+1)}{\delta c_t \log n (a^t-b^t)}\log \frac{8 n|S^t|(t+1)}{\delta c_t\log n (a^t-b^t)},\frac{8n|S^t|}{\delta c_t \log n (a^t-b^t)} \log n  \Big)
\end{align*} 
Since $\lim_{n \rightarrow \infty} \frac{1}{\log n} \log \frac{8|S^t| n(t+1)}{\delta c_t\log n (a^t-b^t)} \rightarrow 1$ for constant $t$, hence we have that 
\begin{align*}
a^t-b^t \ge \frac{8|S^t|(t+1)}{c_t\delta}.
\end{align*} 
By taking a union bound over all the $n$ range spaces $(X,R_u)$ corresponding to the $n$ nodes and applying the statement of Theorem \ref{thm:VC_dim}, we have proved the lemma.
\end{proof}



 \section{The Geometric Block Model: Details}
\label{sec:gbm-detail}

\subsection{Immediate consequence of VRG connectivity}
The following lower bound for GBM can be obtained as a  consequence of Theorem~\ref{thm:rag}.

\begin{theorem*}(\ref{gbm:lower})[Impossibility in GBM]
Any algorithm to recover the partition in ${\rm GBM}(\frac{a \log n}{n},\frac{b \log n}{n})$ will give incorrect output with probability $1-o(1)$ if $a-b < 0.5$ or $a < 1$.
\end{theorem*}
\begin{proof}
Consider the scenario that not only the geometric block model graph  ${\rm GBM}(\frac{a \log n}{n},\frac{b \log n}{n})$ was provided to us, but also the random values $X_u \in [0,1]$ for all vertex $u$ in the graph were provided. We will show that we will still not be able to recover the correct partition of the vertex set $V$ with probability at least $0.5$ (with respect to choices of $X_u,~u, v\in V$ and any randomness in the algorithm).

In this situation, the edge $(u,v)$ where $d_L(X_u,X_v) \le \frac{b \log n}{n}$ does not give any new information than $X_u,X_v$. However the edges $(u,v)$ where $\frac{b \log n}{n} \le d_L(X_u,X_v) \le \frac{a \log n}{n}$ are informative, as existence of such an edge will imply that $u$ and $v$ are in the same part. These edges constitute a vertex-random graph ${\rm VRG}(n, [\frac{b \log n}{n},\frac{a \log n}{n}])$. But if there are more than two components in this vertex-random graph, then it is impossible to separate out the vertices into the correct two parts, as the connected components can be assigned to any of the two parts and the VRG along with the location values ($X_u, u \in V$) will still be consistent. 

What remains to be seen that ${\rm VRG}(n, [\frac{b \log n}{n},\frac{a \log n}{n}])$ will have $\omega(1)$ components with high probability  if $a-b < 0.5$ or $a < 1$. This is certainly true when $a-b < 0.5$ as we have seen in Theorem~\ref{thm:lower_bound}, there can indeed be $\omega(1)$ isolated nodes with high probability. On the other hand, when $a<1$, just by using an analogous argument it is possible to show that there are $\omega(1)$ vertices that do not have any neighbors on the left direction (counterclockwise). We delegate the proof of this claim as Lemma \ref{lem:disc} in the appendix.  
If there are $k$ such vertices, there must be at least $k-1$ disjoint candidates.This completes the proof.
\end{proof}


\begin{theorem}[GBM with known vertex locations]\label{thm:gbmplus}
Suppose a geometric block model graph ${\rm GBM}(\frac{a \log n}{n},\frac{b \log n}{n})$ is provided along with the associated  values of the locations $X_u$ for every vertex $u$. Any algorithm to recover the partition in ${\rm GBM}(\frac{a \log n}{n},\frac{b \log n}{n})$ will give incorrect output with probability $1-o(1)$ if $a-b < 0.5$ or $a < 1$. On the other hand it is possible to recover the partition exactly with probability $1-o(1)$ when $a-b > 0.5$ and $a > 1$.
\end{theorem}
\begin{proof}
We need to only prove that  it is possible to recover the partition exactly with probability $1-o(1)$ when $a-b > 0.5$ and $a > 1$, since the other part is immediate from Theorem~\ref{gbm:lower}. For any pair of vertices $u,v$, we can verify if $d(u,v) \in [\frac{b \log n}{n},\frac{a \log n}{n}]$. If that is the case then by just checking in the GBM graph whether they are connected by an edge or not we can decide whether they belong to the same cluster or not respectively. What remains to be shown that  all vertices can be covered by this procedure. However that will certainly be the case since ${\rm VRG}(n , [\frac{b \log n}{n},\frac{a \log n}{n}])$ is connected with high probability.
\end{proof}

Next we provide the main algorithm for recovery in GBM and its analysis.
\subsection{A recovery algorithm for GBM}

Suppose we are given a graph $G=(V:|V|=n,E)$ with two disjoint parts, $V_1, V_2 \subseteq V$ generated according to ${\rm GBM}(r_s, r_d)$. 
The algorithm (Algorithm ~\ref{alg:alg1}) goes over  all edges  $(u,v)\in E$. It counts the number of triangles  containing the edge  $(u,v)$ by calling the \texttt{process} function that  
 counts the number of common neighbors of $u$ and $v$.  
 
    \texttt{process} outputs `true' if it is confident that the nodes $u$ and $v$ belong to the same cluster and `false' otherwise. More precisely, if the count is within some prescribed values $E_S$ and $E_D$, it returns `false'\footnote{Note that, the thresholds $E_S$ and $E_D$ refer to the maximum and minimum value of triangle-count for an `inter cluster' edge.}.The algorithm removes the edge on getting a `false' from \texttt{process} function. After processing all the edges of the network, the algorithm is left with a reduced graphs (with certain edges deleted from the original). It then finds the connected components in the graph and returns them  as the parts $V_1$ and $V_2$. 

 \begin{remark}
The algorithm can iteratively maintain the connected components over the processed edges (the pairs for which process function has been called and it returned true) like the union-find algorithm. This reduces the number of queries as the algorithm does not need to  call the \texttt{process} function for the edges which are present in the same connected component. 
 \end{remark}

\captionof{algorithm}{Cluster recovery in GBM}
\begin{algorithmic}[1]\label{alg:alg1}
{
\REQUIRE GBM $G = (V,E)$, $r_s, r_d$
\FOR {$(u,v)\in E$}
\IF{{\rm process}($u,v,r_s,r_d$)}
\STATE continue
\ELSE
\STATE $E.remove((u,v))$
\ENDIF
\ENDFOR
\RETURN {\rm connectedComponent}$(V,E)$
}
\end{algorithmic}
 
\captionof{algorithm}{\texttt{process}}
\begin{algorithmic}[1]\label{alg:process}
{
\REQUIRE $u$,$v$, $r_s$, $r_d$
\ENSURE  true/false\\
\COMMENT{Comment: When $a>2b$, $t_1=\min\{t: (2b+t)\log \frac{2b+t}{2b}-t > 1\}, t_2=\min\{t : (2b-t)\log \frac{2b-t}{2b}+t > 1$ and $E_S  = (2b +t_1)\frac{\log n}{n}$ and $E_D  = (2b - t_2)\frac{\log n}{n}$}
\STATE count $\leftarrow |\{z: (z,u)\in E, (z,v)\in E\}|$
\IF{$\frac{\text{count}}{n} \ge E_S(r_d,r_s)$ or $\frac{\text{count}}{n} \le E_D(r_d,r_s)$}
\RETURN true
\ENDIF
\RETURN false
}
\end{algorithmic}


It would have been natural to consider two thresholds $E_D$ and $E_S$ and if the triangle count of an edge is closer to $E_S$ than $E_D$, then the two end-points are assigned to the same cluster and otherwise in separate clusters. Indeed such a natural algorithm has been analyzed in \citep{galhotra2017geometric}. On the other hand, here we remove an edge if the triangle count lies in an interval. This is apparently nonintuitive, but gives a significant improvement over the previously known bound (see Figure~\ref{fig:example}). 

\begin{figure}[htbp]
   \centering
   \includegraphics[width=0.5\textwidth]{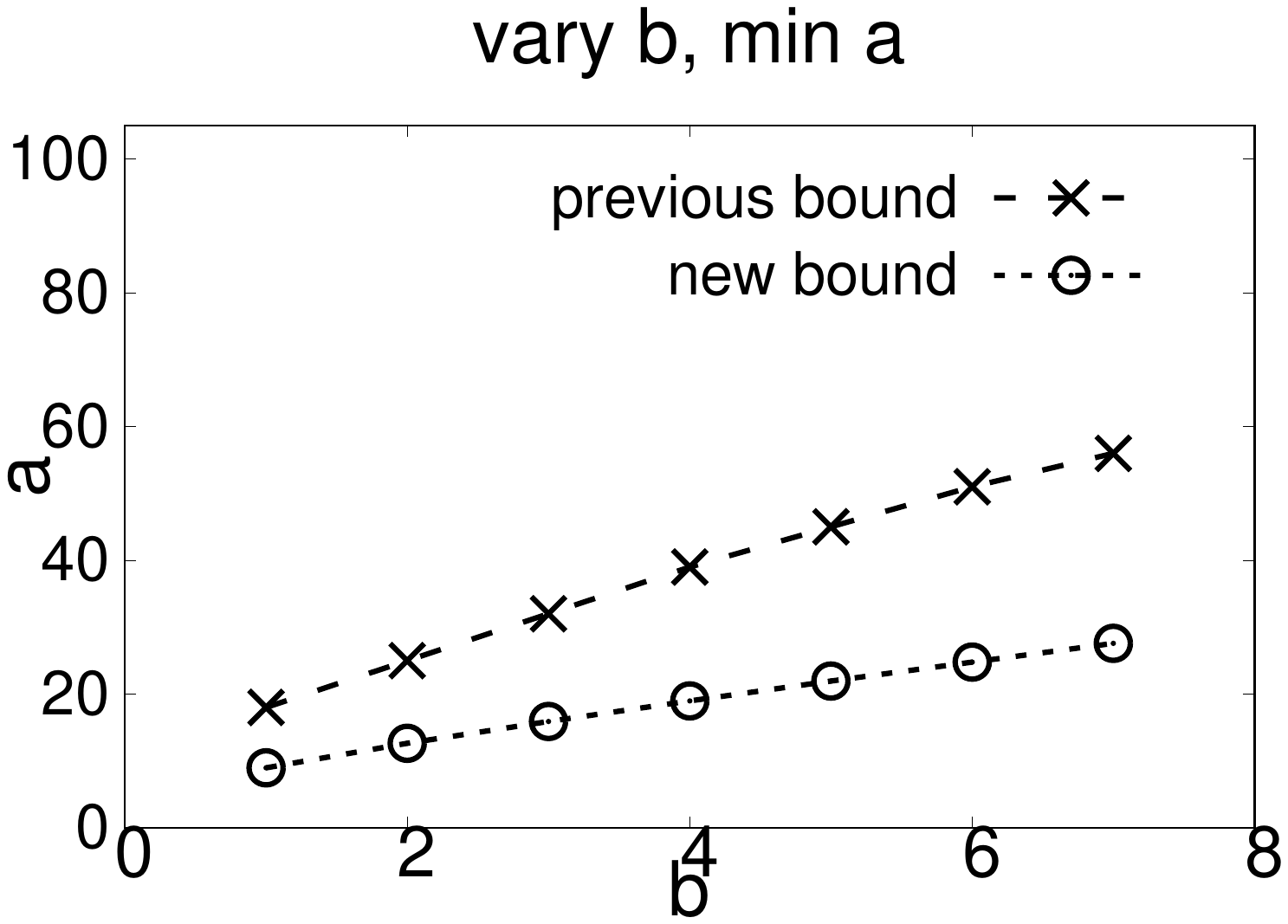} 
   \caption{The minimum gap between $a$ and $b$ permitted by our algorithm vs the previously known bound of \citep{galhotra2017geometric}}
   \label{fig:example}
\end{figure}

\subsection{Analysis of Algorithm~\ref{alg:alg1}}
~\label{sec:theory}
Given a GBM graph $G(V,E)$ with two clusters $V = V_1 \sqcup V_2$, and  a pair of vertices $u,v \in V$, the events $\cE^{u,v}_z, z \in V$ of any other vertex $z$ being a common neighbor of both $u$ and $v$ given $(u,v) \in E$ are dependent ; however given the distance between the corresponding random variables  \correct{$d_L(X_u,X_v) =x$}, the events 
are independent. This is a crucial observation which helps us to overcome the difficulty of handling correlated edge formation.       

Moreover, given the distance between two nodes $u$ and $v$ are the same, the probabilities of $\cE^{u,v}_z\mid (u,v) \in E$ are different when $u$ and $v$ are in the same cluster and when they are in different clusters. Therefore the count of the common neighbors are going to be different, and substantially separated with high probability 
 for two vertices in cases when they are from the same cluster or from different clusters. However, this may not be the case, if we do not restrict the distance to be the same and look at the entire range of possible distances. 
 
 First, we quote two simple lemmas about the expected value of the commons neighbors.
 
 \begin{lemma}\label{lem:sep}
For any two vertices $u,v \in V_i: (u,v) \in E, i =1,2$ belonging to the same cluster with  $d_L(X_u,X_v) = x$, the count of common neighbors $C_{u,v} \equiv |\{z\in V: (z,u), (z,v) \in E\}|$ is a random variable distributed according to ${\rm Bin}(\frac{n}{2}-2, 2r_s-x)$ when $r_s \geq x> 2r_d$ and according to ${\rm Bin}(\frac{n}{2}-2,2r_s-x)+{\rm Bin}(\frac{n}{2},2r_d-x)$ when $x \leq \min(2r_d,r_s)$, where ${\rm Bin}(n,p)$ is a binomial random variable with mean $np$.
\end{lemma}
 
\begin{proof}
Without loss of generality, assume $u,v \in V_1$. For any vertex $z \in V$, let $\cE^{u,v}_z \equiv \{(u,z), (v,z) \in E\}$ be the event that $z$ is a common neighbor.
For $z\in V_1$,
\begin{align*}
&\Pr(\cE^{u,v}_z) =\Pr( (z,u) \in E, (z,v) \in E) \\
& = 2r_s - x,
\end{align*}
since $ \dist(X_u,X_v) =x$.
For $z \in V_2$, we have, 
\begin{align*}
&\Pr(\cE^{u,v}_z) =  \Pr( (z,u), (z,v) \in E ) \\
& =  \begin{cases}2r_d-x & \text{ if } x<2r_d \\ 0 & \text{ otherwise} \end{cases} .
\end{align*}
Now since there are $\frac{n}2-2$ points in $V_1 \setminus \{u,v\}$ and $\frac{n}2$ points in $V_2$, we have the statement of the lemma. 
\end{proof} 

In a similar way, we can prove.
 \begin{lemma}\label{lem:sep2}
For any two vertices $u\in V_1,v \in V_2: (u,v) \in E$ belonging to different clusters with $d_L(X_u,X_v) = x$ , the count of common neighbors $C_{u,v} \equiv |\{z\in V: (z,u), (z,v) \in E\}|$ is a random variable distributed according to ${\rm Bin}(n-2,2r_{d})$ when $r_s > 2r_d$ and according to ${\rm Bin}(n-2,\min(r_{s} + r_d -x,2r_d))$ when $r_s \leq 2r_d$ and $x\leq r_d$.
\end{lemma}
 


The distribution of the number of common neighbors given $(u,v)\in E$ and $d(u,v)=x$ is given in Table ~\ref{tab:tab1}. 
As throughout this paper, we have assumed that there are only two clusters of equal size. The functions change when the cluster sizes are different. 
 In the table, $u\sim v$ means $u$ and $v$ are in the same cluster. 

\begin{table*}[htbp]
\begin{center}
\resizebox{\textwidth}{!}{
\begin{tabular}{|c|p{3.2cm}|p{3cm}|p{3cm}|p{3.6cm}|} 
 \hline
$(u,v) \in E$ & \multicolumn{2}{c}{Distribution of count ($r_s>2r_d$)}&    \multicolumn{2}{c|}{Distribution of count ($r_s\le 2r_d$)} \\
 $d(u,v) =x$      & $u \sim v,  x\leq r_s$ & $u \nsim v, x\leq r_d$ &  $u \sim v, x\leq r_s$ & $u \nsim v, x\leq r_d$\\
 \hline
Motif : $z\mid (z,u)\in E, (z,v)\in E$   &  {\correct{${\rm Bin}(\frac{n}{2}-2,2r_s-x)+\mathbb{1}\{x\leq 2r_d\}{\rm Bin}(\frac{n}{2},2r_d-x)$}} & ${\rm Bin}(n-2,2r_{d})$ &  {\correct{${\rm Bin}(\frac{n}{2}-2,2r_s-x)+{\rm Bin}(\frac{n}{2},2r_d-x)$}}  & \correct{${\rm Bin}(n-2, \min(r_{s}+r_d-x,2r_d))$}\\
\hline
\end{tabular}}
\end{center}
\caption{Distribution of triangle count for an edge $(u,v)$ conditioned on the distance between them $d(u,v) = d_L(X_u, X_v) = x$, when there are two equal sized clusters. Here ${\rm Bin}(n,p)$ denotes a binomial random variable with mean $np$.\label{tab:tab1}}
\end{table*}

At this point note  that, in a ${\rm GBM}(r_s,r_d)$ for any edge $u,v$ that do not belong to the same part, the expected total number of common neighbors of  $u$ and $v$ does not depend on  their distance. We will next show that in this case the normalized total number of common neighbors is concentrated around $2r_d$. Therefore, when Algorithm~\ref{alg:alg1} finished removing all the edges, with high probability all the `inter-cluster' edges are removed. However, some of the `in-cluster' edges will also be removed in the process. This is similar to the case when from an ${\rm VRG}(n, [0,r_s])$, all the edges that correspond to a distance close to $2r_d$ has been removed. 
This situation is shown for the case when $r_s \ge 2r_d$ in Figure~\ref{fig:gbm}.
Finally we show that the edge-reduced  ${\rm VRG}(n, [0,r_s])$ is still connected under certain condition. In  what follows we will assume the ${\rm GBM}(r_s,r_d)$ with $r_s\ge 2r_d$. .The other case of $r_s < 2r_d$ is similar.

In the next lemma, we show a concentration result for the count made in \texttt{process}.

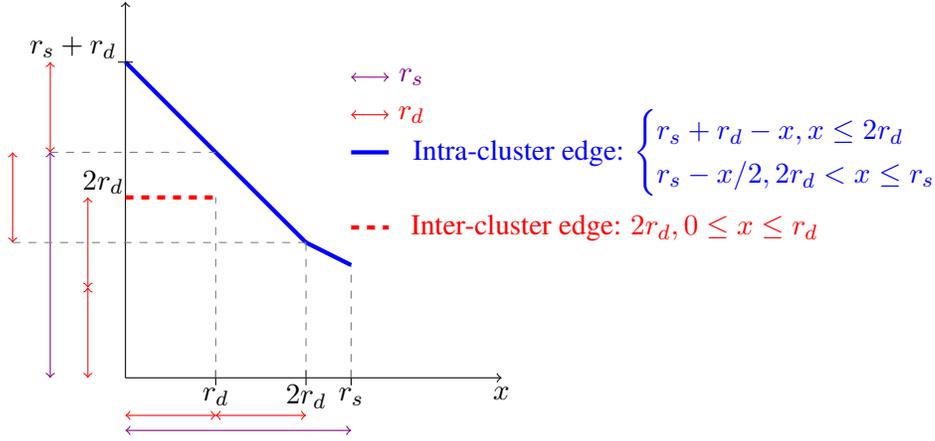
\begin{figure}
\centering
\begin{tikzpicture}


\draw[->] (0,0)--(5,0); 
\draw[->] (0,0)--(0,5); 
\draw[blue, ultra thick] (0,4.2)--(2.4,1.8); 
\draw[blue, ultra thick] (2.4,1.8)--(3,1.5); 
\draw[red, ultra thick, dashed] (0,2.4)--(1.2,2.4); 
\draw[red,<->] (0,-0.5)--(1.2,-0.5); 
\draw[red,<->] (1.2,-0.5)--(2.4,-0.5); 
\draw[violet, <->] (0,-0.7)--(3,-0.7); 
\draw[red,<->] (-0.5,0)--(-0.5,1.2); 
\draw[red,<->] (-0.5,1.2)--(-0.5,2.4); 
\draw[violet,<->] (-1,0)--(-1,3); 
\draw[red,<->] (-1,3)--(-1,4.2); 
\draw[red,<->] (-1.5,3)--(-1.5,1.8); 
\draw (1.2,0.1)--(1.2,-0.1)node at (1.2,-0.25) {$r_d$}; 
\draw (2.4,0.1)--(2.4,-0.1)node at (2.4,-0.25){$2r_d$}; 
\draw (3,0.1)--(3,-0.1)node at (3,-0.25){$r_s$}; 
\draw[gray,dashed] (3,0)--(3,1.5); 
\draw[gray,dashed] (1.2,0)--(1.2,2.4); 
\draw[gray,dashed] (2.4,0)--(2.4,1.8); 

\draw[gray,dashed] (-1.5,1.8)--(2.4,1.8); 
\draw[gray,dashed] (-1,3)--(1.2,3); 
\node at (-0.3,2.6){$2r_d$}; 
\node at (5,-0.2){$x$}; 
\draw (0.1,4.2)--(-0.1,4.2);
\node at (-0.7,4.4){$r_s+r_d$}; 
\draw[<->, red ] (3,3.5)--(3.5,3.5)node at (3.8,3.5) {$r_d$}; 
\draw[<->, violet ] (3,4)--(3.5,4)node at (3.8,4) {$r_s$}; 
\draw[blue, ultra thick ] (3,3)--(3.5,3)node at (7.5,3) {Intra-cluster edge: $ \begin{cases}
r_s+r_d-x, x\leq 2r_d\\
r_s-x/2, 2r_d< x\le r_s
\end{cases}$}; 
\draw[red, ultra thick, dashed ] (3,2)--(3.5,2)node at (6.5,2) {Inter-cluster edge: $2r_d, 0\le x\le r_d$}; 

\end{tikzpicture}
\caption{Average number of common neighbors of $(u,v) \in E$  for varying values of $d(u,v) =x$ when $r_s\ge2r_d$.\label{fig:gbm}}
\end{figure}

\begin{lemma}\label{lem:defn}
Suppose we are given the graph $G(V,E)$ generated according to ${\rm GBM}(r_s\equiv \frac{a\log n}{n},r_d\equiv \frac{b\log n}{n}), a \ge 2b.$ 
Our algorithm with $E_S = (2b + t_1)\frac{\log n}{n}$ and $E_D = (2b - t_2)\frac{\log n}{n}$, removes
all the edges $(u,v)\in E$ such that   $u$ and $v$ are in different parts with probability at least $1-o(1)$, where
\begin{align*}
t_1&=\min\{t: (2b+t)\log \frac{2b+t}{2b}-t > 1\} \\
t_2&=\min\{t : (2b-t)\log \frac{2b-t}{2b}+t > 1\}.
\end{align*}
\end{lemma} 

\begin{proof}
Here we will use the fact that for $a \geq 1$, the number of edges in ${\rm GBM}(r_s\equiv \frac{a\log n}{n},r_d\equiv \frac{b\log n}{n})$ is $O(n\log{n})$ with probability $1-\frac{1}{n^{\Theta(1)}}$. Consider any vertex $u \in V_1$ (symmetrically for $u \in V_2$), since the vertices are thrown uniformly at random in $[0,1]$, the probability that a $v \in V_1$, $v \neq u$, is a neighbor of $u$ is $\frac{a\log{n}}{n}$, and for $v \in V_2$, the corresponding probability is $\frac{b\log{n}}{n}$. Therefore, the expected degree of $u$ is $\frac{(a+b)}{2}\log{n}$. By a simple Chernoff bound argument, the degree of $u$ is therefore $O(\log{n})$ with probability $1-\frac{1}{n^c}$ for $c \geq 2$. By union bound over all the vertices, the total number of edges is $O(n\log{n})$ with probability $1-\frac{1}{n}$.

Let $Z$ denote the random variable that equals the number of common neighbors of two nodes $u,v \in V: (u,v) \in E$ such that $u,v$ are from different parts of the GBM. Using Lemma \ref{lem:sep2}, we know that $Z$ is sampled from the distribution ${\rm Bin}(n-2,2r_d)$, where $r_d = \frac{b\log n}{n}$. Therefore,
\begin{align*}
\Pr(Z \ge nE_S) \le \sum_{i=nE_S}^{n} {n \choose i}(2r_d)^{i}(n-2r_d)^{n-i} \le \exp\Big(-nD\Big((2b + t_1)\frac{\log n}{n}\|\frac{2b\log n}{n}\Big)\Big),
\end{align*}
where $D(p\|q)\equiv p\log\frac{p}{q}+ (1-p)\log \frac{1-p}{1-q}$ is the KL divergence between Bernoulli($p$) and Bernoulli($q$) distributions.
It is easy to see that,
\begin{align*}
nD(\frac{\alpha \log n}{n}||\frac{\beta \log n}{n})= \Big(\alpha\log\frac{\alpha}{\beta}+(\alpha-\beta)\Big)\log n-o(\log n).
\end{align*}
Therefore $\Pr(Z \ge nE_S) \le \frac{1}{n(\log n)^2}$ because 
$
(2b+t_1)\log\frac{2b+t_1}{2b}-t_1 > 1.
$
Similarly, we have that 
\begin{align*}
\Pr(Z \le nE_D) \le \sum_{i=0}^{nE_D} {n \choose i}(2r_d)^{i}(n-2r_d)^{n-i} \le \exp(-nD((2b-t)\frac{\log n}{n}\|\frac{2b\log n}{n}))\le \frac{1}{n(\log n)^2}.
\end{align*}
So all of the inter-cluster edges will be removed by Algorithm~\ref{alg:alg1}   with probability $1 - O(\frac{n \log n}{n(\log n)^2}) =1 -o(1)$, as with probability $1-o(1)$ the total number of edges in the graph is $O(n \log n)$.
\end{proof}
 
After Algorithm~\ref{alg:alg1}  finishes, in the edge-reduced ${\rm GBM}(\frac{a\log n}{n},\frac{b\log n}{n})$, all the edges are `in-cluster' edges with high probability. However some of the  
`in-cluster' edges are also deleted, namely, those that has a count of common neighbors between $E_S$ and $E_D$. In the next two lemmas, we show the necessary condition on  the `in-cluster' edges such that they do not get removed by Algorithm~\ref{alg:alg1}.

\begin{lemma}
Suppose we have the graph $G(V,E)$ generated according to ${\rm GBM}(r_s \equiv \frac{a\log n}{n},r_d\equiv \frac{b\log n}{n}), a \ge 2b$. Define $t_1,t_2, E_D,E_S$ as in  Lemma~\ref{lem:defn}. Consider an edge $(u,v) \in E$ where $u,v$ belong to the same part of the GBM and let  $d(u,v)\equiv x \equiv \frac{\theta \log n}{n}$. Suppose $\theta$ satisfies {\bf either} of the following conditions:
\begin{enumerate}
\item $
\frac{1}{2}\Big((4b+2t_1)\log \frac{4b+2t_1}{2a-\theta}+2a-\theta-4b-2t_1\Big) > 1 \quad \text{ and } \theta \le 2a-4b-2t_1
$
\item $
\frac{1}{2}\Big((4b-2t_2\log \frac{4b-2t_2}{2a-\theta}+2a-\theta-4b+2t_2\Big) > 1
\quad \text{and} \quad 
a \ge \theta \ge \max\{2b,2a-4b+2t_2\}.
$
\end{enumerate}
Then Algorithm~\ref{alg:alg1} with $E_S  = (2b +t_1)\frac{\log n}{n}$ and $E_D  = (2b - t_2)\frac{\log n}{n}$ will not remove this edge with probability at least $1- O(\frac1{n (\log n)^2})$.
%
%
\label{lem:esed}
\end{lemma}
\begin{proof}
Let $Z$ be the number of common neighbors of $u,v$. 
Recall that, $u$ and $v$ are in the same cluster. 
We know from Lemma \ref{lem:sep2} that $Z$ is sampled from the distribution ${\rm Bin}(\frac{n}{2}-2,2r_s-x)+{\rm Bin}(\frac{n}{2},2r_d-x)$ when $x \le 2r_d$, and from the distribution ${\rm Bin}(\frac{n}{2}-2,2r_s-x)$ when $x \ge 2r_d$. 
We have,
\begin{align*}
&\Pr(Z \le nE_S)\\
&=
\begin{cases}
\sum_{i=0}^{nE_S}{\frac{n}{2}-2 \choose i} (2r_s-x)^{i}(1-2r_s+x)^{\frac{n}{2}-i-2}\sum_{j=0}^{nE_S-i}{\frac{n}{2} \choose j} (2r_d-x)^{j}(1-2r_d+x)^{\frac{n}{2}-j} \text{ if $x\le 2r_d$ }\\
\sum_{i=0}^{nE_s}{\frac{n}{2}-2 \choose i} (2r_s-x)^{i}(1-2r_s+x)^{\frac{n}{2}-i} \text{ otherwise} \\
\end{cases} \\
& \le e^{-\frac{n}{2}D(2E_S|| \frac{(2a-\theta)\log n}{n})} \text{ since  } 2a-\theta \ge 4b+2t_1\\
& \le e^{-\frac{n}{2}D(\frac{(4b+2t_1)\log n}{n}|| \frac{(2a-\theta)\log n}{n})} \le \frac{1}{n \log^2 n},
\end{align*}
because of condition 1 of this lemma. Therefore, this edge will not be deleted with high probability.

Similarly, let us find the probability of $Z \ge n E_D = (2b-t_2) \log n.$
Let us just assume the worst case when  $\theta \le 2b$: that the edge is being deleted (see condition 2, this is prohibited if that condition is satisfied).
%
Otherwise, $\theta > 2b$ and,
%
\begin{align*}
\Pr(Z \ge nE_D)&=\sum_{i=nE_D}^{n}{\frac{n}{2}-2 \choose i} (2r_s-x)^{i}(1-2r_s+x)^{\frac{n}{2}-i-2} \\
&\le e^{-\frac{n}{2}D(2E_D\|\frac{(2a-\theta)\log n}{n})} \text{   if   } 2a-\theta \le 4b-2t_2\\
&= e^{-\frac{n}{2}D(\frac{(4b-2t_2)\log n}{n}\|\frac{(2a-\theta)\log n}{n})} \le \frac{1}{n \log^2 n},
\end{align*}
because of condition 2 of this lemma.
\end{proof}

Now we are in a position to prove our main theorem from this part. Let us restate this theorem.

\begin{theorem*}(\ref{gbm:upper})~
Suppose we have the graph $G(V,E)$ generated according to ${\rm GBM}(r_s \equiv \frac{a\log n}{n},r_d\equiv \frac{b\log n}{n}), a \ge 2b$. 
Define,
\begin{align*}
t_1&=\min\{t: (2b+t)\log \frac{2b+t}{2b}-t > 1\} \\
t_2&=\min\{t : (2b-t)\log \frac{2b-t}{2b}+t > 1\}\\
\theta_1 &= \max\{\theta:\frac{1}{2}\Big((4b+2t_1)\log \frac{4b+2t_1}{2a-\theta}+2a-\theta-4b-2t_1\Big) > 1  \text{ and } 0 \le \theta \le 2a-4b-2t_1\}\\
\theta_2 &= \min\{\theta: \frac{1}{2}\Big((4b-2t_2\log \frac{4b-2t_2}{2a-\theta}+2a-\theta-4b+2t_2\Big) > 1
 \text{ and }  
a \ge \theta \ge \max\{2b,2a-4b+2t_2\}\}.
\end{align*}
Then Algorithm~\ref{alg:alg1} with $E_S  = (2b +t_1)\frac{\log n}{n}$ and $E_D  = (2b - t_2)\frac{\log n}{n}$ will recover the correct partition  in the GBM with  probability $1-o(1)$  if   $a-\theta_2+\theta_1>2$ OR $a-\theta_2 > 1, a>2$. 
\end{theorem*}
\begin{proof}
From Lemma~\ref{lem:defn}, we know that after Algorithm~\ref{alg:alg1} goes over all the edges, the edges with end-points being in different parts of the GBM are all removed with probability $1-o(1)$. 
There are $O(n \log n)$ edges in the GBM with  probability $1-o(1).$
From Lemma~\ref{lem:esed}, 
we can say that no edge with both ends at the same part is deleted with probability at least $1-o(1)$ (by simply applying a union bound).

After Algorithm~\ref{alg:alg1} goes over all the edges, the remaining edges from a disjoint union of two vertex-random graphs of $\frac{n}{2}$ vertices each. For any two vertices $u,v$ in the same part, there will be an edge if $d(u,v) \in [0,\theta_1] \cup [\theta_2, a]$. From Corollary \ref{cor:patch}, it is evident that each of these two parts (each part  is of size $\frac{n}{2}$) will be connected if either $a-\theta_2+\theta_1>2$ or $a-\theta_2 > 1, a>2$.
\end{proof}


 It is also possible to incorporate the result of Corollary \ref{cor:extra} as well to get somewhat stronger recovery guarantee for our algorithm.

\sloppy
\section{High Dimensional GBM: Proof of Theorem \ref{theorem:intro-1}}
\label{sec:sparse-high}
In this section, we show that our algorithm for recovery of clusters in GBM, i.e., Algorithm \ref{alg:alg1} extends to higher dimensions. Let us define a higher dimensional GBM precisely.
\begin{definition}[The GBM in High Dimensions]
Given $V = V_1\sqcup V_2, |V_1|=|V_2| = \frac{n}2$,  choose a random vector $X_u$ independently uniformly distributed in $S^t$ for all $u \in V$.
The geometric block model  ${\rm GBM}_t(r_s, r_d)$ with parameters $r_s> r_d$ is a random graph where an edge exists between vertices $u$ and $v$  if and only if,
\begin{align*}
\norm{X_u-X_v}_2 \le r_s & \text{ when } u, v \in V_1 \text{ or } u,v \in V_2\\
\norm{X_u-X_v}_2 \le r_d & \text{ when } u \in V_1, v \in V_2 \text{ or } u\in V_2,v \in V_1.
\end{align*}
\end{definition}

Indeed, for the higher dimensional case the algorithm remains exactly the same, except the value of $E_D$ and $E_S$ in the subroutine \texttt{process} needs to be changed.
Recall that the algorithm proceeds by checking each edge and counting the number of triangle the edge is part of. 
If the count is between $E_D$ and $E_S$ the edge is removed. In this process  we claim to remove all inter-cluster edges with high probability. The main difficulty lies in proving that the original communities remain connected in the redacted graph. For that we crucially use the connectivity results of RAG (from Section~\ref{sec:hrag}) in somewhat different way that what we do for the one dimensional case.

%

In the following, we fix the dimension of the GBM as $t$ and hence remove the subscript from all the notations defined above in order to make them less cumbersome. 

\subsection{Analysis of Algorithm \ref{alg:alg1} in High Dimension}

Let us define a few more terminologies to simplify the expressions for high dimensional space. The volume of a $t$-sphere with unit radius is $|S_t|\equiv a_t=\frac{2\pi^{t+1/2}}{\Gamma(\frac{t+1}{2})}$. Let the spherical cap ${B}_t(O,r) \subset S^t$ define a region on the surface of this $t$-sphere $S^t$ such that every point $u
\in {B_t}(O,r)$ satisfies $\|u-O\|_2 \le  r$. Let us denote the volume of the spherical cap ${B_t}(O,r)$ normalized by $a_t$  by
$B_t(r)$. Similarly ${B_t}(O,[r_1,r_2])$ refers to a region on the $t$-sphere such that every point $u \in {B}_t(O,[r_1,r_2])$ satisfies $r_1 \le \|u-O\|_2 \le r_2$ and ${B_t}(r_1,r_2)$ refers to the volume normalized by $a_t$.  Now consider two such spherical caps ${B_t}(O_1,r_1)$ and ${B}_t(O_2,r_2)$ such that $d(O_1,O_2)=\ell$. In that case let us define the volume of the intersection of the two aforementioned spherical caps (again normalized by $a_t$) by $\mathcal{V}_t(r_1,r_2,\ell)$.

Let us use $u \sim v$ ($u \nsim v$) to denote $u$ and $v$ belong to the same cluster (different clusters). Let $\cE^{u,v}_z$ denote the event that $z$ is a common neighbor of $u$ and $v$ and $e(u,v)$ denote the event that there is an edge between $u$ and $v$.  Following are some simple observations.

\begin{observation}
\label{obs:1}
$\Pr(e(u,v) \mid u\sim v)=B_t(r_s)$ and $\Pr(e(u,v) \mid u\nsim v)=B_t(r_d)$.
\end{observation}
\begin{observation}
\label{obs:2}
$\Pr(\cE^{u,v}_z \mid z\sim u, u\sim v \text{ and } \|u-v\|_2=\ell)=\cV_t(r_s,r_s,\ell)$ and $\Pr(\cE^{u,v}_z \mid z\nsim u, u\sim v \text{ and } \|u-v\|_2=\ell)=\cV_t(r_d,r_d,\ell)$.
\end{observation}

In the following proof, we assume $r_s \leq 2r_d$. The other situation where the gap between $r_s$ and $r_d$ is higher is only easier to handle.

\begin{lemma}\label{lem:sepd}
For any two vertices $u,v \in V_i: (u,v) \in E, i =1,2$ such that $d(u,v)=\ell$ belonging to the same cluster, the count of common neighbors $C_{u,v} \equiv |\{z\in V: (z,u), (z,v) \in E\}|$ is a random variable distributed according to ${\rm Bin}(\frac{n}{2}-2,\mathcal{V}_t(r_s,r_s,\ell))$  when $r_s \geq \ell> 2r_d$ and according to ${\rm Bin}(\frac{n}{2}-2,\mathcal{V}_t(r_s,r_s,\ell))+{\rm Bin}(\frac{n}{2},\mathcal{V}_t(r_d,r_d,\ell)$ when $\ell \le 2r_d$.
\end{lemma}
\begin{lemma}\label{lem:sep2d}
For any two vertices $u\in V_1,v \in V_2: (u,v) \in E$ such that $\|u-v\|_2=\ell$ belonging to different clusters, the count of common neighbors $C_{u,v} \equiv |\{z\in V: (z,u), (z,v) \in E\}|$ is a random variable distributed according to ${\rm Bin}(n-2,B_t(r_d))$ when $r_s > 2r_d$ and according to ${\rm Bin}(n-2,\min(\mathcal{V}_t(r_s,r_d,\ell),B_t(r_d)))$ when $r_s \leq 2r_d$ and $\ell \leq r_d$.
\end{lemma}
\begin{proof}[Proof of Lemma \ref{lem:sepd}]
Without loss of generality, assume $u,v \in V_1$. In order for $(u,v) \in E$, we must have $r_s \geq \ell$. Now there are two cases to consider, $\ell > 2r_d$ and $\ell \leq 2r_d$. In case 1, for $z$ to be a common neighbor of $u$ and $v$, $z$ must be in $V_1$ by triangle inequality. Since, there are $\frac{n}{2}-2$ points in $V_1 \setminus \{u,v\}$, from Observation~\ref{obs:2}, $\cE^{u,v}_z) \sim {\rm Bin}(\frac{n}{2}-2,\mathcal{V}_t(r_s,r_s,\ell))$. In case 2, $z$ can also be part of $V_2$ and there are $\frac{n}{2}$ points in $V_2$, thus again from Observation~\ref{obs:2}, $\cE^{u,v}_z) \sim
{\rm Bin}(\frac{n}{2}-2,\mathcal{V}_t(r_s,r_s,\ell))+{\rm Bin}(\frac{n}{2},\mathcal{V}_t(r_d,r_d,\ell)$. 
\end{proof}
The proof of Lemma \ref{lem:sep2d} is similar. We now use the following version of the Chernoff bound to estimate the deviation on the number of common neighbors in the two cases: $u \sim v$ and $u \nsim v$.

\begin{lemma}[Chernoff Bound]
Let $X_1, \ldots, X_n$ be iid random variables in $\{0,1\}$. Let $X$ denote the sum of these $n$ random variables. Then for any $\delta>0$,
\[
 \begin{cases}
\Pr(X > (1+\delta) \avg(X)) \le e^{-\delta^2\avg(X)/3}  = \frac{1}{n\log^2 n}, \text{ when } \delta = \sqrt{\frac{3(\log{n}+2\log\log n)}{\avg(X)}},\\
\Pr(X < (1-\delta)\avg(X)) \leq e^{-\delta^2\avg(X)/2} = \frac{1}{n\log^2 n},\text{ when } \delta = \sqrt{\frac{2(\log{n}+2\log\log n)}{\avg(X)}}.
\end{cases}
\]
\end{lemma}

We take $E_S = c^{(t)}_s\cdot (B_t(r_d)n+\sqrt{6B_t(r_d) n\log n})$ and $E_D=c^{(t)}_d \cdot( n\mathcal{V}_t(r_s,r_d,r_d)-\sqrt{2nB_t(r_d)\log n })$ where $c^{(t)}_s \geq 1$ and $c^{(t)}_d \leq 1$ are suitable constants that depend on $t$.

\begin{lemma}
For any pair of nodes $(u,v)=e\in E,\  u\nsim v$, the \texttt{MotifCount} algorithm removes the edge $e$ with a probability of $1-O\left(\frac{1}{n\log^2 n}\right)$ when $E_S \geq B_t(r_d)n+\sqrt{6B_t(r_d) n\log n}$ and $E_D \leq n\mathcal{V}_t(r_s,r_d,r_d)-\sqrt{2nB_t(r_d)\log n }$. \label{lem:diff}
\end{lemma}

\begin{proof}
Let $Z$ denote the random variable for the number of common neighbors of two nodes $u,v \in V: (u,v) \in E, \|u-v\|_2 = \ell, u \nsim v$. From Lemma~\ref{lem:sep2d}, $E[Z]\leq (n -2)B_t(r_d)$. Using the Chernoff bound we know that with a probability of at least $1-\frac{1}{n\log^2 n}$
\begin{align*}
Z\le F_{\nsim}=&\  (n -2)B_t(r_d) + \sqrt{3(\log n+ 2\log \log n)(n-2)B_t(r_d) }\\
=& \  B_t(r_d)n+\sqrt{3B_t(r_d) n\log n}+o(1)\\
\le& \  E_S.
\end{align*}

Moreover again from Lemma~\ref{lem:sep2d}, $E[Z]=(n-2)\min(\mathcal{V}_t(r_s,r_d,\ell),B_t(r_d))$ as we assume $r_s \leq 2r_d$. Hence, with probability of at least  $1-\frac{1}{n\log^2 n}$

\begin{align*}
Z\ge f_{\nsim}= &\ \min_{\ell: \ell \le r_d,r_s \le 2r_d} ( ( n -2)\min(\mathcal{V}_t(r_s,r_d,\ell),B_t(r_d)) - \\
  & \sqrt{2(\log n+ 2\log \log n)(n-2)\min(\mathcal{V}_t(r_s,r_d,\ell),B_t(r_d)) }) \\
&\geq \ \min_{\ell: \ell \le r_d,r_s \le 2r_d} ( ( n -2)\min(\mathcal{V}_t(r_s,r_d,\ell),B_t(r_d)) - \\
  & \sqrt{2(\log n+ 2\log \log n)(n-2)B_t(r_d) }) \,\text{since } \cV_t(r_s,r_d,\ell) \subseteq B_t(r_d) \\
>&\  n\mathcal{V}_t(r_s,r_d,r_d)-\sqrt{2nB_t(r_d)\log n }\, \text{since } \cV_t(r_s,r_d,\ell) \text{ is a decreasing function of } \ell \\
\geq & \  E_D.
\end{align*}

Hence, $E_S \le Z\le E_D$ with a probability of $1-\frac{2}{n\log^2 n}$ for $(u,v)\in E, u \nsim v$. Hence $(u,v)$ gets removed with high probability by the algorithm.
\end{proof}
Applying a union bound, we therefore can assume all inter-cluster edges are removed with probability $1-o(1)$ as there is $O(n\log{n})$ edges. 

In the next two lemmas, we provide two different conditions on $\|u-v\|_2$ when $u \sim v$ such that our algorithm does not remove the edge $(u,v)$.  Then we obtain a sufficient condition for the two communities to remain connected by the edges that are not removed. 


\begin{lemma}
\label{lem:essame-1}
Given a pair of nodes $u,v$ belonging to the same cluster such that $(u,v) \in E$, the \texttt{MotifCount} algorithm does not remove the edge $e$ with probability of $1-O\left( \frac{1}{n\log^2 n}\right)$ when $\|u-v\|_2=\ell$ (say) satisfies the following:
\begin{align*}
\frac{n}{2}\Big(\mathcal{V}_t(r_s,r_s,\ell)+\mathcal{V}_t(r_d,r_d,\ell)\Big)- \sqrt{2n\log n}\Big(\sqrt{B_t(r_s)} + \sqrt{B_t(r_d)} \Big) > E_S. 
\end{align*}  
\end{lemma}
\begin{proof}
Let $Z$ denote the random variable corresponding to the number of common neighbors of $u,v$. Let $\mu_s(\ell)= \avg(Z|u \sim v, d(u,v)=\ell)$. From Lemma~\ref{lem:sepd}, $\mu_s(\ell)=(\frac{n}{2}-2)\cV_t(r_s,r_s,l)+\frac{n}{2}\cV_t(r_d,r_d,l)$.

Using the Chernoff bound, with a probability of $1-O\Big(\frac{1}{n\log^2 n}\Big)$

\begin{align*}
Z &>  
 (n/2-2)\mathcal{V}_t(r_s,r_s,\ell)+n/2\mathcal{V}_t(r_d,r_d,\ell)  - \sqrt{2(\log n+2\log\log n)\mathcal{V}_t(r_s,r_s,\ell)n/2}    \\&- \sqrt{2(\log n+2\log\log n)\mathcal{V}_t(r_d,r_d,\ell)n/2}\\
&\ge n/2\mathcal{V}_t(r_s,r_s,\ell)+n/2\mathcal{V}_t(r_d,r_d,\ell) -\Big(\sqrt{B_t(r_s)} + \sqrt{B_t(r_d)} \Big)\sqrt{2n\log n}\\
&=\frac{n}{2}\Big(\mathcal{V}_t(r_s,r_s,\ell)+\mathcal{V}_t(r_d,r_d,\ell)\Big)- \sqrt{2n\log n}\Big(\sqrt{B_t(r_s)} + \sqrt{B_t(r_d)} \Big).
\end{align*}

Therefore,  Algorithm \ref{alg:alg1} will not delete $e$ if 
\begin{align*}
\frac{n}{2}\Big(\mathcal{V}_t(r_s,r_s,\ell)+\mathcal{V}_t(r_d,r_d,\ell)\Big)- \sqrt{2n\log n}\Big(\sqrt{B_t(r_s)} + \sqrt{B_t(r_d)} \Big) > E_S.
\end{align*}

Note that there exists a maximum value of distance (referred to as $\ell_1$) such that whenever $\|u-v\| \leq \ell_1$, the condition will be satisfied.
\end{proof}

%


\begin{lemma}
\label{lem:edsamed-1}
Given a pair of nodes $u,v$ belonging to the same cluster such that $(u,v) \in E$, the \texttt{MotifCount} algorithm does not remove the edge $e$ with probability of $1-O\left( \frac{1}{n\log^2 n}\right)$ when $\ell \equiv \|u-v\|_2$ (say) satisfies the following:
$$\frac{n}{2}\Big(\mathcal{V}_t(r_s,r_s,\ell+\mathcal{V}_t(r_d,r_d,\ell)\Big))+ \sqrt{n\log{n}}\sqrt{[\mathcal{V}_t(r_s,r_s,\ell)+\mathcal{V}_t(r_d,r_d,\ell)]} \leq E_D.$$\end{lemma}
\begin{proof}

Let $Z$ denote the random variable corresponding to the number of common neighbors of $u,v$. Let $\mu_s(\ell)= \avg(Z|u \sim v, \|u-v\|_2=\ell)$. 
From Lemma~\ref{lem:sepd}, $\mu_s(\ell)=(\frac{n}{2}-2)\cV_t(r_s,r_s,l)+\frac{n}{2}\cV_t(r_d,r_d,l)$.

Using the Chernoff bound, with a probability of $1-O\Big(\frac{1}{n\log^2 n}\Big)$
\begin{align*}
Z &<  n/2\mathcal{V}_t(r_s,r_s,\ell)+n/2\mathcal{V}_t(r_d,r_d,\ell)  + \sqrt{2(\log n+2\log\log n)[\mathcal{V}_t(r_s,r_s,\ell)+\mathcal{V}_t(r_d,r_d,\ell)](n/2)}   \\
 &\le \frac{n}{2}\Big(\mathcal{V}_t(r_s,r_s,\ell+\mathcal{V}_t(r_d,r_d,\ell)\Big))+ \sqrt{n\log{n}}\sqrt{[\mathcal{V}_t(r_s,r_s,\ell)+\mathcal{V}_t(r_d,r_d,\ell)]}.
\end{align*}
The  \texttt{MotifCount} algorithm will not remove $e$ if 
$$\frac{n}{2}\Big(\mathcal{V}_t(r_s,r_s,\ell+\mathcal{V}_t(r_d,r_d,\ell)\Big))+ \sqrt{n\log{n}}\sqrt{[\mathcal{V}_t(r_s,r_s,\ell)+\mathcal{V}_t(r_d,r_d,\ell)]} \leq E_D.$$ 
Note that there exists a minimum value of distance (referred to as $\ell_2$) such that whenever $\|u-v\|_2 \geq \ell_2$, the condition will be satisfied.
\end{proof}

%

\begin{lemma}\label{lem:main1d}
Algorithm \ref{alg:alg1} can identify all edges $(u,v)$ correctly for which $\ell\equiv \|u-v\|_2$ satisfies either of  the following:
\begin{align*}
{\rm Cond. 1:} & \quad \frac{n}{2}\Big(\mathcal{V}_t(r_s,r_s,\ell)+\mathcal{V}_t(r_d,r_d,\ell)\Big)- \sqrt{2n\log n}\Big(\sqrt{B_t(r_s)} + \sqrt{B_t(r_d)} \Big) > E_S   
\end{align*}
or
\begin{align*}
{\rm Cond. 2:} & \quad \frac{n}{2}\Big(\mathcal{V}_t(r_s,r_s,\ell+\mathcal{V}_t(r_d,r_d,\ell)\Big))+ \sqrt{n\log{n}}\sqrt{[\mathcal{V}_t(r_s,r_s,\ell)+\mathcal{V}_t(r_d,r_d,\ell)]} \leq E_D. 
\end{align*}
with probability at least $1-O\Big(\frac{1}{\log n}\Big)$.
\end{lemma}

\begin{proof}
Follows from combining Lemma~\ref{lem:essame-1} and Lemma \ref{lem:edsamed-1}, and noting that in the connectivity regime, the number of edges is $O(n\log{n})$. 
\end{proof}

Let $\ell_1$ be the maximum value of $\|u-v\|_2$ such that Cond 1 is satisfied and $\ell_2$ is the minimum value of $\|u-v\|_2$ such that Cond 2 is satisfied. Also note that $\ell_1 \leq \ell_2$. We now give a condition on $\ell_1$ and $\ell_2$ such the two communities are each connected by the edges $(u,v)$ that satisfy either $\|u-v\|_2 \leq \ell_1$ or $\ell_2 \leq \|u-v\|_2 \leq r_s$.

\begin{lemma}
\label{RAG:cond}
If $(\ell_1/2)^t >  {8(t+1)\psi(t)} \frac{ \log n}{ n}$ 
 then the edges $e$ that satisfy $\|u-v\|_2 \leq \ell_1$
  constitute two disjoint connected components corresponding to the two original communities.
\end{lemma}
\begin{proof}
Proof of this lemma follows from the result of connectivity of random annulus graphs (RAG) in dimension $t$, i.e., Theorem~\ref{thm:highdem1}. 
\end{proof}

We now find out the values of $r_s$ and $r_d$ such that $\ell_1$ and $\ell_2$ satisfy the condition of Lemma~\ref{RAG:cond} as well as Cond. 1 and Cond 2. respectively.

\begin{theorem}
Algorithm \ref{alg:alg1} recovers the clusters with probability $1-o(1)$ if $r_s=\Theta((\frac{\log{n}}{n})^{\frac{1}{t}})$ and $r_s-r_d=\Omega( (\frac{\log n}{n})^{\frac{1}{t}})$.
\end{theorem}
\begin{proof}
Let us take $r_s=a_t \Big(\frac{\log n}{n}\Big)^{\frac{1}{t}}$, $r_d \leq b_t \Big(\frac{\log n}{n}\Big)^{\frac{1}{t}}$ for some large constants $a_t$ and $b_t$ that depends on $t$. Then to satisfy Lemma~\ref{RAG:cond}, we can take $\ell_1=a'_t \Big(\frac{\log n}{n}\Big)^{\frac{1}{t}}$ and $\ell_2=b'_t \Big(\frac{\log n}{n}\Big)^{\frac{1}{t}}$ again for suitable constants $a'_t$ and $b'_t$. While it is possible to concisely compute 
$B_t(r)$ and $\cV_t(r_1,r_2,x)$ \cite{li2011concise,ellis2007random}, for the purpose of analysis it is sufficient to know $B_t(r)=\Theta(r^t)$ for fixed $t$. Moreover, $\cV_t(r_1,r_2,x)=\Theta((r_1+r_2-x)^t)$ if $r_1+r_2\geq x \geq \max(r_1,r_2)$, $\Theta(\min\{r_1,r_2\}^t)$ if $x \leq \max(r_1,r_2)$ and $0$ otherwise.

Then, Cond 1. requires 
$$\frac{n}{2}\Big(\mathcal{V}_t(r_s,r_s,\ell_1)+\mathcal{V}_t(r_d,r_d,\ell_1)\Big)- \sqrt{2n\log n}\Big(\sqrt{B_t(r_s)} + \sqrt{B_t(r_d)} \Big) > E_S $$
and Cond 2. requires
$$\frac{n}{2}\Big(\mathcal{V}_t(r_s,r_s,\ell_2+\mathcal{V}_t(r_d,r_d,\ell_2)\Big))+ \sqrt{n\log{n}}\sqrt{[\mathcal{V}_t(r_s,r_s,\ell)+\mathcal{V}_t(r_d,r_d,\ell)]} \leq E_D.$$

By selecting the constants $c^{(t)}_s$ and $c^{t}_d$ involving $E_S$ and $E_D$ suitably, both the conditions are satisfied.
\end{proof}

\bibliographystyle{abbrv}
\bibliography{references} 

\appendix
\section{Omitted Proofs}

\begin{proof}Consider a node $u$ and assume that the position of $u$ is $0$. Associate a random variable $A_u^i$ for $i \in \{1,2,3,4\}$ which takes the value of $1$ when there does not exist any node $x$ such that 

\begin{figure}[h]
\centering
\begin{tikzpicture}[thick, scale=0.9]
\draw [black,<->] (-0.5,1.8) -- (-1,1.8)node at (0.5,1.8) {$a-b-\frac{c}{2}$};;
\draw [gray,<->] (-0.5,1.3) -- (-1,1.3)node at (0.5,1.3) {$\frac{c}{2}$};;
\draw[pattern=north west lines,pattern color=red][preaction={fill=blue!20}] (3.5,1.9) -- (4.5,1.9) -- (4.5,1.7) -- (3.5,1.7) -- cycle node at (4.8,1.8) {$A_u^3$};;
\draw[pattern=grid,pattern color=blue][preaction={fill=red!50}]   (3.5,1.4) -- (4.5,1.4) -- (4.5,1.2) -- (3.5,1.2) -- cycle node at (4.8,1.2) {$A_u^4$};;
\draw[fill=blue,blue]   (1.5,1.9) -- (2.5,1.9) -- (2.5,1.7) -- (1.5,1.7) -- cycle node at (2.8,1.8) {$A_u^1$};;
\draw[pattern=north east lines,pattern color=violet][preaction={fill=green!30}] (1.5,1.4) -- (2.5,1.4) -- (2.5,1.2) -- (1.5,1.2) -- cycle node at (2.8,1.2) {$A_u^2$};;

\draw[fill=blue,blue]  (3.5,0) -- (4.5,0) -- (4.5,-0.3) -- (3.5,-0.3) -- cycle;
\draw[fill=blue,blue]  (0,0) -- (1,0) -- (1,-0.3) -- (0,-0.3) -- cycle;
\draw[fill=blue,blue]  (-0.5,0) -- (-1,0) -- (-1,-0.3) -- (-0.5,-0.3) -- cycle;

\draw[pattern=north east lines,pattern color=violet][preaction={fill=green!30}] (3.5,-0.6) -- (4.5,-0.60) -- (4.5,-0.3) -- (3.5,-0.3) -- cycle;
\draw[pattern=north east lines,pattern color=violet][preaction={fill=green!30}]  (0,-0.6) -- (1,-0.6) -- (1,-0.3) -- (0,-0.3) -- cycle;
\draw[pattern=north east lines,pattern color=violet][preaction={fill=green!30}]  (3,0) -- (3.5,0) -- (3.5,-0.3) -- (3,-0.3) -- cycle;

\draw[pattern=north west lines,pattern color=red][preaction={fill=blue!20}] (-3.5,0) -- (-4.5,0) -- (-4.5,-0.3) -- (-3.5,-0.3) -- cycle;
\draw[pattern=north west lines,pattern color=red][preaction={fill=blue!20}] (0,-0.60) -- (-1,-0.60) -- (-1,-0.3) -- (0,-0.3) -- cycle;
\draw[pattern=north west lines,pattern color=red][preaction={fill=blue!20}] (0.5,-0.6) -- (1,-0.6) -- (1,-0.3) -- (0.5,-0.3) -- cycle;

\draw[pattern=grid,pattern color=blue][preaction={fill=red!50}](-3.5,-0.6) -- (-4.5,-0.6) -- (-4.5,-0.3) -- (-3.5,-0.3) -- cycle;
\draw[pattern=grid,pattern color=blue][preaction={fill=red!50}] (0,-0.6) -- (-1,-0.6) -- (-1,-0.9) -- (0,-0.9) -- cycle;
\draw[pattern=grid,pattern color=blue][preaction={fill=red!50}]  (-3,-0.3) -- (-3.5,-0.3) -- (-3.5,-0.6) -- (-3,-0.6) -- cycle;

\draw [gray,<->] (0,-1.3) -- (0.5,-1.3);
\draw [gray,<->] (0,-1.3) -- (-.5,-1.3);
\draw [black,<->] (-3.50,-1.3) -- (-3,-1.3);
\draw [black,<->] (3.50,-1.3) -- (3,-1.3);

\draw [gray,<->] (-5,0) -- (5,0);
\draw [gray,dashed] (0.5,0) -- (0.5,-1.4);
\draw [gray,dashed] (0,0) -- (0,-1.4);
\draw [gray,dashed] (-0.5,0) -- (-0.5,-1.4);
\draw [gray,dashed] (-3,0) -- (-3,-1.4);
\draw [gray,dashed] (-3.5,0) -- (-3.5,-1.4);
\draw [gray,dashed] (3.5,0) -- (3.5,-1.4);


\draw [gray,dashed] (3,0) -- (3,-1.4);
\draw node[anchor=south] at (3.5,0) {$b$};
\draw node[anchor=south] at (4.5,0) {$a$};
\draw node[anchor=south] at (-3.5,0) {$-b$};
\draw node[anchor=south] at (-4.5,0) {$-a$};
\draw node[anchor=south] at (-1,0) {$-c$};

\draw node[anchor=south] at (1,0) {$c$};

%
%
%
%
\filldraw [black] (0,0) circle (2pt)node[anchor=south] at (0,0) {$u$};
\end{tikzpicture}
\caption{The representation of different intervals corresponding to each random variable as described in Corollary \ref{cor:extra}}
\end{figure}
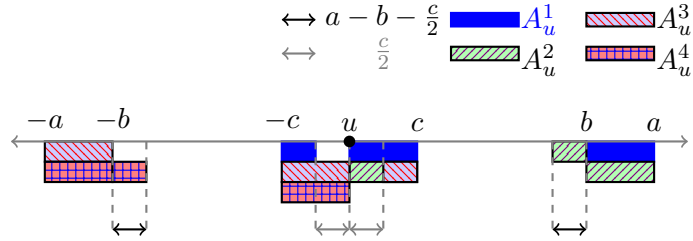

\begin{enumerate}
\item $d(u,x) \in [b\frac{\log n}{n},a\frac{\log n}{n}] \cup [0,c\frac{\log n}{n}]  \cup [\frac{-c\log n}{n},\frac{-c/2\log n}{n}]\} \text{ for } i=1$
\item $d(u,x) \in [b\frac{\log n}{n},a\frac{\log n}{n}] \cup [0,c\frac{\log n}{n}]  \cup [\frac{b-c/2\log n}{n},\frac{(a-c)\log n}{n}]\} \text{  for  } i=2$
\item $d(u,x) \in [-a\frac{\log n}{n},-b\frac{\log n}{n}] \cup [-c\frac{\log n}{n},0]  \cup  [\frac{c/2\log n}{n},\frac{c\log n}{n}]\} \text{  for  } i=3$
\item $d(u,x) \in [-a\frac{\log n}{n},-b\frac{\log n}{n}] \cup [-c\frac{\log n}{n},0]  \cup [\frac{(c-a)\log n}{n},\frac{(c/2-b)\log n}{n}]\} \text{  for  } i=4$
\end{enumerate}

\begin{align*}
\Pr(A_u^i=1)&=
\begin{cases}
\Big(1-(c+a-b+(a-b-c/2))\frac{\log n}{n}\Big)^{n} \text{ when }a-c<b \text{ and } b-c/2 > c\\
\Big(1-(b-c)\frac{\log n}{n}\Big)^{n} \text{ when }a-c<b \text{ and } b-c/2 < c\\
\Big(1-(a)\frac{\log n}{n}\Big)^{n} \text{ when }a-c\ge b \text{ and } b-c/2 < c\\
\Big(1-(c+a-b+c/2)\frac{\log n}{n}\Big)^{n} \text{ when }a-c \ge b \text{ and } b-c/2 \ge c\\
\end{cases}
\end{align*}
Notice that $A_u^1$ and $A_u^2$ being zero implies that either there is a node in $\{x \mid d(u,x) \in [b\frac{\log n}{n},a\frac{\log n}{n}]\cup [0,c\frac{\log n}{n}] \}$ or there exists nodes $(v_1,v_2)$ in $\{x \mid d(u,x) \in [\frac{-c\log n}{n},\frac{-c/2\log n}{n}]\}$ and $\{x \mid d(u,x) \in  [\frac{b-c/2\log n}{n},\frac{(a-c)\log n}{n}]\}$. In the second case, $u$ is connected to $v_1$ and $v_1$ is connected to $v_2$. Therefore $u$ has nodes on left and right and $u$ is connected to both of them although not directly. Similarly $A_u^3$ and $A_u^4$ being zero implies that there exist nodes in $\{x \mid d(u,x) \in  [-a\frac{\log n}{n},-b\frac{\log n}{n}]\cup [-c\frac{\log n}{n},0]\}$ or again $u$ will have nodes on left and right and will be connected to them. So , when all the $4$ events happen together, the only exceptional case is when there are nodes in $\{x \mid d(u,x) \in [b\frac{\log n}{n},a\frac{\log n}{n}]\cup [0,c\frac{\log n}{n}] \}$ and $\{x \mid d(u,x) \in  [-a\frac{\log n}{n},b\frac{\log n}{n}]\cup [-c\frac{\log n}{n},0] \}$. But in that case $u$ has direct neighbors on both its left and right. So, we can conclude that for every node $u$, there exists a node $v$ such that $d(u,v) \in[0,\frac{a\log n}{n}]$ and a node $w$ such that $d(u,w) \in[\frac{-a\log n}{n},0]$ such that $u$ is connected to both $v$ and $w$. This implies that every node $u$ has neighbors on both its left and right  and therefore every node is part of a cycle that covers $[0,1]$. 
\end{proof}

\begin{lemma}\label{lem:disc}
A random geometric graph $G(n,\frac{a \log n}{n})$ will have $\omega(1)$ disconnected components for $a <1.$
\end{lemma}
\begin{proof}
Define an indicator random variable $A_u$ for a node $u$ which is $1$ if it does not have a neighbor on its left. We must have that 
\begin{align*}
\Pr(A_u)=\Big(1-\frac{a \log n}{n}\Big)^{n-1}.
\end{align*}
Therefore we must have that $\sum_u \avg A_u =n^{1-a}=\Omega(1)$ if $a<1$. This statement also holds true with high probability. To show this we need to prove that the variance of $\sum_u \avg A_u$ is bounded. We have that 
\begin{align*}
{\rm Var}(A) <\avg[A]+\sum_{u \neq v} {\rm Cov}(A_u,A_v)=\avg[A]+\sum_{u \neq v} \Pr(A_u=1 \cap A_v=1)-\Pr(A_u=1)\Pr(A_v=1)
\end{align*}
Now, consider the scenario when the  vertices $u$ and $v$ are at a distance more than $\frac{2a \log n}{n}$ apart (happens with probability at least $1-\frac{4a\log n}{n})$. Then the region in $[0,1]$ that is within distance $\frac{a\log n}{n}$ from both of the vertices is empty and therefore $\Pr(A_u=1 \cap A_v=1) = \Pr(A_u=1)\Pr(A_v=1|A_u=1) \leq  \Pr(A_u=1)\Pr(A_v=1) = (\Pr(A_u=1))^2$.  When the vertices  are within distance $\frac{2a \log n}{n}$ of one another, then
$
\Pr(A_u=1 \cap A_v=1) \le \Pr(A_u=1).
$
Therefore,
\begin{align*}
\Pr(A_u=1 \cap A_v=1) \le (1-\frac{4a\log n}{n}) (\Pr(A_u=1))^2 + \frac{4a\log n}{n}\Pr(A_u=1).
\end{align*}
Consequently,
\begin{align*}
\Pr(A_u=1 \cap A_v=1)-\Pr(A_u=1)\Pr(A_v=1) &\le (1-\frac{4a\log n}{n}) (\Pr(A_u=1))^2 \\+ \frac{4a\log n}{n}\Pr(A_u=1) -& (\Pr(A_u=1))^2 
\le  \frac{4a\log n}{n}\Pr(A_u=1).
\end{align*}
Now,
$$
{\rm Var}(A) \le \avg[A] + \binom{n}{2}\frac{4a\log n}{n}\Pr(A_u=1) \le \avg[A](1+ 2a\log n).
$$
By using Chebyshev bound, with probability at least $1-\frac{1}{\log n}$, 
$$A >n^{1-a}-\sqrt{n^{1-a}(1+2a\log n)\log n},$$

Now, observe that if there exists $k$ nodes which do not have a neighbor on one side, then there must exist $k-1$ disconnected components. Hence the number of disconnected components in $G(n,\frac{a \log n}{n})$ is $\omega(1)$.
\end{proof}

\end{document}